\documentclass[acmsmall, screen]{acmart}
\newcommand{\APPENDICITIS}[2]{#2}

\APPENDICITIS{}
{\settopmatter{printacmref=false}}

\AtBeginDocument{%
  \providecommand\BibTeX{{%
    \normalfont B\kern-0.5em{\scshape i\kern-0.25em b}\kern-0.8em\TeX}}}

\setcopyright{rightsretained}
\acmPrice{}
\acmDOI{10.1145/3428275}
\acmYear{2020}
\copyrightyear{2020}
\acmSubmissionID{oopsla20main-p384-p}
\acmJournal{PACMPL}
\acmVolume{4}
\acmNumber{OOPSLA}
\acmArticle{207}
\acmMonth{11}

\usepackage{microtype}

\usepackage{xcolor,xspace,stmaryrd,bbold,mathpartir,wrapfig,graphicx,enumitem}
\usepackage{listings}
\usepackage[xcolor]{changebar}
\usepackage{style/julia}

\renewcommand{\c}[1]{\lstinline[basicstyle={\color{jlstring}\ttfamily\small\selectfont}]{#1}\xspace}
\newcommand{\code}{\c}

\newcommand{\figref}[1]{Fig.~\ref{#1}}
\newcommand{\secref}[1]{Sec.~\ref{#1}}
\newcommand{\appref}[1]{App.~\ref{#1}}
\newcommand{\thmref}[1]{Theorem~\ref{#1}}
\newcommand{\lemref}[1]{Lemma~\ref{#1}}
\newcommand{\juliette}{\textsc{Juliette}\xspace}

\newcommand{\jlmultnum}{357\xspace}
\newcommand{\jladdnum}{166\xspace}
\newcommand{\eval}{\texttt{\small eval}\xspace}
\newcommand{\Eval}{\texttt{\small Eval}\xspace}

\newcommand{\invokelatest}{\texttt{\small invokelatest}\xspace}

\newcommand{\Alt}{~\vert~}

\newcommand{\subst}[2]{\ensuremath{[#1\!\mapsto\!#2]}}
\newcommand{\dom}{\ensuremath{\mathop{dom}}}

\renewcommand{\obar}[1]{\overline{#1}}

\newcommand{\eqtt}{\texttt{=}}
\newcommand{\HEY}[2]{{\footnotesize\textsc{#1-#2}}\xspace}
\newcommand{\WAE}[1]{\HEY E{#1}}
\newcommand{\WAT}[1]{\HEY{T}{#1}}
\newcommand{\WAOE}[1]{\HEY{OE}{#1}}
\newcommand{\WAOD}[1]{\HEY{OD}{#1}}
\newcommand{\WAOT}[1]{\HEY{OT}{#1}}
\newcommand{\WACAN}[1]{\HEY{CN}{#1}}
\newcommand{\err}{{\texttt{error}}\xspace}
\newcommand{\Primop}{\ensuremath{\Delta}\xspace}
\newcommand{\primop}[1]{\ensuremath{\delta_{#1}}\xspace}
\newcommand{\primopd}{\primop{l}}
\newcommand{\primcall}[2]{\ensuremath{\primop{#1}(#2)}\xspace}
\newcommand{\primcalld}[1]{\primcall{l}{#1}}
\newcommand{\PrimopRT}{\ensuremath{\Psi}\xspace} % return type
\renewcommand{\v}{\ensuremath{\texttt{v}}\xspace}
\newcommand{\vs}{\obar{\v}\xspace}
\newcommand{\skp}{\texttt{unit}\xspace}
\newcommand{\mname}{\ensuremath{\texttt{m}}\xspace}
\newcommand{\mval}[1]{{\texttt{#1}}\xspace}
\newcommand{\m}{\mval{m}}
\newcommand{\e}{\ensuremath{\texttt{e}}\xspace}
\newcommand{\es}{\ensuremath{\obar{\e}}\xspace}
\newcommand{\x}{\ensuremath{x}\xspace}
\newcommand{\xs}{\ensuremath{\obar{\x}}\xspace}
\newcommand{\seq}[2]{\ensuremath{#1\,;\,#2}}
\newcommand{\mcall}[2]{\ensuremath{\mathop{#1}(#2)}}
\newcommand{\mcalld}{\mcall{\m}{\vs}}
\newcommand{\md}{\ensuremath{\texttt{md}}\xspace}

\newcommand{\evalg}[1]{\ensuremath{{\color{violet}\llparenthesis\,{\color{black}#1}\,\rrparenthesis}}}
\newcommand{\evalt}[2]{\ensuremath{\evalg{#2}_{\color{violet}#1}}}
\newcommand{\mdef}[3]{\ensuremath{\triangleleft\, \mathop{#1}(#2)\eqtt\,#3 \,\triangleright}}
\newcommand{\mdefd}{\ensuremath{\mdef{\m}{\obar{\jty{\x}{\t}}}{\e}}}
\newcommand{\rslt}{\ensuremath{\texttt{r}}\xspace}
\newcommand{\p}{\ensuremath{\texttt{p}}\xspace}
\renewcommand{\t}{\ensuremath{\tau}\xspace}
\newcommand{\ts}{\ensuremath{\obar{\t}}\xspace}
\newcommand{\g}{\ensuremath{\sigma}\xspace}
\newcommand{\gs}{\ensuremath{\obar{\g}}\xspace}
\newcommand{\mty}{\ensuremath{\mathbb{f}_\m}\xspace}
\newcommand{\Unit}{\ensuremath{\mathbb{1}}\xspace}
\newcommand{\typeof}{\ensuremath{\mathop{\mathbf{typeof}}}\xspace}
\newcommand{\jsub}[2]{\ensuremath{#1 <: #2}}
\newcommand{\jty}[2]{\ensuremath{#1 :: #2}}
\newcommand{\MT}{\ensuremath{\mathrm{M}}\xspace}
\newcommand{\MTg}{\ensuremath{\MT_g}\xspace}
\newcommand{\MText}[2]{\ensuremath{#1 \bullet #2}}
\newcommand{\getmd}{\ensuremath{\mathop{\mathbf{getmd}}}\xspace}
\newcommand{\latest}{\ensuremath{\mathop{\mathbf{latest}}}}
\newcommand{\applcbl}{\ensuremath{\mathop{\mathrm{applicable}}}}
\newcommand{\containseq}{\ensuremath{\mathop{\mathrm{contains}}}}
\newcommand{\Cx}{\ensuremath{\texttt{C}}\xspace}
\newcommand{\Xx}{\ensuremath{\texttt{X}}\xspace}
\newcommand{\OEx}{\ensuremath{\texttt{E}}\xspace}
\newcommand{\hole}{\ensuremath{\square}\xspace}
\newcommand{\plugx}[2]{\ensuremath{#1\left[#2\right]}}
\newcommand{\plugCx}[2]{\ensuremath{{\color{violet} #1\left[{\color{black} #2}\right]}}}
\newcommand{\rep}{\ensuremath{\mathtt{rep}}\xspace}
\newcommand{\rV}{\mathop{\mathcal{V}}}
\newcommand{\rX}{\mathop{\mathcal{X}}}
\newcommand{\rC}{\mathop{\mathcal{C}}}
\newcommand{\rVd}{\ensuremath{\rV(\v)}\xspace}
\newcommand{\rXX}[1]{\ensuremath{\rX(#1,\ \mcalld)}\xspace}
\newcommand{\rXd}{\rXX{\Xx}}
\newcommand{\rCC}[1]{\ensuremath{\rC(#1,\ \rdx)}\xspace}
\newcommand{\rCd}{\rCC{\Cx}}
\newcommand{\CanSym}{\ensuremath{\mathop{\boldsymbol{Can}}}\xspace}
\newcommand{\Can}[2]{\ensuremath{\CanSym\left(#1\ \mathbf{as}\ #2\right)}\xspace}
\newcommand{\evalstep}{\ensuremath{\rightarrow}\xspace}
\newcommand{\evalstepc}{\ensuremath{\evalstep^*}\xspace}
\newcommand{\stwa}[2]{\ensuremath{\langle #1, #2 \rangle}}
\newcommand{\evalwast}[2]{\ensuremath{#1\ \evalstep\ #2}}
\newcommand{\evalwa}[4]{\ensuremath{\evalwast{\stwa{#1}{#2}}{\stwa{#3}{#4}}}}
\newcommand{\evalwad}[2]{\evalwa{\MT}{#1}{\MT}{#2}}
\newcommand{\evalwastc}[2]{\ensuremath{#1\ \evalstepc\ #2}}
\newcommand{\evalwac}[4]{\ensuremath{\evalwastc{\stwa{#1}{#2}}{\stwa{#3}{#4}}}}
\newcommand{\evalerrwa}[3]{\ensuremath{#1\ \vdash\ #2\,\evalstep\,#3}}
\newcommand{\evalerrwad}[2]{\evalerrwa{\MT}{#1}{#2}}
\newcommand{\Gm}{\ensuremath{\Gamma}\xspace}
\newcommand{\typedwa}[4]{\ensuremath{#2\ \vdash_{#1}\ #3\ :\ #4}}
\newcommand{\typedwad}[2]{\typedwa{}{\Gm}{#1}{#2}}
\newcommand{\gm}{\ensuremath{\gamma}\xspace}
\newcommand{\substok}[2]{\ensuremath{#1 \vdash #2}\xspace}
\newcommand{\substokd}{\substok{\Gm}{\gm}}
\newcommand{\nv}{\ensuremath{\nu}\xspace}
\newcommand{\nvs}{\ensuremath{\obar{\nv}}\xspace}
\newcommand{\mdeq}[3]{\ensuremath{#2(#1) \leadsto #3}\xspace}
\newcommand{\mdeqd}{\mdeq{\obar{\sigma}}{\m}{\m'}}
\newcommand{\SpecEnv}{\Phi}
\newcommand{\mspec}[4]{\ensuremath{\vdash^{#1}_{{\color{violet}#2}\leadsto{\color{violet}#3}} #4}\xspace}
\newcommand{\mspecd}{\mspec{\SpecEnv}{\MT}{\MT'}{\mdeqd}}
\newcommand{\expropt}[6]{\ensuremath{#2\ \vdash^{#1}_{{\color{violet}#3}\leadsto{\color{violet}#5}}\ #4 \leadsto #6}}
\newcommand{\exproptd}[2]{\expropt{\SpecEnv}{\Gm}{\MT}{#1}{\MT'}{#2}}
\newcommand{\expropte}[2]{\expropt{\SpecEnv}{}{\MT}{#1}{\MT'}{#2}}
\newcommand{\mdopt}[5]{\ensuremath{\vdash^{#1}_{{\color{violet}#2}\leadsto{\color{violet}#3}} #4 \leadsto #5}\xspace}
\newcommand{\mdoptd}[2]{\mdopt{\SpecEnv}{\MT}{\MT'}{#1}{#2}}
\newcommand{\tableopt}[4]{\ensuremath{\vdash^{#1} #3 \leadsto #4}}
\newcommand{\tableoptd}{\tableopt{\SpecEnv}{\e}{\MT}{\MT'}}
\newcommand{\tableoptref}[4]{\ensuremath{\vdash^{#1}_{#2} #3 \leadsto #4}}
\newcommand{\tableoptrefd}{\tableoptref{\SpecEnv}{\e}{\MT}{\MT'}}
\newcommand{\tableoptexpr}[3]{\ensuremath{\vdash_{{\color{violet}#1}\leadsto{\color{violet}#2}} #3}}
\newcommand{\tableoptexprd}[1]{\tableoptexpr{\MT}{\MT'}{#1}}
\newcommand{\rdx}{\ensuremath{\texttt{rdx}}\xspace}
\newcommand{\vnm}{\ensuremath{\texttt{v}^{\neq \m}}\xspace}

\begin{document}
\title{World Age in Julia} % TODO: ? or Optimization-Firendly Semantics for Dynamic Code Loading in Julia?
\subtitle{Optimizing Method Dispatch in the Presence of Eval\APPENDICITIS{}{ (Extended Version)}}

\author{Julia Belyakova}\affiliation{\institution{Northeastern University}\country{USA} }
\author{Benjamin Chung}\affiliation{\institution{Northeastern University}\country{USA} }
\author{Jack Gelinas}\affiliation{\institution{Northeastern University}\country{USA} }
\author{Jameson Nash}\affiliation{\institution{Julia Computing}\country{USA} }
\author{Ross Tate}\affiliation{\institution{Cornell University}\country{USA} }
\author{Jan Vitek}\affiliation{\institution{Northeastern University / Czech Technical University}\country{USA / Czech Republic} }

\renewcommand{\shortauthors}{Belyakova et al.}




\begin{abstract}
Dynamic programming languages face semantic and performance challenges in
the presence of features, such as \eval, that can inject new code into a
running program. The Julia programming language introduces the novel
concept of world age to insulate optimized code from one of the most
disruptive side-effects of \eval: changes to the
definition of an existing function. This paper provides the first formal
semantics of world age in a core calculus named \juliette, and shows how
world age enables compiler optimizations, such as inlining, in the presence
of \eval. While Julia also provides programmers with the means to bypass
world age, we found that this mechanism is not used extensively: a static
analysis of over 4,000 registered Julia packages shows that only
4--9\% of packages bypass world age.
This suggests that Julia's semantics aligns with programmer expectations.
\end{abstract}

\maketitle

\section{Introduction}\label{sec:intro}
%% =============================================================================

The Julia programming language~\cite{BezansonEKS17} aims to decrease the gap
between productivity and performance languages in scientific computing.
While Julia provides productivity features such as dynamic types, optional
type annotations, reflection, garbage collection, symmetric multiple dispatch,
and {dynamic code loading}, its designers carefully arranged those
features to allow for heavy compiler optimization. The key to performance lies
in the synergy between language design, language-implementation techniques, and
programming style~\cite{oopsla18a}.

The goal of this paper is to shed light on one particular design
challenge: how to support \eval that enables dynamic code loading,
\emph{and} achieve good performance. The \eval construct comes from
Lisp~\cite{lisp} and is found in most dynamic languages, but its expressive
power varies from one language to another.  Usually \eval takes a string or
a syntax tree as an argument and executes it in some environment.  In
JavaScript and R, \eval may execute in the current lexical environment; in
Lisp and Clojure, it is limited to the ``top level''.  On this spectrum,
Julia takes the latter approach, which enables compiler optimizations that
would otherwise be unsound. For example, in a program in
\figref{fig:eval-scope}, multiplication \c{x*2} in the body of function \c f
can be safely optimized to an efficient integer multiplication for the call
\c{f(42)}. This is because \eval only accesses the top-level environment and
thus cannot change the value of a local parameter \c x, which is known to be
the integer 42.  For a global \c x, such an optimization would be unsound.

\begin{figure*}
 \begin{minipage}{.27\textwidth}
  \begin{lstlisting}
> x = 3.14
> f(x) = (
    eval(:(x = 0));
    x * 2)

> f(42) # 84
> x     # 0
\end{lstlisting}\vspace{-2mm}
  \caption{Scope of \eval in Julia}\label{fig:eval-scope}
\end{minipage}
\begin{minipage}{.05\textwidth}
\hspace{1em}
\end{minipage}
\begin{minipage}{.31\textwidth}\begin{lstlisting}[morekeywords={defn}]
> (defn g [] 2)
> (defn f [x]
    (eval `(defn g [] ~x))
    (* x (g)))
> (f 42) ; 1764
> (g)    ; 42
> (f 42) ; 1764
\end{lstlisting}\vspace{-2mm}
    \caption{Eval in Clojure}\label{fig:clojure}
  \end{minipage}
\begin{minipage}{.31\textwidth}\begin{lstlisting}
> g() = 2
> f(x) = (
    eval(:(g()=$x));
    x * g())
> f(42) # 84
> g()   # 42
> f(42) # 1764
\end{lstlisting}\vspace{-2mm}
\caption{Eval in Julia}\label{fig:eval-methods}
\end{minipage}
\end{figure*}

What is unique about the design of \eval in Julia is the treatment of
function definitions. Many compilers rely on the information about functions
for optimizations, but those optimizations can be jeopardized by the presence
of \eval. To explain how Julia handles the interaction of \eval and functions,
we contrast it with the Clojure language. \figref{fig:clojure} shows a Clojure
program with a call to a function \c f which, within its body,
updates function \c g by invoking \eval.
Then, the call to \c f returns 1764 because the new definition of \c g is
used. \figref{fig:eval-methods} shows a Julia equivalent of the same program.
Here, the second call to \c f returns 1764 just like in Clojure,
but \emph{the first call returns 84}.
This is because, while the first invocation of \c f is running,
it does not see the redefinition of \c g made by \eval:
the redefinition becomes visible only after the first call (to \c{f(42)}) returns
to the top level. From the compiler's point of view,
this means that calling \eval does not force recompilation of any
methods that are ``in-flight.'' Thus, it is safe to devirtualize, specialize,
and inline functions in the presence of \eval without the need for deoptimization.
For example, \c{x*g()} can be safely replaced with \c{x*2} in \c f for
the first call \c{f(42)}.

\begin{wrapfigure}{r}{.48\textwidth}\vspace{-4mm}
  \begin{minipage}{.48\textwidth}\begin{lstlisting}
*(x::Int,     y::Int)     = mul_int(x, y)
*(x::Float64, y::Float64) = mul_float(x,y)
*(a::Number,  b::AbstractVector) = ...
*(x::Bool,    y::Bool)    = x & y
\end{lstlisting}\end{minipage}\vspace{-2mm}
\caption{Multiple definitions of \c *}\label{fig:mult}\vspace{-2mm}\end{wrapfigure}

Julia made the choice to restrict access to newly defined methods due to
pressing performance concerns. Julia heavily relies on symmetric multiple
dispatch~\cite{Bobrow86}, which allows a function to have multiple
implementations, called methods, distinguished by their parameter type
annotations. At run time, a call is dispatched to the most specific method
applicable to the types of its arguments. While some functions might have
only one method, plenty have dozens or even hundreds of them.  For example,
the multiplication function alone has \jlmultnum standard methods (see an excerpt
in \figref{fig:mult}). If Julia were to always use generic method invocation to
dispatch \c *, programs would become unbearably slow. By constraining \eval,
the compiler can avoid generic invocations. In \figref{fig:eval-scope}, the compiler can
pick and inline the right definition of \c * when compiling \c{f(42)}.
It should be noted that this optimization-friendly \eval semantics does not
apply to data. Function definitions are treated differently from
variables, and changes to global variables (such as in \figref{fig:eval-scope})
can be observed immediately.

Arguably, despite being unusual, Julia's semantics is easy to understand for
programmers. There is always a clear point where new definitions become
visible---at the top level---and thus, users can avoid surprises and dependence
on the exact position of \eval in the code.
However, in case the default semantics is not desirable,
Julia also provides an escape hatch: the built-in function
\c{invokelatest(f)}, which forces the implementation to invoke the
most recent definition of \c f. A slower alternative to \c{invokelatest}
is to call \c f within \eval, which always executes in the top level.

This language mechanism that delays the effect of \eval on function
definitions is called \emph{world age}. In the Julia documentation,
world age is described operationally~\cite{oopsla18a}: every method
defined in a program is associated with an age, and for each function call,
Julia ensures that the current age is larger than the age of the
method about to be invoked. One can think of the world age as a
counter that allows the implementation to ignore all methods that were
born after the last top-level call started. Much of its specification
is tied to implementation details and efficiency considerations.  Our
contributions are as follows:
\begin{itemize}
\item \emph{A core calculus for world age:} We introduce \juliette, a
  calculus that models the notion of world age abstractly. In the calculus,
  the implementation-oriented world-age counters are replaced with
  method tables that are explicitly copied at the top level,
  and \eval is simplified down to an operation that
  evaluates its argument in a specific method table.
\item \emph{Formalization of optimizations:} We formalize and prove
  correct three compiler optimizations, namely inlining,
  devirtualization, and specialization.
\item \emph{Corpus analysis:} We analyze Julia packages to understand
  how \eval is used, and estimate the potential impact of world age on
  library code. We also identify a number of programming
  patterns by manual inspection of selected packages.
\item \emph{Testing the semantics:} We develop a Redex model of our
  calculus and optimizations to allow rapid experimentation and testing.
\end{itemize}

\noindent
The corpus analysis and the Redex model are
publicly available.\footnote{\url{https://github.com/julbinb/juliette-wa}}
\APPENDICITIS{The formalization with detailed proofs can be found
in the extended version of the paper~\cite{belyakova2020world}.}{}

\section{Background}
%% =============================================================================

We start with an overview of the features of Julia relevant to our
work, and review related work.

\subsection{Julia Overview}\label{sec:julia} %% ==============================

Despite the extensive use of types and type annotations for dispatch
and compiler optimizations, Julia is not statically typed.
A formalization of types and subtyping is provided by~\citet{oopsla18b},
and a general introduction to the language is given by~\citet{BezansonEKS17}.

\paragraph{Values}
Values are either instances of \emph{primitive types}---sequences of bits---or
\emph{composite types}---collection of fields holding values. Every value has
a concrete type (or tag). This tag is either inferred statically or stored
in the boxed value.  Tags are used to resolve multiple dispatch semantically
and can be queried with \typeof.

\paragraph{Types} Programmers
can declare three kinds of user-defined types: \emph{abstract types},
\emph{primitive types}, and \emph{composite types}. Abstract types cannot be
instantiated, while concrete types can. For example, \c{Float64} is
concrete, and is a subtype of abstract type \c{Number}.  Concrete types
have no subtypes. Additionally, user-defined type constructors can have bounded
type parameters and can declare up to a single supertype.

\paragraph{Annotations}
Type annotations include a number of built-in type constructors, such as
union and tuple types.  Tuple types, written \c{Tuple\{A,...\}}, describe
immutable values that have a special role in the language: every method
takes a single tuple argument.  The \c{::} operator ascribes a type to a
definition.  We will use \t to denote annotations.

\paragraph{Subtyping}
The subtyping relation, \c{<:}, is used in run-time casts and multiple
dispatch. Julia combines nominal subtyping, union types, iterated union
types, covariant and invariant constructors, and singleton types.  Tuple
types are covariant in their parameters, so, for instance,
\c{Tuple\{Float64,Float64\}} is a subtype of \c{Tuple\{Number, Number\}}.

\paragraph{Multiple dispatch}
A function can have multiple methods where each method declares what
argument types it can handle; an unspecified type defaults to \c{Any}.
At run time, dispatching a call \c{f(v)} amounts to picking the best
applicable method from all the methods of function \c f.  For this,
the dispatch mechanism first filters out methods whose type
annotations \t are not a supertype of the type tag of \c v.  Then it
takes the method whose type annotation $\t_i$ is the most specific of
the remaining ones.  If the set of applicable methods is empty, or
there is no single best method, a run-time error is raised.

\paragraph{Reflection}

Julia provides a number of built-in functions for run-time introspection
and meta-programming. For
instance, the methods of any function \c f may be listed using
\c{methods(f)}. All the methods are stored in a special data structure,
called the \emph{method table}.  It is possible to search the method table
for methods accepting a given type: for instance, \c{methods(*, (Int,
  Float64))} will show methods of \c * that accept an integer-float pair.
The \eval function takes an expression object and evaluates it in the global
environment of a specified module. For example, \c{eval(:(1+2))} will take
the expression \c{:(1+2)} and return \c3.

\begin{wrapfigure}{r}{.45\textwidth}\vspace{-6mm}
  \begin{minipage}{.45\textwidth}\begin{lstlisting}
for op in (:+, :*, :&, :|)
  eval(:($op(a,b,c) = $op($op(a,b),c)))
end
\end{lstlisting}\end{minipage}\vspace{-3mm}
\caption{Code generation}\label{fig:eval-for}\vspace{-2mm}\end{wrapfigure}

\noindent
\Eval is frequently used for meta-programming as part of code generation.
For example, \figref{fig:eval-for} generalizes some of
the basic binary operators to three arguments, generating four new methods.
Instead of building expressions explicitly, one can also invoke the parser on a string.
For instance, \c{eval(Meta.parse("id(x) = x"))} creates an identity method.

\subsection{Related Work} %====================================================

This paper is concerned with controlling the visibility of function
definitions. Most programming languages control \emph{where} definitions are
visible, as part of their scoping mechanisms. Controlling \emph{when}
function definitions become visible is less common.

Languages with an interactive development environment had to deal with the
addition of new definitions for functions from the
start~\cite{lisp}. Originally, these languages were interpreted. In that
setting, allowing new functions to become visible immediately was both easy
to implement and did not incur any performance overhead.

Just-in-time compilation changed the performance landscape, allowing dynamic
languages to have competitive performance. However, this meant that to generate
efficient code, compilers had to commit to particular versions of functions.
If any function is redefined, all code that depends on that function must
be recompiled; furthermore, any function currently executing has to be
deoptimized using mechanisms such as on-stack-replacement~\cite{Hoelzle92}.
The drawback of deoptimization is that it makes the compiler more complex
and hinders some optimizations. For example, a special \c{assume} instruction
is introduced as a barrier to optimizations by \citet{popl18}, who formalized
the speculation and deoptimization happening in a model compiler.

Java allows for dynamic loading of new classes and provides
sophisticated controls for where those classes are visible.
This is done by the class-loading framework that is
part of the virtual machine~\cite{LB98}.  Much research happened in that
context to allow the Java compiler to optimize code in the presence of dynamic
loading. \citet{detlefs99} describe a technique, which they call preexistence,
that can devirtualize a method call when the receiver object predates the
introduction of a new class.  Further research looked at performing
dependency analysis to identify which methods are affected by the newly
added definitions, to be then recompiled on
demand~\cite{nguyen2005interprocedural}. \citet{glew2005method}
describes a type-safe means of inlining and devirtualization: when newly
loaded code is reachable from previously optimized code, these optimizations
must be rechecked.

Controlling \emph{when} definitions take effect is important in
dynamic software updating, where running systems are updated with new
code~\cite{lee}. \citet{hicks} introduce a calculus for reasoning about
representation-consistent dynamic software updating in C-like languages.
One of the key elements for their result is the presence of an \c{update}
instruction that specifies when an update is allowed to happen. This has
similarities to the world-age mechanism described here.

Substantial amounts of effort have been put into building calculi that
support \eval and similar constructs. For example, \citet{politz12}
described the ECMAScript 5.1 semantics for \eval, among other
features. \citet{glew2005method} formalized dynamic class loading in the
framework of Featherweight Java, and \citet{matthews2008operational}
developed a calculus for \eval in Scheme.  These works formalize the
semantics of dynamically modifiable code in their respective languages, but,
unlike Julia, the languages formalized do not have features explicitly
designed to support efficient implementation.

\section{World age in Julia}\label{sec:wa-julia}
%% =============================================================================

The world-age mechanism in Julia limits the set of methods that can be
invoked from a given call site. World age fixes the set of method
definitions reachable from the currently executing method, isolating
it from dynamically generated ones. In turn, this allows the compiler
to optimize code without need for deoptimization, and limits the
number of required synchronization points in a multi-threaded program.
If full access to methods is required, however, Julia provides escape
hatches to bypass world age by sacrificing performance.

\subsection{Defining World Age}

The primary goal of the world-age mechanism is to align the language's
semantics with the assumptions made by the Julia just-in-time compiler's
optimizations. Semantically, newly added methods (i.e. ones defined 
using \eval) only become
visible when execution returns to the top level, and the set of callable methods
for an execution is fixed when it leaves the top level. Compilation of
methods is triggered---only at the top level---when one of the following holds:
(1) a function is called with previously unobserved types of arguments, or (2)
a previously compiled function needs to be recompiled due to a change in its
own definition or  one of its dependencies. Since the set of visible methods
gets fixed at a top-level call, and compilation only occurs from the top level,
the compiler may assume that the currently known set of methods is complete
and can optimize accordingly.

% ed: this sentence's semantic meaning was lost in the rewrite; should it be 
% reinserted?
% This ``lazy'' compilation means that methods defined within \eval will not be
% compiled until the execution returns to the top level, and thus executing an
% \eval-defined method is initially slow. 

% The goal of the world age is to limit methods that can be invoked from any
% call site to those that are known ``statically''
% (here, defined as ``whenever the compiler runs'').
% In Julia, compilation happens whenever a function is called with
% a type signature that has not been observed before.
% Then, at every top-level statement, the compiler saves
% all of the current method definitions into the known set.
% This process is repeated each time the execution flow returns to the top level.
% As a result, if a recompilation is needed, it can be identified
% at the top-level by comparing the
% known set at original compilation time to that at the new call time. Thus,
% all recompilation occurs at top level calls, rather than potentially being
% buried deep inside a call stack.

\begin{figure}[!bt]
  \begin{minipage}{7cm}
    \begin{minipage}{3cm}
\begin{lstlisting}[linewidth=2.6cm]
f() = ... # t1
g() = ... # t2
g() = ... # t3
\end{lstlisting}
\hspace{-0.1cm}\(
\begin{array}{rcl}
\texttt{f}&\mapsto&[\texttt{t}_1,\infty]\\
\texttt{g}_1&\mapsto&[\texttt{t}_2,\texttt{t}_3]\\
\texttt{g}_2&\mapsto&[\texttt{t}_3,\infty]\\
\end{array}
\)
\end{minipage}
\begin{minipage}{3.8cm}
    \includegraphics{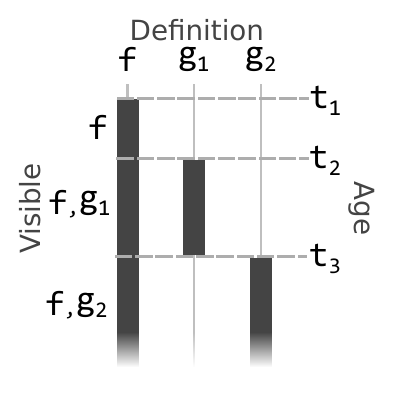}
\end{minipage}
\caption{Age ranges}\label{defage}
  \end{minipage}
  \hspace{5mm}
  \begin{minipage}{6cm}
\begin{lstlisting}[linewidth=6.5cm]
> function ntl()
    eval(:(f() = 1))
    f()
  end
> ntl()
ERROR: MethodError: no method matching f()
The applicable method may be too new:
running in world age 26806, while current
world is 26807. Closest candidates are:
  f() at REPL[10]:2
\end{lstlisting}
\vspace{-3mm}\caption{Age error}\label{fig:toonew}
\end{minipage}

\end{figure}

For performance reasons, the world-age mechanism is implemented by a
simple monotonic counter. The counter is incremented every time a method is
defined, and its value becomes the method's ``birth age''.
Every method also can a store a ``death age'' (that is initially infinity), which
is set when it is replaced or deleted. Methods with their
ages are stored in a global data structure called a method table.
The birth and death ages of a method determine the minimum and maximum world age
from which the method can be invoked.
This is illustrated in \figref{defage}. Here, we define functions \c f and \c g,
where \c f has only one method, whereas \c g has two methods, one replacing the other.
Let us observe what definitions are visible at each step.
Since \c f is defined once at time $\texttt{t}_1$ (i.e. when the world-age
counter is equal to $\texttt{t}_1$) and
never redefined, it has a birth age of $\texttt{t}_1$ and a death age of
$\infty$; thus, it can be used anytime after $\texttt{t}_1$ and
forevermore. In contrast, method $\texttt{g}_1$, created at time $\texttt{t}_2$,
is redefined at time $\texttt{t}_3$, so its birth age is $\texttt{t}_2$ and
death age is~$\texttt{t}_3$; therefore, $\texttt{g}_1$ can be called from
world age $\texttt{t}$ where $\texttt{t}_2 \leq \texttt{t} \leq \texttt{t}_3$.
Finally, $\texttt{g}_2$ is never redefined,
so its age ranges from $\texttt{t}_3$ to $\infty$.
% For performance reasons, the world age mechanism is implemented by a
% simple monotonic counter. The counter is incremented every time a method is
% defined, thereby giving each method an age. Each method, saved in a global
% table of methods, then remembers its birth age, defining the minimum world
% age from which it can be accessed. In order to support method overwriting
% and deletion, each method also stores its death age, i.e. the maximum age for
% invocations. Thus, a call site can only jump to a method after its birth and
% before its death. This is illustrated in \figref{defage}. Here, we define
% methods \c f and \c g, where \c g's definition changes from $\texttt{g}_1$
% to $\texttt{g}_2$ part-way through and observe what definitions are visible
% at each step. Since \c f is defined once at time $\texttt{t}_1$ and
% never redefined, it has a minimum age of $\texttt{t}_1$ and a maximum of
% $\infty$; thus, it can be used anytime after $\texttt{t}_1$ and
% forevermore. In contrast, the $\texttt{g}_1$ created at time $\texttt{t}_2$ is
% replaced at $\texttt{t}_3$, so its minimum age is $\texttt{t}_2$ and maximum
% age is $\texttt{t}_3$.  Finally, $\texttt{g}_2$ is never overwritten, so its
% age ranges from $\texttt{t}_3$ to $\infty$.

This language design can be restrictive. The easiest way to run afoul of world
age is to attempt to define a method with \eval and call it immediately thereafter. In \figref{fig:toonew},
function \c{ntl} creates a new function \c f using \eval and attempts to
call it immediately without returning to the top level.
Let us dissect the example and its error message. Function \c{ntl} was
invoked from the top level with age of \c{26806}, thus limiting the set
of visible methods to the ones born by this time. Then, \c{ntl} used \eval to
define \c f, giving it a birth age of \c{26807}. Finally, to call \c f, \c{ntl}
needs a method of \c f that was born by \c{26806}, but none exists.
Since the only method of \c f was created at \c{26807}, a \c{MethodError} is
raised indicating that no method was found.

\subsection{Breaking the Age Barrier}

There are situations in which the world-age mechanism is too restrictive:
for example, when a program wishes to programmatically generate code and then
use it immediately.
To accommodate these circumstances, Julia provides two ways for programmers to
execute code ahead of its birth age.
The first is \eval itself, which executes its arguments as if they were
at the top level, thus allowing any existing method to be called.
However, \eval needs to interpret arbitrary ASTs and is rather slow.
Luckily, in many circumstances, the program only wishes to bypass the world-age
restriction for a single function call. For this, one can use
\invokelatest, a built-in function that calls its argument from the latest
world age. While substantially slower than a normal call, \invokelatest is faster
than \eval. Moreover, \invokelatest is passed arguments directly from the
calling context so that values do not need to be inserted into \eval's AST.
Both \eval and \invokelatest can be used to amend the example shown in
\figref{fig:toonew} to call \c{f} in the latest world age. If we
replace the bare call to \c{f} with a call to \c{eval(:(f()))}, then the
call to \c{ntl} will produce 1. Similarly, \c{Base.invokelatest(f)} will get
the same result.

Using either of these mechanisms, programmers can opt out from the
limitation imposed by world age, but this comes with performance implications.
Since neither \invokelatest nor \eval can be optimized, and both can kick off
additional JIT compilation, they can have substantial performance impact.
However, this impact is limited to only these explicitly impacted call
sites. As a result, programmers can carefully design their programs
to minimize the number of broken barriers, thus minimizing the
performance impact of the dynamism.

\subsection{World Age in Practice}\label{sec:analysis}

We have argued that world age is useful for performance of Julia
programs, but does its semantics match programmers' expectations? We propose
 a slightly indirect
answer to this, analyzing a corpus of programs and observing how often they
use the two escape hatches mentioned above. Of those, \invokelatest is the clearest
indicator, as there is no reason to use it except to bypass the world
age. Similarly, \eval'ing a function call means evaluating that call in the
latest world age, thereby allowing it to see the latest method definitions.
The \eval indicator is imprecise, however, as there are uses of \eval that
are not impacted by world age.

We take as our corpus all 4,011 registered Julia packages as of August 2020.
The results of statically analyzing the code base are shown in \figref{fig:wapkgusage}.
The analysis shows that 2,846 of the 4,011 packages used neither \eval
nor \invokelatest, and thus are definitely age agnostic. Of the remaining
packages, 1,094 used \eval only, and so \emph{could} be impacted.
15 packages used \invokelatest only, and some 56 used both.
We can reasonably presume that at least these latter 71 packages are
impacted by world age because they bypass it using \invokelatest.
Drawing conclusions about \eval-using packages requires further analysis.

\begin{figure}[!h]
  \begin{minipage}{0.45\textwidth}
  \centering
  \includegraphics[scale=0.5]{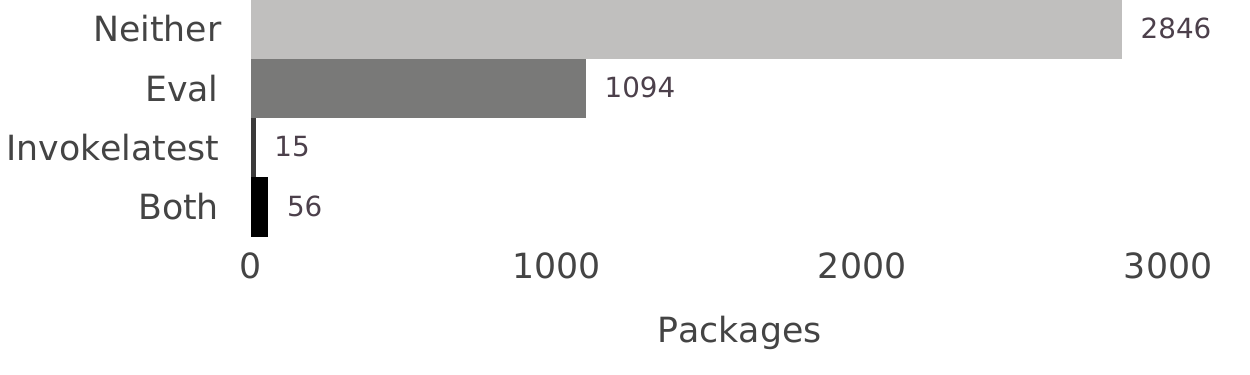}
  \vspace{-3mm}
  \caption{\eval and \invokelatest use by package}
  \label{fig:wapkgusage}
  \end{minipage}
  \begin{minipage}{0.45\textwidth}
  \includegraphics[scale=0.6]{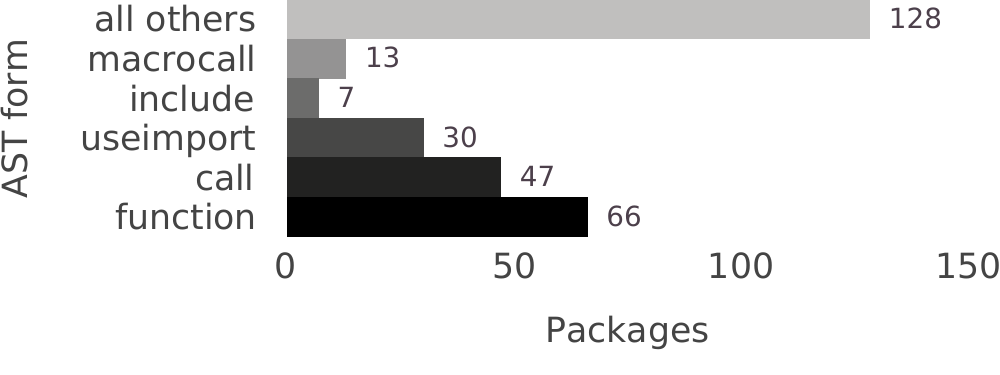}
  \vspace{-3mm}
  \caption{Static use of AST forms in all packages}
  \label{fig:all_static_ast}
  \end{minipage}
\end{figure}

To understand if packages that only use \eval are impacted by world
age, we statically analyzed the location of calls to \eval and their
arguments by parsing files that contain \eval. For each call, we 
classify the argument ASTs, recursively traversing them and counting
occurrences of relevant nodes. The analysis is conservative: it assumes
that an AST that is not statically obvious (such as a variable) could
contain anything.
%To get a better sense of how programs use \eval, we looked at the
%distribution of AST forms in \eval arguments,
\figref{fig:all_static_ast} 
shows how many packages use world-age relevant AST forms. Only uses of
\eval from within functions---where world-age could be relevant---are shown;
top-level uses of \eval---which cannot be affected by world age---are filtered out.
The ``all-others'' category encompasses all AST forms not relevant to world age.
While this aggregate is, taken as a whole, more common than any other single
AST form, none of the constituent AST forms is more prevalent than
function calls. Therefore, most common arguments to \eval are function definitions,
followed by function calls and loading of modules and other files.

Using the results of the static analysis,
we estimate that about 4--9\% of the 4,011 packages might be affected
by world age. The upper bound (360 packages) is a conservative estimate,
which includes 289 packages with potentially world-age-related calls
to \eval but without calls to \invokelatest, and the 71 packages
that use \invokelatest.
The lower bound (186 packages) includes 115 \invokelatest-free packages that
call \eval with both function definitions and function calls,
and the 71 packages with \invokelatest. 

% This data indicates that world age is relevant for at least 115 of the 289
% packages, identified as those that both define a function and call a method
% within \eval. Thus, overall, when we also include \invokelatest, at least 186
% packages out of 4011 (4.6\%) registered Julia packages might be affected by
% world age. 
%This data excludes several known macro calls that do not impact world age % (e.g. \c{@inline} or \c{@printf}).

To validate our static results, we dynamically analyzed 32 packages
out of the~186 identified as possibly affected by world age. These packages
were selected by randomly sampling a subset of 49 packages, which 
was then further reduced by removing packages that did not run, whose
tests failed, or that did not call \eval or \invokelatest at least once.
Over this corpus, the dynamic analysis was implemented by adding
instrumentation to record calls to \eval and \invokelatest, recording the ASTs
and functions, respectively, as well as the stack traces for each invocation.
% We then attempted to run the test suites for the chosen packages with the
% instrumentation.
% Our final data set contains 32 packages whose tests passed and
% which made at least one call to \eval or \invokelatest. 
%; the list of packages is in appendix.

% filter figure 11 down to the forms used in figure 12 and lay them out side by side
% delete figure 9
% paragraphs describing the filtering from 1094 to 289 (how we got to those affected by world age statically) to dynamic corpus

\begin{figure}[!ht]
  \begin{minipage}{0.49\textwidth}
  \includegraphics[scale=0.5]{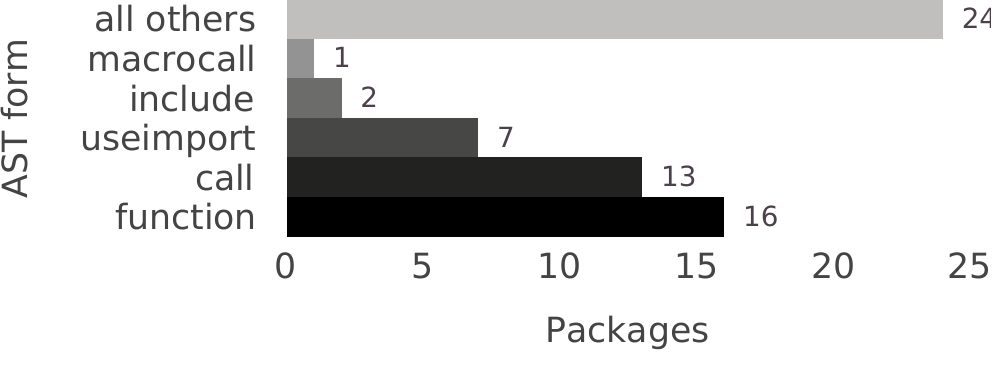}
    \vspace{-3mm}
  \caption{Static use of AST forms used by package}
  \label{fig:stasts}
  \end{minipage}
  \begin{minipage}{0.49\textwidth}
  \includegraphics[scale=0.5]{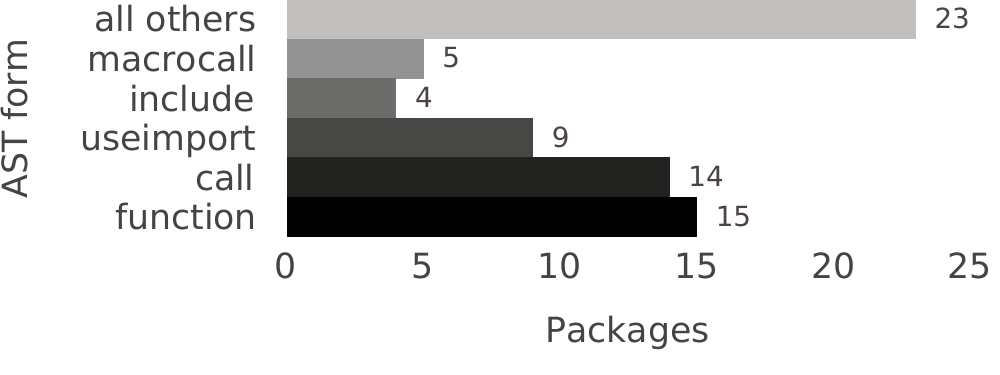}
    \vspace{-3mm}
  \caption{Dynamic use of AST forms used by package}
  \label{fig:dynasts}
  \end{minipage}
  \vspace{-1mm}
\end{figure}

The results of the static and dynamic analysis of the 32 packages
are given in \figref{fig:stasts} and \figref{fig:dynasts}, respectively.
%\figref{fig:dynasts} shows the number of packages that dynamically used each
%of the AST forms at least once, while \figref{fig:stasts} gives the static
%analysis results for the same packages.
Both analysis methods agreed that the
most common world-age-relevant use of \eval was to define functions, followed
by making function calls and importing other packages. In general, the dynamic
analysis was able to identify more packages that used each AST form, as it can
examine every AST ran through \eval, not only statically declared ones. However,
this accuracy is dependent on test coverage.

%% This dynamic analysis is an incomplete evaluation of our static
%% analysis, since we could not execute the test suites of some packages
%% (which might cause the AST usage to be different). Additionally, our
%% selection criteria (based on statically identified uses of \eval that
%% are within function bodies and whose arguments are computed) may cause
%% bias towards packages that exhibit broadly comparable behavior between
%% statically and dynamically identified AST forms.

%% With these caveats, however, the static and dynamic analyses
%% broadly agree that function definitions and invocations are very
%% common usages of \eval. In combination with \invokelatest (used in
%% 19 of the 32 executed packages), these features indicate that world
%% age may be semantically important in practice.

\subsection{Programming with World Age}

We now turn to common patterns found by manual inspection of select packages
of the corpus.

\paragraph{Boilerplating.}

The most common use of \eval is to automatically generate code for
boilerplate functions. These generated functions are typically created at the top-level so that
they can be used by the rest of the program. Consider the
\c{DualNumbers.jl} package, which provides a common dual number representation
for automatic differentiation. A dual number, which is a pair of the normal value and an ``epsilon'', which represents the derivative of the value, should support the same
operations as any number does and mostly defers to the standard operations.
For example, the \c{real} function, which gets the real component of a number
when applied to a dual number should recurse into both the actual and epsilon
value. \Eval can generate all of the needed implementations at package
load time (\c{@eval} is a macro that passes its argument to \eval as an AST).
\begin{lstlisting}
  for op in (:real, :imag, :conj, :float, :complex)
     @eval Base.$op(z::Dual) = Dual($op(value(z)), $op(epsilon(z)))
  end
\end{lstlisting}
A common sub-pattern is to generate proxies for interfaces defined by an
external system. For this purpose, the \c{CxxWrap.jl} library uses \eval at the
top level to generate (with the aid of a helper method that generates the ASTs)
proxies for arbitrary C++ libraries.
\begin{lstlisting}
  eval(build_function_expression(func, funcidx, julia_mod))
\end{lstlisting}

\paragraph{Defensive callbacks.}

The most widely used pattern for \invokelatest deals with function values of
unknown age.  For example, when invoking a callback provided by a client, a
library may protect itself against the case where the provided function was
defined after the library was loaded. There are two forms of this pattern.
The simplest uses \invokelatest for all callbacks, such as the library
\c{Symata.jl}:
\begin{lstlisting}
  for hook in preexecute_hooks
    invokelatest(hook)
  end
\end{lstlisting}
Every hook in \c{preexecute_hooks} is protected against world-age errors (at
the cost of slower function calls).  To avoid this slowdown, the second
common pattern catches world-age exceptions and falls back to \invokelatest such as in
from the \c{Genie.jl} web server:
\begin{lstlisting}
  fr::String = try
    f()::String
  catch
    Base.invokelatest(f)::String
  end
\end{lstlisting}
This may cause surprises, however. If a sufficiently old method exists, the
call may succeed but invoke the wrong method.\footnote{In Julia,
  higher-order functions are passed by name as generic functions,
  so a callback will be subject to multiple dispatch.}
  This pattern may also catch unwanted exceptions and
execute \c f twice, including its side-effects.

\paragraph{Domain-specific generation}

As a language targeting scientific computing, Julia has a large number of
packages that do various symbolic domain reasoning.  Examples include
symbolic math libraries, such as \c{Symata} and \c{GAP}, which have the
functionality to generate executable code for symbolic expressions.
\c{Symata} provides the following method to convert an internal expression
(a \c{Mxpr}) into a callable function. Here, \c{Symata} uses a translation
function \c{mxpr_to_expr} to convert the \c{Symata} \c{mxpr} into a Julia
\c{Expr}, then wraps it in a function definition (written using explicit AST
forms), before passing it to \eval.

\begin{lstlisting}
  function Compile(a::Mxpr{:List}, body)
    aux = MtoECompile()
    jexpr = Expr(:function,
                 Expr(:tuple, [mxpr_to_expr(x, aux) for x in margs(a)]...),
                 mxpr_to_expr(body, aux))
    Core.eval(Main, jexpr)
  end
\end{lstlisting}

\paragraph{Bottleneck}

Generated code is commonly used in Julia as a way to mediate between a high-level 
DSL and a numerical library. Compilation from the DSL to executable
code can dramatically improve efficiency while still retaining a high-level
representation. However, functions generated thusly cannot be called from the
code that generated them, since they are too new. Furthermore, this code is expected
to be high-performance, so using \invokelatest for every call is not acceptable. The
bottleneck pattern overcomes these issues. The idea is to split the program
into two parts: one that generates code, and another that runs it. The two
parts are bridged with a single \invokelatest call (the ``bottleneck''),
allowing the second part to call the generated code efficiently. The pattern
is used in the \c{DiffEqBase} library, part of the DifferentialEquations family
of libraries that provides numerical differential equation solvers.

\begin{lstlisting}
  if hasfield(typeof(_prob),:f) && hasfield(typeof(_prob.f),:f) &&
       typeof(_prob.f.f) <: EvalFunc
    Base.invokelatest(__solve,_prob,args...; kwargs...)
  else
    __solve(_prob,args...;kwargs...)
  end
\end{lstlisting}
Here, if \c{_prob} has a field \c f, which has another field \c f, and the
type of said inner-inner \c f is an \c{EvalFunc} (an internally-defined
wrapper around any function that was generated with \eval), then it will
invoke the \c{__solve} function using \invokelatest, thus allowing \c{__solve}
to call said method. Otherwise, it will do the invocation normally.

\paragraph{Superfluous eval}

This is a rare anti-pattern, probably indicating
a misunderstanding of world age by some Julia programmers. For example,
\code{Alpine.jl} package has the following call to \eval:
\begin{lstlisting}
  if isa(m.disc_var_pick, Function)
    eval(m.disc_var_pick)(m)
\end{lstlisting}
Here, \code{eval(m.disc_var_pick)} does nothing useful but imposes a performance
overhead. Because \c{m.disc_var_pick} is already a function value, calling
\eval on it is similar to using \c{eval(42)} instead of \c{42} directly;
this neither bypasses the world age nor even interprets an AST.

\paragraph{Name-based dispatch}
Another anti-pattern uses \eval to convert function names to functions.
For example, \c{ClassImbalance.jl} package chooses a function to call,
using its uninterpreted name:
\begin{lstlisting}
  func = (labeltype == :majority) ? :argmax : :argmin
  indx = eval(func)(counts)
\end{lstlisting}
It would be more efficient to operate with function values directly,
i.e. \c{func = ... : argmin} and then call it with \c{func(counts)}.
Similarly, when a symbol being looked up is generated dynamically, as it is
in the following example from \c{TextAnalysis.jl}, the use of \eval could be
avoided.
\begin{lstlisting}
  newscheme = uppercase(newscheme)
  if !in(newscheme, available_schemes) ...
  newscheme = eval(Symbol(newscheme))()
\end{lstlisting}
This pattern could be replaced with a call \c{getfield(TextAnalysis, Symbol(newscheme))},
where \c{getfield} is a special built-in function that finds a value in the
environment by its name. Using \c{getfield} would be more efficient than \eval.

\section{Juliette, a world age calculus}\label{sec:wa-formal}
%% =============================================================================

To formally study world age, we propose a core calculus, named \juliette, that
captures the essence of Julia's semantics and permits us to reason about
the correctness of some of the optimizations performed by the compiler.

Designing such a calculus is always an exercise in parsimony, balancing the
need to highlight principles while avoiding entanglements with particular
implementation choices. The first decision to grapple with is how to
represent world age. While efficient, counters are also pervasive and cause
confusion.\footnote{Although Julia's documentation attempts to explain world
  age~\cite{bib:juliadoc-methods}, questions such as
  \href{discourse.julialang.org/t/world-age-problem-explanation/9714}{this
    one} pop up periodically.} Furthermore, they obscure reasoning about
program-state equivalence; two programs with different initial counter
values could, if care is not taken, appear different. Dispensing with the counters used by Julia's compiler
is appealing.

An alternative that we chose is a more abstract representation of world age, one
that captures its intent: control over method visibility. \juliette
uses \emph{method tables} to represent sets of methods available for
dispatch. The \emph{global table} is the method table that records all
definitions and always reflects the ``true age'' of the world; the global table is
part of \juliette program state. \emph{Local tables} are method tables used
to resolve method dispatch during execution and may lag behind the global
table when new functions are introduced. Local tables are then baked into program syntax
to make them explicit during execution. As in Julia,
\juliette separates method tables (which represent code) from data: as mentioned
in \secref{sec:intro}, the world-age semantics only applies to code.
As global variables interact with \eval in the standard way,
we omit them from the calculus.

The treatment of methods is similar in both \juliette and Julia up to (lexically) local method definitions.
In both systems, a generic function is defined by the set of methods with the
same name. In Julia, local methods are syntactic sugar for
global methods with fresh names. For simplicity, we do not model this
aspect of Julia:  \juliette methods are always added to the global method table.
All function calls are resolved using the set of methods found in the
current local table. A function value \m denotes the name of a function and
is not itself a method definition. Then, since \juliette omits global variables,
its global environment is entirely captured by the global method table.

Although in Julia \eval incorporates two features---top-level evaluation and
quotation\footnote{Represented with the \c{\$} operator in Julia, as in
\texttt{eval(:(g() = \$x))} in \figref{fig:eval-methods}.}---only
top-level evaluation is relevant to world age, and this is what we model
in \juliette. Instead of an \eval construct, the calculus has
operations for evaluating expressions in different method-table contexts.
In particular, \juliette
offers a \emph{global evaluation construct} \evalg{\e} (pronounced ``banana
brackets'') that accesses the most recent set of methods. This is equivalent
to \eval's behavior, which evaluates in the latest world age.
Since \juliette does not have global variables, \evalg{\e} reads
from the local environment directly instead of using quotation.

Every function call \mcall{\m}{\vs} in \juliette gets resolved in the closest
enclosing local method table \MT by using an \emph{evaluation-in-a-table}
construct \evalt\MT{\mcall\m\vs}.
Any top-level function call first takes a snippet of the current global table
and then evaluates the call in that \emph{frozen} snippet.
That is, \evalg{\mcall\m\vs} steps
to \evalt\MT{\mcall\m\vs} where \MT is the current global table. Thus, once
a snippet of the global table becomes local table,
all inner function calls of \mcall{\m}{\vs}
will be resolved using this table, reflecting the fact that a currently
executing top-level function call does not see updates to the global table.

To focus on world age, \juliette omits irrelevant features such as loops or
mutable variables. Furthermore, the calculus is parameterized over values,
types, type annotations, a subtyping relation, and primitive operations.
For the purposes of this paper, only minimal assumptions are needed about those.

\subsection{Syntax} %% --------------------------------------------------

\begin{figure}
\[\footnotesize
\begin{array}{ccl@{\qquad}l}
    \\ \e & ::= & & \text{\emph{Expression}}
    \\ &\Alt& \v & \text{value}
    \\ &\Alt& \x & \text{variable}
    \\ &\Alt& \seq{\e_1}{\e_2} & \text{sequencing}
    \\ &\Alt& \primcalld{\es} & \text{primop call}
    \\ &\Alt& \mcall{\e}{\es} & \text{function call}
    \\ &\Alt& \md & \text{method definition}
    \\ &\Alt& \evalg{\e} & \text{global evaluation}
    \\ &\Alt& \evalt{\MT}{\e} & \text{evaluation in a table}
    \\
    \\ \p & ::= & \evalg{\e} & \text{\emph{Program}}
    \\
    \\ \md & ::= & \mdefd & \text{\emph{Method definition}}
\end{array}\hspace{5mm}
\begin{array}{ccl@{\qquad}l}
    \\ \v & := & \ldots & \text{\emph{Value}}
    \\ &\Alt& \skp & \text{unit value}
    \\ &\Alt& \m & \text{generic function}
    \\
    \\ \g & := & \ldots & \text{\emph{Type tag}}
    \\ &\Alt& \Unit & \text{unit type}
    \\ &\Alt& \mty  & \text{type tag of function \m}
    \\
    \\ \t & := & \ldots & \text{\emph{Type annotation}}
    \\ &\Alt& \top & \text{top type}
\end{array}
\]
\caption{Surface syntax}\label{syntax}
\end{figure}

The surface syntax of \juliette is given in \figref{syntax}. It includes
method definitions \md, function calls~\mcall{\e}{\es}, sequencing
\seq{\e_1}{\e_2}, global evaluation \evalg{\e},
evaluation in a table \evalt{\MT}{\e}, variables \x, values~\v,
primitive calls~\primcall{l}{\es}, type tags \g, and type annotations \t.
Values \v include \skp (unit value, called \c{nothing} in Julia) and \m (generic
function value). Primitive operators $\primopd$ represent built-in functions
such as \c{Base.mul_int}. Type tags \g include \Unit (unit type, called
\c{Nothing} in Julia) and $\mty$ (tag of function value \m). Type annotations \t
include $\top \in \t$ ($\top$ is the top type, called \c{Any} in Julia) and $\g
\subseteq \t$ (all type tags serve as valid type annotations).

\subsection{Semantics}\label{subsec:wa-semantics}
%% --------------------------------------------------

The internal syntax of \juliette is given in the top of \figref{semantics}.  It
includes evaluation result \rslt (either value or error), method table \MT,
and two evaluation contexts, \Xx and \Cx, which are used to define
small-step operational semantics of \juliette.  Evaluation contexts \Xx are
responsible for simple sequencing, such as the order of argument evaluation;
these contexts never contain global/table evaluation expressions \evalg{\cdot}
and \evalt{\MT}{\cdot}.
World evaluation contexts \Cx, on the other hand, capture the full grammar
of expressions.  %In addition, world evaluation contexts encapsulate the idea
%of a fixed world age with the construct \evalt{\MT}{\cdot} (evaluation under
%a fixed method table \MT).

Program state is a pair \stwa{\MT}{\plugCx{\Cx}{\e}} of a global method
table \MT and an expression \plugCx{\Cx}{\e}.  %The global method table is
%empty at program start. 
We define the semantics of the calculus using two
judgments: a normal small-step evaluation denoted by
\evalwa{\MT}{\plugCx{\Cx}{\e}}{\MT'}{\plugCx{\Cx}{\e'}}, and a step to an
error \evalerrwa{\MT}{\plugCx{\Cx}{\e}}{\err}.
%Complete program evaluation is defined in the middle of \figref{semantics}.
The $\typeof(\v) \in \g$
operator returns the tag of a value.  We require that $\typeof(\skp) =
\Unit$ and $\typeof(\m) = \mty$.  We write $\typeof(\vs)$ as a shorthand for
$\obar{\typeof(\v)}$.  Function $\Primop(l, \vs) \in \rslt$ computes primop
calls, and function $\PrimopRT(l,\gs) \in \g$ indicates the tag of $l$'s
return value when called with arguments of types \gs.  These functions have
to agree, i.e.  $\forall \vs,\gs. (\typeof(\vs)=\gs \land \Primop(l, \vs) =
\v' \implies \typeof(\v') = \PrimopRT(l,\gs))$.
The subtyping relation \jsub{\t_1}{\t_2} is used for multiple dispatch.
We require that $\jsub{\t}{\top}$
($\top$ is indeed the top type) and $\jsub{\g_1}{\g_2} \Leftrightarrow \g_1
\equiv \g_2$ (tags are final, i.e. do not have subtypes).

% ~\\
% \begin{mathpar}
% \inferrule*[right=\WAE{Normal}]
%   { \evalwa{\MT}{\plugCx{\Cx}{\e}}{\MT'}{\plugCx{\Cx}{\e'}} }
%   { \evalfullwa{\MT}{\plugCx{\Cx}{\e}}{\MT'}{\plugCx{\Cx}{\e'}} }
% \hspace{1cm}
% \inferrule*[right=\WAE{Error}]
%   { \evalerrwa{\MT}{\plugCx{\Cx}{\e}}{\err} }
%   { \evalfullwa{\MT}{\plugCx{\Cx}{\e}}{\MT}{\err} }
% \end{mathpar}
\begin{figure}
  \footnotesize
  \[
  \begin{array}{ccl@{\qquad}l}
      \\ \rslt & ::= & & \text{\emph{Result}}
      \\ &\Alt& \v   & \text{value}
      \\ &\Alt& \err & \text{error}
      \\
      \\ \MT & ::= & & \text{\emph{Method table}}
      \\ &\Alt& \varnothing       & \text{empty table}
      \\ &\Alt& \MText{\MT}{\md}  & \text{table extension}
      \\

  \end{array}
  \begin{array}{ccl@{\qquad}l}
      \\ \Xx & ::= & & \text{\emph{Simple evaluation context}}
      \\ &\Alt& \hole & \text{hole}
      \\ &\Alt& \seq{\Xx}{\e} & \text{sequence}
      \\ &\Alt& \primcall{l}{\vs\ \Xx\ \es} & \text{primop call (argument)}
      \\ &\Alt& \mcall{\Xx}{\es} & \text{function call (callee)}
      \\ &\Alt& \mcall{\v}{\vs\ \Xx\ \es} & \text{function call (argument)}
      \\
      \\ \Cx & ::= & & \text{\emph{World evaluation context}}
      \\ &\Alt& \Xx & \text{simple context}
      \\ &\Alt& \plugx{\Xx}{\evalg{\Cx}} & \text{global evaluation}
      \\ &\Alt& \plugx{\Xx}{\evalt{\MT}{\Cx}} & \text{evaluation in a table \MT}
  \end{array}
  \]
~\\
\begin{mathpar}
\inferrule[\WAE{Seq}]
  { }
  { \evalwad{\plugCx{\Cx}{\seq{\v}{\e}}}{\plugCx{\Cx}{\e}} }

\inferrule[\WAE{Primop}]
  { \Primop(l, \vs) = \v' }
  { \evalwad{\plugCx{\Cx}{\primcalld{\vs}}}{\plugCx{\Cx}{\v'}} }

\inferrule[\WAE{MD}]
  { \md \equiv \mdefd }
  { \evalwa
      {\MT}{\plugCx{\Cx}{\md}}
      {\MText{\MT}{\md}}{\plugCx{\Cx}{\m}} }

\inferrule[\WAE{CallGlobal}]
  { }
  { \evalwad
      {\plugCx\Cx{\evalg{\plugx\Xx{ \mcall\m\vs}}}}
      {\plugCx\Cx{\evalg{\plugx\Xx{ \evalt{\MT}{\mcall\m\vs}}}}}}

\inferrule[\WAE{CallLocal}]
  { \typeof(\vs) = \gs \\ \getmd(\MT',\m,\gs)=\mdefd }
  { \evalwad
      {\plugCx\Cx{\evalt{\MT'}{\plugx\Xx{\mcall\m\vs}}}}
      {\plugCx\Cx{\evalt{\MT'}{\plugx\Xx{\e\subst\xs\vs}}}} }

\inferrule[\WAE{ValGlobal}]
  { }
  { \evalwad{\plugCx{\Cx}{\evalg{\v}}}{\plugCx{\Cx}{\v}} }

\inferrule[\WAE{ValLocal}]
  { }
  { \evalwad{\plugCx\Cx{\evalt{\MT'}\v}}{\plugCx\Cx\v} }
\end{mathpar}

~\\

\begin{mathpar}\footnotesize
\inferrule[\WAE{VarErr}]
  { }
  { \evalerrwad{\plugCx{\Cx}{\x}}{\err} }

\inferrule[\WAE{PrimopErr}]
  { \Primop(l, \vs) = \err }
  { \evalerrwad{\plugCx{\Cx}{\primcalld{\vs}}}{\err} }

\inferrule[\WAE{CalleeErr}]
{ \v_c \neq \m }
{ \evalerrwad
    {\plugCx{\Cx}{\mcall{\v_c}{\vs}}}
    {\err} }

\inferrule[\WAE{CallErr}]
{ \typeof(\vs) = \gs \\ \getmd(\MT', \m, \gs) = \err }
{ \evalerrwad
    {\plugCx{\Cx}{\evalt{\MT'}{\plugx{\Xx}{\mcall{\m}{\vs}}}}}
    {\err} }
\end{mathpar}
~\\[3mm]
\[\footnotesize
\begin{array}{rcl}
  \getmd(\MT, \m, \gs) & = &
    \min(\applcbl(\latest(\MT), \m, \gs)) \\
  \\
  \latest(\MT) & = & \latest(\emptyset, \MT) \\
  \latest(mds, \varnothing) & = & mds \\
  \latest(mds, \MText{\MT}{\md}) & = & \latest(mds \cup \md, \MT) 
    \text{ if } \neg \containseq(mds, \md) \\
  \latest(mds, \MText{\MT}{\md}) & = & \latest(mds, \MT)
    \qquad\, \text{ if } \containseq(mds, \md) \\
  \\
  \applcbl(mds, \m, \gs) & = & \{\mdefd \in mds\ |\
    \jsub{\gs}{\ts} \} \\
  \\
  \min(mds) & = & \mdefd \in mds \text{ such that }
    \forall \mdef{\m}{\obar{\jty{\_}{\t'}}}{\_}
    \in mds\ .\ \jsub{\ts}{\ts'}\\
  \min(mds) & = & \err \text{ otherwise}
  \\
  \containseq(mds, \md) & = & \exists\,\md' \in mds \text{ such that } \\
    & & \quad (\md \equiv \mdef{\m}{\obar{\jty{\_}{\t}}}{\_}) \ \land\
        (\md' \equiv \mdef{\m}{\obar{\jty{\_}{\t'}}}{\_}) \
        \land \ \jsub{\ts}{\ts'} \land\ \jsub{\ts'}{\ts}
\end{array}
\]
\caption{Internal Syntax and Semantics}\label{semantics}
\end{figure}

\paragraph{Normal Evaluation}

These rules capture successful program executions.
Rule \WAE{Seq} is completely standard: it throws away the evaluated part of
a sequencing expression. Rules \WAE{ValGlobal} and \WAE{ValLocal} pass value
\v to the outer context. This is similar to Julia where \eval returns the
result of evaluating the argument to its caller.
Rule \WAE{MD} is responsible for updating the global table: a method
definition $\md$ will extend the current global table \MT into
$\MText{\MT}{\md}$, and itself evaluate to $\m$, which is a function value.
Note that \WAE{MD} only extends the method table and leaves existing
definitions in place.  If the table contains multiple definitions of a
method with the same signature, it is then the dispatcher's responsibility
to select the right method; this mechanism is described below in more
detail.

The two call forms $\WAE{CallGlobal}$ and $\WAE{CallLocal}$ form the core of
the calculus. The rule $\WAE{CallGlobal}$ describes the case where a method
is called directly from a global evaluation expression. In Julia, this means
either a top-level call, an \invokelatest call, or a call within \eval such
as \c{eval(:(g(...)))}. The ``direct'' part is encoded with the use of a
simple evaluation context \Xx.  In this global-call case, we need to save
the current method table into the evaluation context for a subsequent use by
$\WAE{CallLocal}$. To do this, we annotate
the call \mcall{\m}{\vs} with a copy of the current global method table
$\MT$, producing \evalt{\MT}{\mcall\m\vs}.

To perform a local call---or, equivalently, a call after the invocation has
been wrapped in an annotation specifying the current global
table---$\WAE{CallLocal}$ is used.  This rule resolves the call according to
the tag-based multiple-dispatch semantics in the ``deepest'' method table
$\MT'$ (the use of \Xx makes sure there are no method tables between $\MT'$
and the call). Once an appropriate method has been found, it proceeds as a
normal invocation rule would, replacing the method invocation with the
substituted-for-arguments method body. Note that the body of the method is
still wrapped in the \evalt{\MT'}{} context. This ensures that nested calls
will be resolved in the same table (unless they are more deeply wrapped
in a global evaluation \evalg{}).

An auxiliary meta-function $\getmd(\MT, \m, \gs)$, which is used to resolve
multiple dispatch, is defined in the bottom of \figref{semantics}. This function
returns the most specific method applicable to arguments with type tags \gs, or
errs if such a method does not exist.  If the method table contains
multiple equivalent methods, older ones are ignored. For example, for the
program
\[
  \evalg{\seq{\mdef{\mval{g}}{}{2}}
  {\seq{\mdef{\mval{g}}{}{42}}{\mcall{\mval{g}}{}}}},
\]
function call \mcall{\mval{g}}{} is going to be resolved in the table
$\MText{(\MText{\varnothing}{\mdef{\mval{g}}{}{2}})}{\mdef{\mval{g}}{}{42}}$,
which contains two equivalent methods (we call methods equivalent if they
have the same name and their argument type annotations are equivalent with
respect to subtyping). In this case, the function \getmd will
return method \mdef{\mval{g}}{}{42} because it is the newest method out of
the two.

Note that functions can be mutually recursive because of the dynamic nature
of function call resolution.

\paragraph{Error Evaluation}

These rules capture all possible error states of \juliette.
Rule \WAE{VarErr} covers the case of a free variable, an \c{UndefVarError} in
Julia.  \WAE{PrimopErr} accounts for errors in primitive operations such as
\c{DivideError}.  \WAE{CalleeErr} fires when a non-function value is called.
Finally, \WAE{CallErr} accounts for multiple-dispatch resolution errors,
e.g. when the set of applicable methods is empty (no method found), and when
there is no best method (ambiguous method).

% -----

%\clearpage
\subsection{Example}

\figref{fig:jl-wa-example} shows a translation of the program from
\figref{fig:eval-methods} to \juliette. First note that, as part of
the translation, we wrap the entire program in \evalg{}, indicating
that the outermost scope is the top level.  Translation of method
calls and definitions then proceeds, using $\mdef{\m}{\xs}{\e}$ as a
shorthand for $\mdef{\m}{\obar{\jty{\x}{\top}}}{\e}$ where $\top$ is
the top type. Method bodies are converted by replacing \eval
invocations with their expressions wrapped in \evalg{}. The \evalg{}
context of \e in Juliette effectively acts the same way that \eval of
\e does in Julia, but evaluates variables in \e using local, rather
than global, scope.

\begin{figure}[!h]
\begin{minipage}{6cm}\vspace{0.3cm}
\begin{lstlisting}
g()  = 2
f(x) = (eval(:(g() = $x)); x * g())
f(42)
\end{lstlisting}
\end{minipage}
\hspace{1cm}
\begin{minipage}{6cm}
\[
\begin{array}{rl}
  {\color{violet} \llparenthesis}
  & \mdef{\mval{g}}{}{2}\ ; \\
  & \mdef{\mval{f}}{\x}{(\seq{\evalg{\mdef{\mval{g}}{}{\x}}}
                              {\x * \mval{g}()})}\ ; \\
  & \mcall{\mval{f}}{42} \quad {\color{violet} \rrparenthesis}
\end{array}
\]
\end{minipage}
\caption{From Julia (left) to \juliette (right)}\label{fig:jl-wa-example}
\end{figure}

Now we will show the execution of this translated program according to our
small-step semantics.
The initial state is $\stwa{\varnothing}{\p}$
where \p is the program on the right of \figref{fig:jl-wa-example} (and the
\c{*} operator is a primop). The first several steps of evaluation use rules
\WAE{MD} and \WAE{Seq} to add the definitions of \mval{g} and \mval{f} to
the global table.  This produces the state
\[
  \stwa{\MT_0}{\evalg{\mcall{\mval{f}}{42}}},
\]
where
\[
\begin{array}{ccl}
  \MT_0 & = & (\MText{\varnothing}{\mdef{\mval{g}}{}{2}}) \\
  & & \MText{}{\mdef{\mval{f}}{\x}{(\seq{\evalg{\mdef{\mval{g}}{}{\x}}}
              {\x * \mval{g}()})}}
\end{array}
\]
Next, using the \WAE{CallGlobal} rule, the top-level call \mcall{\mval{f}}{42}
steps to \evalt{\MT_0}{\mcall{\mval{f}}{42}}. This then produces the state
\[
  \stwa{\MT_0}{\evalg{\evalt{\MT_0}{\mcall{\mval{f}}{42}}}},
\]
copying the global table into the context \evalt{\MT_0}{\cdot}.
Now, rule \WAE{CallLocal} can be used to resolve the call \mcall{\mval{f}}{42}
in the table $\MT_0$. Method $\mdef{\mval{f}}{\x}{\ldots}$ is the only
method of \mval{f} and it is applicable to the integer argument
($\typeof(42) = \jsub{\mathtt{Int}}{\top}$), so the program steps to:
\[
  \stwa{\MT_0}{\evalg{\ \evalt{\MT_0}
  { \seq{\evalg{\mdef{\mval{g}}{}{42}}}{42 * \mcall{\mval{g}}{}} }\ }}.
\]
The next expression to evaluate is the new \mval{g} definition,
$\mdef{\mval{g}}{}{42}$. Rule \WAE{MD} fires and the program steps to
\[
  \stwa{\MT_1}{\evalg{\ \evalt{\MT_0}
  { \seq{\evalg{\mval{g}}}{42 * \mcall{\mval{g}}{}} }\ }},
\]
where
\[
\begin{array}{cclcl}
  \MT_1 & = & \MT_0 & = & ((\MText{\varnothing}{\mdef{\mval{g}}{}{2}}) \\
  & & \MText{}{\mdef{\mval{g}}{}{42}} &
    & \MText{}{\mdef{\mval{f}}{\x}{(\seq{\evalg{\mdef{\mval{g}}{}{\x}}}
              {\x * \mval{g}()})}}) \\
  & & & & \MText{}{\mdef{\mval{g}}{}{42}}.
\end{array}
\]
The next two steps are:
\begin{eqnarray}
  \stwa{\MT_1}{\evalg{\ \evalt{\MT_0}
    { \seq{\evalg{\mval{g}}}{42 * \mcall{\mval{g}}{}} }\ }}
    & \xrightarrow{\WAE{ValGlobal}}
    & \stwa{\MT_1}{\evalg{\evalt{\MT_0}
        { \seq{\mval{g}}{42 * \mcall{\mval{g}}{}} }}}
  \\
    & \xrightarrow{\WAE{Seq}}
    & \stwa{\MT_1}{\evalg{\evalt{\MT_0}
        { 42 * \mcall{\mval{g}}{} }}}. \label{eq:ex-step}
\end{eqnarray}
Note that the last program state is represented by
$\stwa{\MT_1}{\plugCx{\Cx}{\evalt{\MT_0}{\plugx{\Xx}{\mcall{\mval{g}}{}}}}}$,
where $\Cx = \evalg{\hole}$ and $\Xx = 42 * \hole$.
So we have to use \WAE{CallLocal} again to resolve \mcall{\mval{g}}{}
in the $\MT_0$ that is fixed in the context.
Table $\MT_0$ has only one definition of \mval{g}, the one that returns $2$,
so the program steps to:
\[
  \stwa{\MT_1}{\evalg{\evalt{\MT_0}{ 42 * 2 }}}.
\]
Finally, the application of \WAE{Primop}, \WAE{ValLocal}, and \WAE{ValGlobal}
leads to the final state:
\begin{equation}\label{eq:ex-final-orig}
  \stwa{\MT_1}{84}.
\end{equation}

Now, consider a modification of the original program where
in the definition of \mval{f}, the call \mcall{\mval{g}}{} is wrapped into
a global evaluation \evalg{\mval{g}()}:

\begin{minipage}{7cm}\vspace{0.3cm}
\begin{lstlisting}
  g()  = 2
  f(x) = (eval(:(g() = $x)); x * eval(:(g())))
  f(42)
\end{lstlisting}
\end{minipage}
\hspace{1cm}
\begin{minipage}{6cm}
\[
\begin{array}{rl}
  {\color{violet} \llparenthesis}
  & \mdef{\mval{g}}{}{2}\ ; \\
  & \mdef{\mval{f}}{\x}{(\seq{\evalg{\mdef{\mval{g}}{}{\x}}}
                             {\x * \evalg{\mval{g}()}})}\ ; \\
  & \mcall{\mval{f}}{42} \quad {\color{violet} \rrparenthesis}
\end{array}
\]
\end{minipage}\vspace{0.3cm}

\noindent
At the beginning, the modified program will run similarly to the original one,
and with step \eqref{eq:ex-step}, it will reach the state:
\[
  \stwa{\MT_1}{\evalg{\evalt{\MT_0}{ 42 * \evalg{\mcall{\mval{g}}{}} }}}.
\]
Here, \evalg{\evalt{\MT_0}{42 * \evalg{\mcall{\mval{g}}{}}}} is represented
by $\plugCx{\Cx}{\evalg{\plugx{\Xx}{\mcall{\mval{g}}{}}}}$,
where $\Cx = \evalg{\evalt{\MT_0}{42 * \hole}}$ and $\Xx = \hole$.
Therefore, the call \mcall{\mval{g}}{} is back at the top level.
With \WAE{CallGlobal} rule,
the call steps to $\evalt{\MT_1}{\mcall{\mval{g}}{}}$ because $\MT_1$
is the \emph{current global} table, thus producing the state:
\[
  \stwa{\MT_1}{\evalg{\evalt{\MT_0}{
    42 * \evalg{\ \evalt{\MT_1}{\mcall{\mval{g}}{}}\ } }}}.
\]
Resolved in $\MT_1$, call \mcall{\mval{g}}{} returns $42$,
and thus the whole program ends in the final state:
\[
  \stwa{\MT_1}{1764}.
\]
Note that the resulting global table is the same as in \eqref{eq:ex-final-orig},
but the return value is different.

\subsection{Properties}\label{subsec:properties}
%% --------------------------------------------------

\juliette operational semantics is deterministic,
and all failure states are captured by error evaluation
(meaning that a \juliette program never gets stuck).

\begin{lemma}[Unique Form of Expression]\label{lem:paper:expr-form-unique}
  Any expression \e can be uniquely represented in one of the following ways:
  \begin{enumerate}[label=(\alph*)]
    \item $\e = \v$; or
    \item $\e = \plugx{\Xx}{\mcalld}$; or
    \item $\e = \plugx{\Cx}{\rdx}$,
  \end{enumerate}
  where \rdx (shown in~\figref{fig:wa-rdx-main})
  is a subset of expressions driving the reduction.
\end{lemma}
\begin{proof}
  By induction on \e\APPENDICITIS{.}
  {, using auxiliary definitions and lemmas about the representation of
   expressions from \appref{app:proofs:canonical}. The full proof is given
   by \lemref{lem:expr-form-canonical} on page~\pageref{lem:expr-form-canonical}.}
\end{proof}

\begin{theorem}[Progress]\label{thm:paper:progress}
  For any program \p and method table \MTg, 
  the program either reduces to a value,
  or it makes a step to another program,
  or it errs.
  That is, one of the following holds:
  \begin{enumerate}[label=(\alph*)]
    \item $\evalwa{\MTg}{\p}{\MTg'}{\v}$; or
    \item $\evalwa{\MTg}{\p}{\MTg'}{\p'}$; or
    \item $\evalerrwa{\MTg}{\p}{\err}$.
  \end{enumerate}
\end{theorem}
\begin{proof}
  By case analysis on $\p = \evalg{\e}$, using \lemref{lem:paper:expr-form-unique}.
  \APPENDICITIS{}
  {The full proof (provided in \appref{app:proofs:progress},
   page~\pageref{proof:progress}) relies on
   an auxiliary lemma~\ref{lem:redex-steps} about reducible expressions.}
\end{proof}

\begin{theorem}[Determinism]\label{thm:paper:eval-deterministic}
  \juliette semantics is deterministic.
\end{theorem}
\begin{proof}
  \APPENDICITIS
     {The proof}
     {The full proof is provided in \appref{app:proofs:determinism},
      page~\pageref{proof:eval-deterministic}.  It}
  relies on the fact that (1) any expression that steps
  can be represented as $\plugCx{\Cx}{\rdx}$, and (2) such a representation
  is unique by \lemref{lem:paper:expr-form-unique}.
  By case analysis on \rdx, we can see that for all redex bases
  except \primcalld{\vs} and \evalt{\MT}{\plugx{\Xx}{\mcall{\m}{\vs}}},
  there is exactly one (normal- or error-evaluation) rule applicable.
  For \primcalld{\vs} and \evalt{\MT}{\plugx{\Xx}{\mcall{\m}{\vs}}},
  there are two rules for each, but their premises are incompatible.
  Thus, for any expression \plugCx{\Cx}{\rdx}, exactly one rule is applicable.
\end{proof}

\begin{figure}
\[\footnotesize
  \begin{array}{ccl}
      \\ \rdx & ::= &
      \\ &\Alt& \x
      \\ &\Alt& \seq{\v}{\e}
      \\ & &
 \end{array}\hspace{0.05cm}
 \begin{array}{ccl}
  \\ & &
  \\ &\Alt& \primcalld{\vs}
  \\ &\Alt& \mcall{\vnm}{\vs}
 \end{array}\hspace{0.05cm}
 \begin{array}{ccl}
  \\ & &
  \\ &\Alt& \md
  \\ & &
  \end{array}\hspace{0.05cm}
 \begin{array}{ccl}
  \\ & &
  \\ &\Alt& \evalg{\v}
  \\ &\Alt& \evalt{\MT}{\v}
  \end{array}\hspace{0.05cm}
  \begin{array}{ccl}
   \\ & &
   \\ &\Alt& \evalg{\plugx{\Xx}{\mcall{\m}{\vs}}}
   \\ &\Alt& \evalt{\MT}{\plugx{\Xx}{\mcall{\m}{\vs}}}
   \end{array}
\]
\caption{Redex bases}\label{fig:wa-rdx-main}
\end{figure}

\subsection{Optimizations}\label{subsec:optimizations}
%% --------------------------------------------------

The world-age semantics allows function-call optimization even in the presence
of \eval. Recall how an {\eval}ed or top-level function call
\evalg{\mcall{\m}{\vs}} steps.
First, rule \WAE{CallGlobal} is applied: it fixes the current state of
the global table \MT in the call's context, stepping the call to
$\evalt{\MT}{\mcall{\m}{\vs}}$.  Then, the call \mcall{\m}{\vs} itself, and all
of its nested calls (unless they are additionally wrapped into \evalg{}),
are resolved using the now-local table \MT. Therefore, \MT provides all necessary
information for the resolution of such calls, and they can be optimized based
on the method table \MT.
%to resolve this (and all future) method calls until we either
%create a further nested context or return from this current context.

Next, we will focus on  three generic-call optimizations: inlining,
specialization, and transforming generic calls into direct calls
(devirtualization). Namely, we provide formal definitions of these
optimizations and show them correct.

\paragraph{Inlining}

If a function call is known to dispatch to a certain method
using a \emph{fixed} method table, it
might be possible to \emph{inline} the body of the method in place of the
call. For example, consider a program on the left of
\figref{fig:opt-src-example}.  The call \c{f(5)} has no choice but to
dispatch to the only definition of~\c f.  Because the call \c{g(x)} in
\c{f(5)} is not wrapped in an \eval, it is known that the call to \c g is
going to be dispatched in the context with exactly two methods of \c g:
\c{g1} and \c{g2}.  Furthermore, since \c x is known to be of type (tag)
\c{Int} inside \c f, we know that \c{g(x)} has to dispatch to the method
\c{g1} (because \c{Int <: Any} but \c{Int </: Bool}).
Thus, it is possible to optimize method \c f for the call \c{f(5)} by inlining
\c{g(x)}, which yields the following optimized definition of \c f:

%\vspace{-0.4em}
\begin{minipage}{5.8cm}
\begin{lstlisting}
f(x::Int) = x * (x + x)
\end{lstlisting}
\end{minipage}
\vspace{-0.5em}

\paragraph{Direct-call optimization}

When inlining is not possible or desirable, but it is clear which method
is going to be invoked, a function call can be replaced by a direct invocation.
Consider the example on the right of \figref{fig:opt-src-example}.
The only difference from the previous example is that the argument of \c g
inside~\c f is not a variable but an expression \c{(println(x); x)}.  This
expression always returns an integer, so we know that at run time, that \c g
will be dispatch to method \c{g1}. However, unlike previously,
the call to \c g cannot be inlined using direct syntactic substitution. In that
case, the value of \c x would be printed twice instead of just once, because
inlining would transform \c{g((println(x);x))} into \c{(println(x);x) +
  (println(x);x)} and thus change the observable behavior of the program.
It is still possible to optimize \c f, by replacing the generic call to \c g
with a \emph{direct call} to the method \c{g1}. In pseudo-code, this can be
written as:

%\vspace{-0.4em}
\begin{minipage}{5.8cm}
\begin{lstlisting}
f(x::Int) = x * g@g1((println(x); x))
\end{lstlisting}
\end{minipage}

\noindent
In the calculus, we model a direct call as a call to a new function with
a single method such that the name of the function is not used anywhere
in the original method table or expression. For example, for the program
above, we can add function \c h with only one method \c{h(x::Int)=x+x},
allowing \c f to be optimized to:

%\vspace{-0.4em}
\begin{minipage}{5.8cm}
\begin{lstlisting}
f(x::Int) = x * h(println(x); x)
\end{lstlisting}
\end{minipage}
\vspace{-0.4em}

\begin{figure}
  \begin{minipage}{5.4cm}
\begin{lstlisting}
g(x::Any)  = x + x    # g1
g(x::Bool) = x        # g2
f(x::Int)  = x * g(x)
f(5)
\end{lstlisting}
  \end{minipage}
  \hspace{1cm}
  \begin{minipage}{5.4cm}
\begin{lstlisting}
g(x::Any)  = x + x    # g1
g(x::Bool) = x        # g2
f(x::Int)  = x * g((println(x);x))
f(5)
\end{lstlisting}
  \end{minipage}
\caption{Candidate programs for inlining (on the left)
  and direct call optimization (on the right)}\label{fig:opt-src-example}
\end{figure}

\paragraph{Specialization}

The final optimization we consider is specialization of methods for argument
types. In \figref{fig:opt-src-example}, method \c{g1} is defined for \c x of
type \c{Any}, meaning that the call \c{x+x} can be dispatched to
any of at least \jladdnum standard methods. But because, within \c f, \c g
is known to be called with an argument of type \c{Int} (due to \c x in \c f
having that type), it is possible to generate a new implementation
of \c g \emph{specialized} for this argument type. The advantage is that the
specialized implementation can directly use efficient integer addition.
Thus, combined with the direct call, we have:

%\vspace{-0.4em}
\begin{minipage}{5.8cm}
\begin{lstlisting}
g(x::Int) = Base.add_int(x, x) # g3
f(x::Int) = x * g@g3((println(x); x))
\end{lstlisting}
\end{minipage}
\vspace{-0.4em}

\noindent
In the calculus, specialization is modeled similarly to direct calls:
as a function with a fresh name.

\subsection{Optimization Correctness}\label{subsec:correctness}
%% --------------------------------------------------

In this section, we present a formal definition of optimizations and
state the main theorem about their correctness.
The general idea of optimizations is as follows: if an expression \e is going
to be executed in a fixed method table \MT, it is safe to instead execute
\e in a table $\MT'$ obtained by optimizing method definitions of \MT
(like we did with the definition of \c f in the examples above).

As demonstrated by the examples, the first ingredient of optimizations
is type information, which is necessary to ``statically'' resolve
function calls; for this, we use a simple concrete-typing relation defined in
\figref{fig:wa-typing}.
The relation \typedwad{\e}{\g} propagates information about variables and
type tags of values, and succeeds only if the expression would always reduce
to a value of concrete type \g \emph{if} it reduces to any value.
This is because to resolve a function call, we need to know
the type tags of its arguments. A typing relation can be more complex to enable
further optimization opportunities (and it is much more complex in Julia),
but typing of Julia is a separate topic that is out of scope of this paper:
here, we focus on compiler optimization and use concrete typing only as a tool.

\figref{fig:wa-table-opt} shows the judgments related to method-table optimization.
The rule \WAOT{MethodTable} says that an optimized version $\MT'$ of
table $\MT$ (1) has to have all the methods of $\MT$, although they can be
optimized, and (2) can have more methods given that their names do not appear
%neither in the original table \MT nor in the expression \e for which the table
%is getting optimized. The latter enables adding new methods that model
in the original table \MT. The latter enables adding new methods that model
direct calls and specializations, and the former allows for optimization
of existing methods. According to the rule \WAOD{MD}, a method
in $\MT'$ optimizes a method in $\MT$ if it has the same signature
(i.e. name and argument types), and its body is an optimization of the original
body of the method being optimized.

The method-optimization environment $\SpecEnv$
tracks direct calls and specializations: \mdeqd tells that when arguments
of \m have type tags \gs, a call to \m in \MT can be replaced by a call
to $\m'$ in $\MT'$. Note that all entries of $\SpecEnv$ need to be valid
according to {\small \textsc{MethodOpt-Valid}}: assuming that the methods
are in the optimization relation, their bodies indeed have to be in that
relation (the assumption is needed to handle recursion).
Both \WAOD{MD} and {\small \textsc{MethodOpt-Valid}} rely
on expression optimization to relate method bodies.

% -----

\begin{figure}\footnotesize
\[
\begin{array}{rcll}
  \Gm & ::= & & \text{\emph{Typing environment}}\\
      & \Alt & \varnothing \\
      & \Alt & \Gm, \x:\t \\
\end{array}
\]
\begin{mathpar}
\inferrule[\WAT{Var}]
  { \Gm(\x) = \g }
  { \typedwad{\x}{\g} }

\inferrule[\WAT{Val}]
  { \typeof(\v) = \g }
  { \typedwad{\v}{\g} }

\inferrule[\WAT{MD}]
  { }
  { \typedwad{\mdefd}{\mty} }

\inferrule[\WAT{Seq}]
  { \typedwad{\e_2}{\g} }
  { \typedwad{(\seq{\e_1}{\e_2})}{\g} }

\inferrule[\WAT{Primop}]
  { \PrimopRT(l, \gs)=\g' \\ \typedwad{\e_i}{\g_i} }
  { \typedwad{\primopd(\es)}{\g'} }

\inferrule[\WAT{EvalGlobal}]
  { \typedwad{\e}{\g} }
  { \typedwad{\evalg{\e}}{\g} }

\inferrule[\WAT{EvalLocal}]
  { \typedwad{\e}{\g} }
  { \typedwad{\evalt{\MT}{\e}}{\g} }
\end{mathpar}
\caption{Concrete-typing judgment}\label{fig:wa-typing}
\end{figure}

% -----

\begin{figure}\footnotesize
  \[
  \begin{array}{rcll}
    \SpecEnv & ::= & & \text{\emph{Method-optimization environment}} \\
             & \Alt & \varnothing \\
             & \Alt & \SpecEnv, \mdeqd \\
  \end{array}
  \]
  \begin{mathpar}
  \\
  \inferrule[MethodOpt-Valid]
    { \getmd(\MT,\m,\gs) = \mdef{\m}{\obar{\jty{\x}{\t}}}{\e_b} \\\\
      \getmd(\MT',\m',\gs) = \mdef{\m'}{\obar{\jty{\x'}{\t'}}}{\e'_b} \\\\
      \expropt{\SpecEnv}{\obar{\x:\g}}{\MT}{\e_b}{\MT'}{\e'_b\subst{\xs'}{\xs}}
    }
    { \mspecd }

  \inferrule[\WAOD{MD}]
  { \expropt{\SpecEnv}{\obar{\x:\t}}{\MT}{\e}{\MT'}{\e'\subst{\x'}{\x}} }
  { \mdoptd{\mdefd}{\mdef{\mname}{\obar{\jty{\x'}{\t}}}{\e'}} }
  \\
  \inferrule[\WAOT{MethodTable}]
    { \MT  = \MText{\MText{\md_1}{\ldots}}{\md_n} \qquad
      \MT' = \MText{\MText{\MText{\MText{\MText{\md'_1}{\ldots}}{\md'_n}}
                    {\md'_{n+1}}}{\ldots}}{\md'_k} \\
      \forall 1 \leq i \leq n.\ \mdoptd{\md_i}{\md'_i} \\
      \forall (\mdeqd) \in \SpecEnv.\ \mspecd \\
      \forall n+1 \leq j \leq k. \mathop{name}(\md'_j)
        \text{ does not occur in } \MT }
    { \tableoptd }
  
  \inferrule[MNamesCompat]
  { \forall \m \text{ referenced by } \e. \ \
    \m \in \dom(\MT) \iff \m \in \dom(\MT') }
  { \tableoptexprd{\e} }

  \inferrule[\WAOT{MethodTable-Expr}]
    { \tableoptd \\ \tableoptexprd{\e} }
    { \tableoptrefd }
  \end{mathpar}
\caption{Method table \& definition optimization}\label{fig:wa-table-opt}
\end{figure}
% \forall n+1 \leq j \leq k. \mathop{name}(\md'_j)
% \text{ does not occur in } \MT \text{ or } \e }

% -----

% \nv & ::= & & \text{\emph{Near-value}} \\
% &\Alt& \v & \text{value} \\
% &\Alt& \x & \text{variable} \\
\begin{figure}\footnotesize
\[
\begin{array}{rcll}
  \nv & ::= & \v\ \Alt\, \x & \text{\emph{Near-value}} \\
  \\
\end{array}
\]
\begin{mathpar}
\inferrule[\WAOE{Val}]
  { \v \neq \m }
  { \exproptd{\v}{\v} }

\inferrule[\WAOE{ValFun}]
  { \tableoptexprd{\m} }
  { \exproptd{\m}{\m} }

\inferrule[\WAOE{Var}]
  { }
  { \exproptd{\x}{\x} }
\\

\inferrule[\WAOE{Global}]
  { \tableoptexprd{\e} }
  { \exproptd{\evalg{\e}}{\evalg{\e}} }

\inferrule[\WAOE{Local}]
  { \tableoptexprd{\e} }
  { \exproptd{\evalt{\MT_l}{\e}}{\evalt{\MT_l}{\e}} }

\inferrule[\WAOE{MD}]
  { \tableoptexprd{\mathop{name}(\md)} }
  { \exproptd{\md}{\md} }
\\

\inferrule[\WAOE{Seq}]
  { \exproptd{\e_1}{\e'_1} \\ \exproptd{\e_2}{\e'_2} }
  { \exproptd{\seq{\e_1}{\e_2}}{\seq{\e'_1}{\e'_2}} }

\inferrule[\WAOE{Primop}]
  { \forall i.\ \exproptd{\e_i}{\e'_i} }
  { \exproptd{\primcall{l}{\es}}{\primcall{l}{\es'}} }

\inferrule[\WAOE{Call}]
  { \exproptd{\e_c}{\e'_c} \\ \forall i.\ \exproptd{\e_i}{\e'_i} }
  { \exproptd{\mcall{\e_c}{\es}}{\mcall{\e'_c}{\es'}} }

\inferrule[\WAOE{Inline}]
  { \forall i.\ \typedwad{\nv_i}{\g_i} \\
    \getmd(\MT,\m,\gs) = \mdef{\m}{\obar{\jty{\x}{\t}}}{\e_b} \\
    \expropt{\SpecEnv}{\Gm}{\MT}{\e_b\subst{\xs}{\nvs}}{\MT'}{\e'} }
  { \exproptd{\mcall{\m}{\nvs}}{\seq{\skp}{\e'}} }

\inferrule[\WAOE{Direct}]
  { \forall i.\ \exproptd{\e_i}{\e'_i} \\ \typedwad{\e'_i}{\g_i} \\
    (\mdeqd) \in \SpecEnv }
  { \exproptd{\mcall{\m}{\es}}{\mcall{\m'}{\es'}} }
\end{mathpar}
\caption{Expression optimization}\label{fig:wa-expr-opt}
\end{figure}

Finally, the expression-optimization relation is shown in \figref{fig:wa-expr-opt}.
Note that the rules do not allow for function-call optimizations inside the
global-evaluation construct \evalg{}: the only applicable rule in that case
is \WAOE{Global}. Function calls can only be optimized if they are fixed-table
calls.  Rules \WAOE{Inline} and \WAOE{Direct} correspond to the inlining and
the direct-call/specialization optimizations, respectively. As discussed earlier,
inlining cannot be done if a function is called with expression arguments.
Therefore, in \WAOE{Inline} we use an auxiliary definition \nv,
``near-value'', which is either a value or a variable.
Because we model direct-call and specialization optimizations as calls to
freshly-named methods, the main job is done in the table-optimization rule
\WAOT{MethodTable}; rule \WAOE{Direct} only records the fact of invoking
a specific method. If the method definition of $\m'$ from \mdeqd
has the same parameter-type annotations
as \m, it represents a direct call to an original method of \m; otherwise, it
represents a specialized method. Note that for all optimizations, function-call
arguments have to be concretely typed. Otherwise, we do not know
definitively how a function call is going to be dispatched at run time.

The optimizations defined in \figref{fig:wa-table-opt}--\ref{fig:wa-expr-opt}
are sound. That is, the evaluation of the original and optimized programs
yield the same result. To show this, we establish
a bisimulation relation between original and optimized expressions
(after the following auxiliary lemmas):

\begin{lemma}[Context Irrelevance]\label{lem:paper:context-irrelevant}
  For all \Cx, $\Cx'$, \rdx, $\e'$, \MTg, $\MTg'$, the following holds:
  \[ \evalwa{\MTg}{\plugCx{\Cx}{\rdx}}{\MTg'}{\plugCx{\Cx}{\e'}}
      \quad \iff \quad 
      \evalwa{\MTg}{\plugCx{\Cx'}{\rdx}}{\MTg'}{\plugCx{\Cx'}{\e'}}. \]
\end{lemma}
\begin{proof}
  By analyzing normal-evaluation steps, we can see that only \rdx matters
  for the reduction. Formally, the proof goes by inspecting a reduction step
  for \plugCx{\Cx}{\rdx} (\plugCx{\Cx'}{\rdx}) and building a corresponding
  step for \plugCx{\Cx'}{\rdx} (\plugCx{\Cx}{\rdx}).
\end{proof}

\begin{lemma}[Simple-Context Irrelevance]\label{lem:paper:simple-context-irrelevant}
  For all \MT, \Cx, \e, \MTg, $\e'$, $\MTg'$, \Xx, the following holds:
  \[ 
     \evalwa{\MTg}{\plugCx{\Cx}{\evalt{\MT}{\e}}}
            {\MTg'}{\plugCx{\Cx}{\evalt{\MT}{\e'}}}
     \quad\implies\quad
     \evalwa{\MTg}{\plugCx{\Cx}{\evalt{\MT}{\plugx{\Xx}{\e}}}}
            {\MTg'}{\plugCx{\Cx}{\evalt{\MT}{\plugx{\Xx}{\e'}}}}.
  \]
\end{lemma}
\begin{proof}
  By \lemref{lem:paper:expr-form-unique}, \e is either \v or \plugx{\Xx_e}{\mcalld}
  or \plugx{\Cx_e}{\rdx}. If \e is \v, the assumption of the lemma does not hold
  (\plugCx{\Cx}{\evalt{\MT}{\v}} would step to \plugCx{\Cx}{\v}),
  so only \plugx{\Xx_e}{\mcalld} and \plugx{\Cx_e}{\rdx} cases are possible.
  \begin{itemize}
    \item When \e is \plugx{\Xx_e}{\mcalld},
      $\evalt{\MT}{\e} = \evalt{\MT}{\plugx{\Xx_e}{\mcalld}}$ is a redex,
      and \plugCx{\Cx}{\evalt{\MT}{\e}} steps by rule \WAE{CallLocal}.
      But $\evalt{\MT}{\plugx{\Xx}{\e}} =
        \evalt{\MT}{\plugx{\Xx}{\plugx{\Xx_e}{\mcalld}}}$
      is also a redex, and \plugCx{\Cx}{\evalt{\MT}{\plugx{\Xx}{\e}}}
      steps by rule \WAE{CallLocal} similarly to \plugCx{\Cx}{\evalt{\MT}{\e}}.
    \item When \e is \plugx{\Cx_e}{\rdx},
      $\evalt{\MT}{\e} = \evalt{\MT}{\plugx{\Cx_e}{\rdx}}$
      and $\plugCx{\Cx}{\evalt{\MT}{\e}} = \plugCx{\Cx'}{\rdx}$
      where $\Cx' = \plugCx{\Cx}{\evalt{\MT}{\Cx_e}}$.
      Since $\plugCx{\Cx}{\evalt{\MT}{\plugx{\Xx}{\e}}} = \plugCx{\Cx''}{\rdx}$
      for $\Cx'' = \plugCx{\Cx}{\evalt{\MT}{\plugx{\Xx}{\Cx_e}}}$,
      \plugCx{\Cx}{\evalt{\MT}{\e}} and \plugCx{\Cx}{\evalt{\MT}{\plugx{\Xx}{\e}}}
      step similarly by \lemref{lem:paper:context-irrelevant}.
  \end{itemize}
\end{proof}

\begin{lemma}[Optimization Preserves Values]\label{lem:paper:optimization-preserves-values}
  For all $\SpecEnv, \MT, \MT', \Gm, \e, \v$, the following hold:
  \[ \exproptd{\v}{\e} \quad\implies\quad \e = \v
      \qquad\text{ and }\qquad
     \exproptd{\e}{\v} \quad\implies\quad \e = \v.\]
\end{lemma}
\begin{proof}
  By case analysis on the optimization relation.
\end{proof}

% -----

\begin{figure}
  \footnotesize
  \[
  \begin{array}{rcll}
    \gm & ::= & \obar{\x \mapsto \v} \text{ where } \forall i,j. \x_i \neq \x_j
              & \text{\emph{Value substitution}}
  \end{array}
  \]
  \begin{mathpar}
  \inferrule*[right=\gm-Ok]
    { \dom(\Gm) = \dom(\gm) \\
      \forall \x \in \dom(\gm).\
        \left(\Gm(\x) = \g \iff \typeof(\gm(\x)) = \g\right) }
    { \substokd }
  \end{mathpar}
  \caption{Value substitution}\label{fig:wa-value-subst}
\end{figure}

% -----

\begin{lemma}[Value Substitution Preserves Optimization]
  \label{lem:paper:subst-preserves-optim}
  For all $\SpecEnv, \Gm, \e, \MT, \e', \MT', \gm$, such that
  $\forall \v \in \gm.\ \tableoptexprd{\v}$,
  the following holds:
  \[
    \left(\exproptd{\e}{\e'}\ \land\ \substokd \right)
    \quad\implies\quad
    \expropt{\SpecEnv}{}{\MT}{\gm(\e)}{\MT'}{\gm(\e')}.
  \]
\end{lemma}
\begin{proof}
  By induction on the derivation of \exproptd{\e}{\e'}.
  \APPENDICITIS{}{The full proof is given
  in \appref{app:proofs:value-subst}, on page \pageref{proof:subst-preserves-optim}.}
\end{proof}

\begin{lemma}[Bisimulation]\label{lem:paper:expr-optimization-bisimulation}
For all method tables $\MT, \MT'$, method-optimization environment $\SpecEnv$,\\
and expressions $\e_1$, ${\e_1}\!'$, such that
\[ \tableopt{\SpecEnv}{\e_1}{\MT}{\MT'} \quad \text{and} \quad
   \expropt{\SpecEnv}{}{\MT}{\e_1}{\MT'}{\e_1'}, \]
for all global tables $\MT_g, {\MT_g}\!'$ and world context \Cx,
the following hold:
\begin{enumerate}
  \item Forward direction:
    \[
    \begin{array}{rl}
      \forall \e_2. &
        \evalwa{\MT_g}{\plugCx\Cx{\evalt{\MT}{\e_1}}}
               {\MT_g'}{\plugCx\Cx{\evalt{\MT}{\e_2}}} \\
      & \implies \\
      & \exists \e_2'.\
        \evalwa{\MT_g}{\plugCx\Cx{\evalt{\MT'}{\e_1'}}}
               {\MT_g'}{\plugCx\Cx{\evalt{\MT'}{\e_2'}}}
        \ \ \ \land\ \
        \expropt{\SpecEnv}{}{\MT}{\e_2}{\MT'}{\e_2'}.
    \end{array}
    \]
  \item Backward direction:
    \[
    \begin{array}{rl}
      \forall \e_2'. &
        \evalwa{\MT_g}{\plugCx\Cx{\evalt{\MT'}{\e_1'}}}
               {\MT_g'}{\plugCx\Cx{\evalt{\MT'}{\e_2'}}} \\
      & \implies \\
      & \exists \e_2.\
        \evalwa{\MT_g}{\plugCx\Cx{\evalt{\MT}{\e_1}}}
               {\MT_g'}{\plugCx\Cx{\evalt{\MT}{\e_2}}}
        \ \ \ \land\ \
        \expropt{\SpecEnv}{}{\MT}{\e_2}{\MT'}{\e_2'}.
    \end{array}
    \]
\end{enumerate}
\end{lemma}
\begin{proof}
  The proof goes by induction on the derivation of
  optimization \expropt{\SpecEnv}{}{\MT}{\e_1}{\MT'}{\e_1'}.
  \APPENDICITIS{}{The full proof is given
  in~\appref{app:proofs:opt-correctness}, page~\pageref{proof:optimization-correct}.}
  For each case, both directions are proved by analyzing
  possible normal-evaluation steps. More specifically, the forward-direction
  proof strategy is as follows (the backward direction is similar):
  \begin{enumerate}
    \item Observe that to make the required step, $\e_1$ should have a
      certain representation. %For example, if $\e_1$ is a sequence,
      %to step within \evalt{\MT}{\cdot} context, it has to be either some
      %\plugx{\Xx}{\mcalld} or \plugx{\Cx}{\rdx}.
    %\item 
      Consider all possible representations that satisfy this
      requirement. %For instance, a sequence \seq{\e_{11}}{\e_{12}} is
      %\plugx{\Xx}{\mcalld} only if $\e_{11}$ is \plugx{\Xx_1}{\mcalld}.
      %For \plugx{\Cx}{\rdx}, $\e_{11}$ has to be either $\v_{11}$
      %or \plugx{\Cx_1}{\rdx}.
    \item For each representation, analyze the suitable
      normal-evaluation rule (recall that the semantics is deterministic,
      so there will be just one such rule).
    \item If $\e_1$ represents an immediate redex
      (e.g. \seq{\v_{11}}{\e_{12}}), the optimized expression will be an
      immediate redex too (possibly, of a different form).  Otherwise, use
      induction hypothesis and auxiliary facts about contexts and evaluation
      to show that the optimized expression steps in a similar fashion,
      in particular, facts from \lemref{lem:paper:context-irrelevant}
      and \lemref{lem:paper:simple-context-irrelevant}.
    \item Finally, show that the resulting expressions are in the optimization
      relation. This will follow from the assumptions and induction.
  \end{enumerate}
  As an example, consider the proof of the forward direction for
  the sequence case \textbf{\WAOE{Seq}}. %with $\e_1 = \seq{\e_{11}}{\e_{12}}$
  %where $\e_{11} = \plugCx{\Cx_1}{\rdx}$ goes like this.
  By assumption, we have $\e_1 = \seq{\e_{11}}{\e_{12}}$ and
  $\e_1' = \seq{\e'_{11}}{\e'_{12}}$ where
  \begin{mathpar}
    \inferrule*[right=\WAOE{Seq}]
    { \expropte{\e_{11}}{\e'_{11}} \\ \exproptd{\e_{12}}{\e'_{12}} }
    { \expropte{\seq{\e_{11}}{\e_{12}}}{\seq{\e'_{11}}{\e'_{12}}} }.
  \end{mathpar}
  For \plugCx\Cx{\evalt{\MT}{\seq{\e_{11}}{\e_{12}}}} to reduce,
  by case analysis, we know there are three possibilities.
  \begin{enumerate}
    \item $\e_{11} = \v_{11}$ and
      $\evalwa{\MTg}{\plugCx\Cx{\evalt{\MT}{\seq{\v_{11}}{\e_{12}}}}}
              {\MTg}{\plugCx\Cx{\evalt{\MT}{\e_{12}}}}$
      by rule \WAE{Seq}. Then by \lemref{lem:paper:optimization-preserves-values},
      $\e'_{11} = \v_{11}$ and the optimized
      expression steps by the same rule:
      \[ \evalwa{\MTg}{\plugCx\Cx{\evalt{\MT'}{\seq{\v_{11}}{\e'_{12}}}}}
                {\MTg}{\plugCx\Cx{\evalt{\MT'}{\e'_{12}}}}. \]
      The desired optimization relation holds by one of the assumptions:
      \exproptd{\e_{12}}{\e'_{12}}.
    \item $\e_{11} = \plugx{\Xx_1}{\mcalld}$ and the original expression
      steps by \WAE{CallLocal}:
      \[
        \evalwa
        {\MT_g}{\plugCx\Cx{\evalt{\MT}{\seq{\plugx{\Xx_1}{\mcalld}}{\e_{12}}}}}
        {\MT_g}{\plugCx\Cx{\evalt{\MT}{\seq{\plugx{\Xx_1}{\e_b\subst{\xs}{\vs}}}{\e_{12}}}}}.
      \]
      Since \plugCx\Cx{\evalt{\MT}{\plugx{\Xx_1}{\mcalld}}} reduces similarly,
      by the induction hypothesis, $\exists\ \e'_{21}$ such that
      \[
        \evalwa
          {\MT_g}{\plugCx\Cx{\evalt{\MT'}{\e'_{11}}}}
          {\MT_g}{\plugCx\Cx{\evalt{\MT'}{\e'_{21}}}}
          \quad\text{ and }\quad
          \exproptd{\plugx{\Xx_1}{\e_b\subst{\xs}{\vs}}}{\e'_{21}}.
      \]
      But then, by~\lemref{lem:paper:simple-context-irrelevant}, the entire
      optimized expression \plugCx\Cx{\evalt{\MT}{\seq{\e'_{11}}{\e'_{12}}}}
      steps too, and the desired optimization relation holds:
      \begin{mathpar}
        \inferrule*[right=\WAOE{Seq}]
        { \expropt{\SpecEnv}{}{\MT}{\plugx{\Xx_1}{\e_b\subst{\xs}{\vs}}}{\MT'}{\e'_{21}} \\
          \expropt{\SpecEnv}{}{\MT}{\e_{12}}{\MT'}{\e'_{12}} }
        { \expropt{\SpecEnv}{}{\MT}{\seq{\plugx{\Xx_1}{\e_b\subst{\xs}{\vs}}}{\e_{12}}}
                              {\MT'}{\seq{\e'_{21}}{\e'_{12}}} }.
      \end{mathpar}
    \item $\e_{11} = \plugCx{\Cx_1}{\rdx}$ and
      by~\lemref{lem:paper:context-irrelevant}:
      \[
      \begin{array}{c}
      \evalwa{\MTg}{\plugCx\Cx{\evalt{\MT}{\seq{\plugCx{\Cx_1}{\rdx}}{\e_{12}}}}}
              {\MTg'}{\plugCx\Cx{\evalt{\MT}{\seq{\plugCx{\Cx_1}{\e'}}{\e_{12}}}}}
      \\ \iff
      \\ \evalwa{\MTg}{\plugCx\Cx{\evalt{\MT}{\plugCx{\Cx_1}{\rdx}}}}
                {\MTg'}{\plugCx\Cx{\evalt{\MT}{\plugCx{\Cx_1}{\e'}}}}.
      \end{array}
      \]
      Since \plugCx\Cx{\evalt{\MT}{\plugCx{\Cx_1}{\rdx}}} reduces,
      by the induction hypothesis, $\exists\ \e'_{21}$ such that
      \[
        \evalwa
          {\MT_g}{\plugCx\Cx{\evalt{\MT'}{\e'_{11}}}}
          {\MT_g}{\plugCx\Cx{\evalt{\MT'}{\e'_{21}}}}
          \quad\text{ and }\quad
          \exproptd{\plugCx{\Cx_1}{\e'}}{\e'_{21}}.
      \]
      Similarly to the previous case, the entire 
      \plugCx\Cx{\evalt{\MT}{\seq{\e'_{11}}{\e'_{12}}}} steps,
      and the desired optimization relation holds.
  \end{enumerate}
\end{proof}

\begin{lemma}[Reflexivity of Optimization]\label{lem:paper:optimization-reflexive}
  For all $\MT, \MT', \SpecEnv, \Gm, \e$, the following holds:
  \[
    \tableoptref{\SpecEnv}{\e}{\MT}{\MT'}
    \quad\implies\quad
    \expropt{\SpecEnv}{\Gm}{\MT}{\e}{\MT'}{\e}. 
  \]
\end{lemma}
\begin{proof}
  By induction on \e.
  The only interesting cases are \m, \md, \evalg{\e'}, and \evalt{\MT_l}{\e'}.
  For example, consider the case of \m (others are similar).
  Rule \WAOE{ValFun} requires a method named \m to either exist in both tables
  or do not appear in either (this rules out the case
  where \evalt{\MT}{\m()} would err but \evalt{\MT'}{\m()} succeed).
  This requirement is guaranteed by the assumption that
  \tableoptref{\SpecEnv}{\e}{\MT}{\MT'}, which by inversion,
  gives the necessary \tableoptexprd{\e}.
\end{proof}

The main result, \thmref{thm:paper:table-opt}, is a corollary of
\lemref{lem:paper:expr-optimization-bisimulation}.  It states that a fixed-table expression can be
soundly evaluated in an optimized table.

\begin{theorem}[Correctness of Table Optimization]\label{thm:paper:table-opt}
For all $\MT, \MT', \SpecEnv, e$
satisfying $\tableoptref{\SpecEnv}{\e}{\MT}{\MT'}$,
for all $\MT_g,\MT_g',\Cx,\v,$ the following holds:
\[ \evalwac{\MT_g}{\plugCx\Cx{\evalt{\MT}{\e}}}{\MT_g'}{\v}
   \iff
   \evalwac{\MT_g}{\plugCx\Cx{\evalt{\MT'}{\e}}}{\MT_g'}{\v}. \]
\end{theorem}
\begin{proof}
First of all, note that \expropt{\SpecEnv}{}{\MT}{\e}{\MT'}{\e}
by \lemref{lem:paper:optimization-reflexive},
and that \tableopt{\SpecEnv}{\e}{\MT}{\MT'} follows from
\tableoptref{\SpecEnv}{\e}{\MT}{\MT'}.
Then, we proceed by induction on $\evalstep^*$ (reflexive-transitive closure
of normal evaluation). In the interesting case of the forward direction, we have:
\begin{mathpar}
\inferrule[]
  { \evalwa{\MT_g}{\plugCx\Cx{\evalt{\MT}{\e}}}
           {\MT_g''}{\plugCx\Cx{\evalt{\MT}{\e'_1}}} \\
    \evalwac{\MT_g''}{\plugCx\Cx{\evalt{\MT}{\e'_1}}}{\MT_g'}{\v} }
  { \evalwac{\MT_g}{\plugCx\Cx{\evalt{\MT}{\e}}}{\MT_g'}{\v} }.
\end{mathpar}
By applying \lemref{lem:paper:expr-optimization-bisimulation}
to the first premise, we get:
\[ \evalwa{\MT_g}{\plugCx\Cx{\evalt{\MT'}{\e}}}
          {\MT_g''}{\plugCx\Cx{\evalt{\MT'}{\e'_2}}}
   \quad \text{and} \quad
   \expropt{\SpecEnv}{}{\MT}{\e'_1}{\MT'}{\e'_2}. \]
By applying the induction hypothesis to the second premise, we get:
\[ \evalwac{\MT_g''}{\plugCx\Cx{\evalt{\MT'}{\e'_2}}}{\MT_g'}{\v}. \]
By combining the results, we can get the desired derivation:
\begin{mathpar}
\inferrule[]
  { \evalwa{\MT_g}{\plugCx\Cx{\evalt{\MT'}{\e}}}
           {\MT_g''}{\plugCx\Cx{\evalt{\MT'}{\e'_2}}} \\
    \evalwac{\MT_g''}{\plugCx\Cx{\evalt{\MT'}{\e'_2}}}{\MT_g'}{\v} }
  { \evalwac{\MT_g}{\plugCx\Cx{\evalt{\MT'}{\e}}}{\MT_g'}{\v} }.
\end{mathpar}
The backward direction proceeds similarly.
\end{proof}

%\medskip
%\noindent
\thmref{thm:paper:table-opt}, in particular, justifies Julia's choice to
execute top-level calls using optimized methods.
Once a top-level call \stwa{\MT}{\evalg{\mcall{\m}{\vs}}}
steps to a fixed-table call \stwa{\MT}{\evalt{\MT}{\mcall{\m}{\vs}}},
it is sound to optimize table \MT into $\MT'$
(using inlining, direct calls, and specialization),
and evaluate the call in the optimized table
\stwa{\MT}{\evalt{\MT'}{\mcall{\m}{\vs}}}.

\subsection{Testing the Semantics}\label{sec:validation}
%% =============================================================================

To check if \juliette behaves as we expect, we implemented it in
Redex~\cite{bib:redex-book} and ran it along with Julia
on a small set of 9 litmus tests (provided in~\appref{app:litmus});
Julia agrees with \juliette on all of them.
The tests cover the intersection of the semantics of \juliette and Julia,
and demonstrate the interaction of \eval, method definitions, and method calls.
In particular, the litmus tests ensure: that the executing semantics
prohibits calls to too-new methods, that this restriction can be skipped
with \eval or \invokelatest, and that the semantics of \eval executes
successive statements in the latest age.

\begin{wrapfigure}{r}{8.5cm}
\begin{minipage}{4cm}\begin{lstlisting}
r2() = r1()
m()  = (
  eval(:(r1() = 2));
  r2())
m() # error
\end{lstlisting}\end{minipage}
\begin{minipage}{3mm}
  \hspace{3mm}
\end{minipage}
\begin{minipage}{4cm}\begin{lstlisting}
r2() = r1()
m()  = (
  eval(:(r1() = 2));
  Base.invokelatest(r2))
m() == 2 # passes
\end{lstlisting}\end{minipage}
\caption{Litmus tests}\label{fig:litmus-example}
\end{wrapfigure}

Two of the litmus tests are shown in~\figref{fig:litmus-example};
each test is made up of a small program and its expected output.
The tests examine the case where a method \c{r2} is placed ``in between'' the
generated method \c{r1} and an older \c{m}. In the first test, \c{m} errs.
While \c{r2} is callable from the age that \c{m} was called in, \c{r1}
is not.  In the second test, we use \invokelatest to execute \c{r2} in the
latest world age; this allows the invocation of the dynamically
generated \c{r1}.

To use the litmus tests, we need to (1) translate them
from Julia into our grammar and (2)~implement the semantics of \juliette
into an executable form.  The former is done by translating ASTs.  The
latter is realized with a Redex mechanization, which is publicly available
on GitHub\footnote{\url{https://github.com/julbinb/juliette-wa}}
along with the litmus tests. The model implements the calculus almost literally.
Values, tags, and type annotations are instantiated with several concrete
examples, such as numbers and strings.  Primitive operations include
arithmetic and \c{print}.  The only difference between the paper and Redex
is handling of function names. Similar to Julia, in the Redex model, a
definition of the method named \c f introduces a global constant $f$.  When
referenced, the constant evaluates to a function value $\mval{f}$.  Thus,
instead of a single error evaluation rule \WAE{Var} from
\figref{semantics}, the Redex model has the following
two rules, one for normal evaluation and one for erroneous evaluation:

%\vspace{0.5em}
\begin{mathpar}
  \inferrule[\WAE{VarMethod}]
  { x \in \dom(\MT) }
  { \evalwad{\plugCx{\Cx}{\x}}{\plugCx{\Cx}{\mval{x}}} }

  \inferrule[\WAE{VarErr}]
  { x \notin \dom(\MT) }
  { \evalerrwad{\plugCx{\Cx}{\x}}{\err}. }
\end{mathpar}
%\vspace{2mm}

\noindent
The new rules treat the global method table as a global environment:
\WAE{VarMethod} evaluates a global variable to its underlying function value,
and \WAE{VarErr} errs if a variable is not found in the global
environment; all local variables should be eliminated by substitution.  All
paper-style programs can be written in the Redex model, and the extension
makes it easier to compare and translate Julia programs to corresponding
Redex programs.  Thus, the litmus test on the left of
\figref{fig:litmus-example} translates to the Redex model as follows (the
grammar is written in S-expressions style):
\begin{lstlisting}
(evalg (seq (seq
  (mdef "r2" () (mcall r1)) # r2() = r1()
  (mdef "m"  () (seq
                  (evalg (mdef "r1" () 2)) # eval(:(r1() = 2))
                  (mcall r2))))            # r2()
  (mcall m))) # m()
\end{lstlisting}
The Redex model also implements the optimization judgments presented
in~\secref{subsec:optimizations}, as well as a straightforward optimization
algorithm that is checked against the judgments.
\APPENDICITIS{}{The definition of the algorithm is provided in~\appref{app:opt-algo}.}

Discussion with Julia's developers confirmed that our understanding of world
age is correct, and that the table-based semantics has a correspondence to
the age-based implementation. Namely, it is possible to generate \juliette
method tables %, which are being fixed for top-level calls,
from the global data
structure used by Julia to store methods.

\section{Conclusion} %==========================================================

Julia's approach to dynamic code loading is distinct; instead of striving to
achieve performance \emph{in spite of} the language's semantics, the
designers of Julia chose to restrict expressiveness so that they could keep
their compiler simple \emph{and} generate fast code.  World age aligns Julia's
dynamic semantics with its just-in-time compiler's static
approximation. As a result, statically resolved function calls have the same
behavior as dynamic invocations.

This equivalence---that statically and dynamically resolved methods behave
the same---allows Julia to forsake some of the complexity
of modern compilers.  Instead of needing deoptimization to handle newly
added definitions, Julia simply does not allow running code to see those
definitions. Thus, optimizations can rely on the results of static reasoning
about the method table, while remaining sound in the presence of \eval.
%Thus, static reasoning for optimization can use information
%about the global method table while remaining sound in the presence of \eval.
If necessary, the programmer can explicitly ask for newly defined methods,
making the performance penalty explicit and user-controllable.

World age need not be limited to Julia. Any language that supports updating
existing function definitions may benefit from such a mechanism, namely
control over when those new definitions can be observed and when function
calls can be optimized.
From Java to languages like R, having a
%a way to control when those
%new definitions can be observed. From Java to languages like R, having a
clear semantics for updating code, especially in the presence of
concurrency, can be beneficial, as it would improve our ability to reason about
programs written in those languages.

Although the world-age semantics presented in the paper follows Julia,
\emph{a} world-age semantics does not have to. For instance,
%However, some of the choices made by the Julia
%designers may have to be revisited: for example, instead of fixing the world
%age at every top-level call, a language can have an explicit construct for that.
an alternative world-age semantics could pick another point when the age
counter is incremented. The notion of top level makes sense in the context
of an interactive development environment, but is unclear in, for example, a
web server that may receive new code to install from time to time. Such a
continuously running system may need a definition of quiescence that is
different from the top-level used in Julia. One alternative is to provide
an explicit \c{freeze} construct that allows programmers to opt-in to the
world-age system. This would allow existing languages to incorporate world
age without affecting existing code.

The calculus we present here is a basic foundation intended to capture the
operation of world age. Future work may build on this to formalize the
semantics of Julia as a whole, but, notably, the additional semantics
%required for the rest of Julia
will not impact the world-age mechanism itself. Of particular note is
mutable state: it is orthogonal to world age because Julia decouples code state from
data state by design. This was a pragmatic decision, as the compiler
depends on knowing the contents of the method table for its optimization.
Optimizations based on global variables are much less frequent.

\begin{acks}
We thank the anonymous reviewers for their insightful comments and
suggestions to improve this paper.  This work was supported by
\grantsponsor{ONR}{Office of Naval Research (ONR)}{} award
\grantnum{ONR}{503353}, the \grantsponsor{NSF}{National Science
  Foundation}{} awards \grantnum{NSF}{1759736},
\grantnum{NSF}{1925644} and \grantnum{NSF}{1618732}, the
\grantsponsor{BC}{Czech Ministry of Education from the Czech
  Operational Programme Research, Development, and Education}{}, under
grant agreement No.
\grantnum{BC}{CZ.02.1.01/0.0/0.0/15\_003/0000421}, and the
\grantsponsor{ELE}{European Research Council under the European
  Union's Horizon 2020 research and innovation programme}{}, under
grant agreement No.  \grantnum{ELE}{695412}.
\end{acks}

%% =============================================================================
\appendix

%\clearpage
\section{Litmus Tests}\label{app:litmus}
%% --------------------------------------------------
As a basic test of functionality, we provide 9 litmus tests shown
in~\figref{fig:litmus-tests}, written in Julia, that exercise the basic world
age semantics as well as key Julia semantics surrounding world age. The tests
suffice to identify the following semantic characteristics:
\begin{enumerate}[label=(\alph*)]
\item too-new methods cannot be called using a normal invocation;
\item \invokelatest uses the latest world age;
\item \eval uses the latest world age;
\item successive \eval statements run in the latest world age;
\item only age at the top-level is relevant for invocation visibility;
\item ``latest'' calls propagate the new world age;
\item \eval executes in the top-level scope;
\item normal invocation uses overridden methods if added method too new;
\item \eval will use latest definition of an overridden method.
\end{enumerate}

\begin{figure}[h]
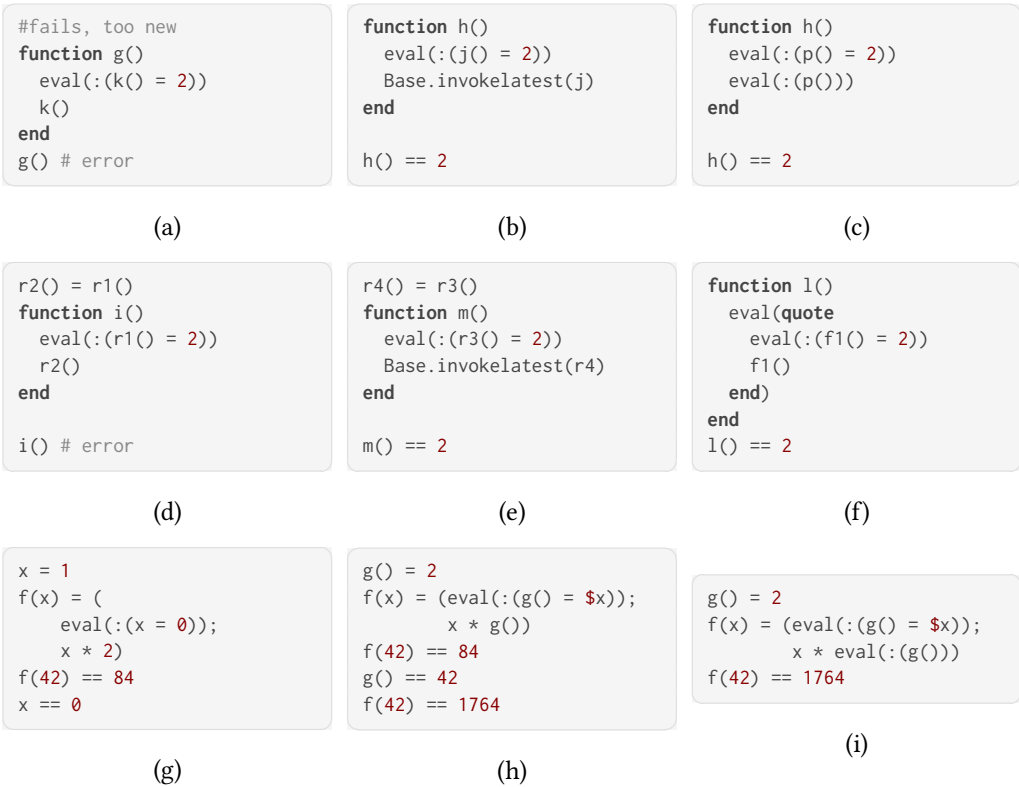

\begin{minipage}{4.5cm}
\begin{lstlisting}
#fails, too new
function g()
  eval(:(k() = 2))
  k()
end
g() # error
\end{lstlisting}
\centering
(a)
\end{minipage}
\begin{minipage}{4.5cm}
\begin{lstlisting}
function h()
  eval(:(j() = 2))
  Base.invokelatest(j)
end

h() == 2
\end{lstlisting}
\centering
(b)
\end{minipage}
\begin{minipage}{4.5cm}
\begin{lstlisting}
function h()
  eval(:(p() = 2))
  eval(:(p()))
end

h() == 2
\end{lstlisting}
\centering
(c)
\end{minipage}

\begin{minipage}{4.5cm}
\begin{lstlisting}
r2() = r1()
function i()
  eval(:(r1() = 2))
  r2()
end

i() # error
\end{lstlisting}
\centering
(d)
\end{minipage}
\begin{minipage}{4.5cm}
\begin{lstlisting}
r4() = r3()
function m()
  eval(:(r3() = 2))
  Base.invokelatest(r4)
end

m() == 2
\end{lstlisting}
\centering
(e)
\end{minipage}
\begin{minipage}{4.5cm}
\begin{lstlisting}
function l()
  eval(quote
    eval(:(f1() = 2))
    f1()
  end)
end
l() == 2
\end{lstlisting}
\centering
(f)
\end{minipage}
\begin{minipage}{4.5cm}
\begin{lstlisting}
x = 1
f(x) = (
    eval(:(x = 0));
    x * 2)
f(42) == 84
x == 0
\end{lstlisting}
\centering
(g)
\end{minipage}
\begin{minipage}{4.5cm}
\begin{lstlisting}
g() = 2
f(x) = (eval(:(g() = $x));
        x * g())
f(42) == 84
g() == 42
f(42) == 1764
\end{lstlisting}
\centering
(h)
\end{minipage}
\begin{minipage}{4.5cm}
\begin{lstlisting}
g() = 2
f(x) = (eval(:(g() = $x));
        x * eval(:(g()))
f(42) == 1764
\end{lstlisting}
\centering
(i)
\end{minipage}
\caption{Litmus Tests}
\label{fig:litmus-tests}
\end{figure}

%\end{document} %%% While editing, hide this, reveal if needed.

\APPENDICITIS{}{
  %\clearpage
  \section{Proofs}\label{app:proofs}
%% ===================================================================

In this section, we provide detailed proofs for statements
from \secref{subsec:properties} and \secref{subsec:correctness}:
progress and determinism of \juliette semantics,
and correctness of the optimizations.
The proofs rely on a number of auxiliary lemmas and definitions
related to contexts and the form of expressions.

For convenience, all statements from the main text are reproduced here.
Definitions of the semantics and optimization relations can be found
in the main text: \secref{subsec:wa-semantics} defines the semantics,
\secref{subsec:optimizations} gives an informal account of optimizations,
and \secref{subsec:correctness} defines the optimizations formally.

\subsection{Preliminary}
%% --------------------------------------------------

In what follows, we use a new value form for non-functional values,
$\vnm = \v \setminus \m$. %The updated syntax of \juliette is shown
%in \figref{fig:wa-calculus-syntax-upd}.
% \begin{itemize}
%   \item First, we extend expressions \e with evaluation in a method table
%     $\evalt{\MT}{\e}$ (earlier, this syntax form was a part of the context \Cx
%     but not the surface language expressions \e).
%   \item Second, we introduce a new value form for non-functional values
%     $\vnm = \v \setminus \m$.
% \end{itemize}

The following facts will be used implicitly, without a reference:
\begin{eqnarray}
  \forall \Xx_1.\forall \Xx_2. \exists \Xx. & \plugx{\Xx_1}{\Xx_2} & = \quad\Xx\\
  \forall \Xx.\forall \Cx'. \exists \Cx.    & \plugx{\Xx}{\Cx'} & = \quad\Cx \\
  \forall \Cx_1.\forall \Cx_2. \exists \Cx. & \plugx{\Cx_1}{\Cx_2} & = \quad\Cx
\end{eqnarray}
The first two are proved by induction on \Xx,
and the last one is proved by induction on \Cx.

% \\ \e & ::= & & \text{\emph{Internal Expression}}
% \\ &\Altf& \faded{\v} & \text{\faded{value}}
% \\ &\Altf& \faded{\ldots} &
% \\ &\Altf& \evalg{\faded{\e}} & \text{\faded{global evaluation}}
% \\ &\Alt& \evalt{\MT}{\e} & \text{evaluation in table \MT}
% \\
% \\ \faded{\p} & \faded{::=} & \evalg{\faded{\e}} & \text{\emph{\faded{Program}}}
% \\
% \begin{figure}
%   \[
%   \begin{array}{ccl@{\qquad}l}
%       \\ \iexp & := & \e \Alt \err & \text{\emph{Expression or Error}}
%       \\
%       \\ \vnm & := & \faded{\ldots} & \text{\emph{Simple Value}}
%       \\ &\Altf& \faded{\skp} & \text{\color{gray}unit value}
%       \\
%       \\ \v & := & & \text{\emph{Value}}
%       \\ &\Alt& \vnm & \text{simple value}
%       \\ &\Altf& {\color{gray}\m} & \text{\color{gray}generic function}
%   \end{array}
%   \]
%   \caption{Updated \juliette syntax}\label{fig:wa-calculus-syntax-upd}
% \end{figure}
% \begin{mathpar}
% \\
% \inferrule*[right=\WAE{Normal}]
%   { \evalwa{\MT}{\plugCx{\Cx}{\e}}{\MT'}{\plugCx{\Cx}{\e'}} }
%   { \evalfullwa{\MT}{\plugCx{\Cx}{\e}}{\MT'}{\plugCx{\Cx}{\e'}} }
% \hspace{1cm}
% \inferrule*[right=\WAE{Error}]
%   { \evalerrwa{\MT}{\plugCx{\Cx}{\e}}{\err} }
%   { \evalfullwa{\MT}{\plugCx{\Cx}{\e}}{\MT}{\err} }
% \end{mathpar}

\subsection{Progress of \juliette Programs}\label{app:proofs:progress}
%% --------------------------------------------------

To prove the \textbf{Progress \thmref{thm:paper:progress}},
we will need several auxiliary lemmas below.
Roughly, the idea is to show that any \juliette program can be represented in a
certain way (\lemref{lem:expr-form}) that allows for concluding progress
(\lemref{lem:redex-steps}).

The desired representation of expressions closely follows normal-
and error-evaluation rules:
it relies on a subset of expressions \rdx, which we call redex bases.
A redex base gives rise to an ``interesting'' reduction,
in the spirit of $\beta$-reduction from the
traditional context-based semantics of lambda-calculus.
Redex bases are defined in \figref{fig:wa-calculus-redex},
which reproduces \figref{fig:wa-rdx-main} and additionally hints
which part of the semantics every redex base corresponds to
(normal, error, or both).

\begin{figure}
  \[
  \begin{array}{ccl@{\qquad}ll}
      \\ \rdx & ::= & & \text{\emph{Redex Base}} &
      \\ &\Alt& \x & \text{variable} & \text{(error)}
      \\ &\Alt& \seq{\v}{\e} & \text{sequencing} & \text{(normal)}
      \\ &\Alt& \primcalld{\vs} & \text{primop call} & \text{(normal/error)}
      \\ &\Alt& \mcall{\vnm}{\vs} & \text{non-function call} & \text{(error)}
      \\ &\Alt& \md & \text{method definition} & \text{(normal)}
      \\ &\Alt& \evalg{\v} & \text{value in global context} & \text{(normal)}
      \\ &\Alt& \evalt{\MT}{\v} & \text{value in table context} & \text{(normal)}
      \\ &\Alt& \evalg{\plugx{\Xx}{\mcall{\m}{\vs}}} 
            & \text{function call in global context} & \text{(normal)}
      \\ &\Alt& \evalt{\MT}{\plugx{\Xx}{\mcall{\m}{\vs}}} 
            & \text{function call in table context} & \text{(normal/error)}
  \end{array}
  \]
  \caption{Redex Bases}\label{fig:wa-calculus-redex}
\end{figure}

\begin{lemma}[Form of Expression]\label{lem:expr-form}
  For any expression \e, one of the following holds:
  \begin{enumerate}[label=(\alph*)]
    \item $\e = \v$; or
    \item $\e = \plugx{\Xx}{\mcall{\m}{\vs}}$; or
    \item $\e = \plugCx{\Cx}{\rdx}$.
  \end{enumerate}
\end{lemma}
\begin{proof}
By induction on the structure of \e.
\begin{description}
  \item[\x] This is (c): $\e = \plugCx{\hole}{\x}$.
  \item[\v] This is (a): $\e = \v$.
  \item[\seq{\e_1}{\e_2}] by the induction hypothesis,
    $\e_1$ is one of the following:
    \begin{description}
      \item[$\v_{11}$] Then $\e = \plugCx{\hole}{\seq{\v_{11}}{\e_2}}$,
        which gives us (c).
      \item[$\plugx{\Xx_1}{\mcall{\m_1}{\vs_1}}$] Then
        $\e = \seq{(\plugx{\Xx_1}{\mcall{\m_1}{\vs_1}})}{\e_2}
            = \plugx{(\seq{\Xx_1}{\e_2})}{\mcall{\m_1}{\vs_1}}
            = \plugx{\Xx}{\mcall{\m_1}{\vs_1}}$ for $\Xx = (\seq{\Xx_1}{\e_2})$,
        which gives us (b).
      \item[$\plugCx{\Cx_1}{\rdx_1}$] Then
        $\e = \seq{(\plugCx{\Cx_1}{\rdx_1})}{\e_2}
            = \plugCx{(\seq{\Cx_1}{\e_2})}{\rdx_1} = \plugCx{\Cx}{\rdx_1}$
            for $\Cx = (\seq{\Cx_1}{\e_2})$,
        which gives us (c).
    \end{description}
  \item[$\primcalld{\es}$] Reasoning similarly to $\seq{\e_1}{\e_2}$,
    by induction hypotheses, several cases are possible:
    \begin{itemize}
      \item $\e = \primcalld{\vs} = \plugCx{\hole}{\primcalld{\vs}}$,
        which is case (c).
      \item $\e = \primcalld{\vs,\ \plugx{\Xx_i}{\mcall{\m_i}{\vs_i}},\ \es}
                = \plugx{\primcalld{\vs\ \Xx_i\ \es}}{\mcall{\m_i}{\vs_i}}
                = \plugx{\Xx}{\mcall{\m_i}{\vs_i}}$
        for $\Xx = \primcalld{\vs\ \Xx_i\ \es}$, i.e. case (b).
      \item $\e = \primcalld{\vs,\ \plugCx{\Cx_i}{\rdx_i},\ \es}
                = \plugCx{\primcalld{\vs\ \Cx_i\ \es}}{\rdx_i}
                = \plugCx{\Cx}{\rdx_i}$
        for $\Cx = \primcalld{\vs\ \Cx_i\ \es}$, i.e. case (c).
    \end{itemize}
  \item[$\mcall{\e_c}{\es}$] Similarly to the above, several cases are possible:
    \begin{itemize}
      \item $\e = \mcall{\v_c}{\vs}$
        \begin{itemize}
          \item $\v_c = \vnm_c$ gives case (c),
            for $\e = \plugCx{\hole}{\mcall{\vnm_c}{\vs}}$.
          \item $\v_c = \m$ gives case (b),
            for $\e = \plugx{\hole}{\mcall{\m}{\vs}}$.
        \end{itemize}
      \item $\e = \mcall{(\plugx{\Xx_c}{\mcall{\m_c}{\vs_c}})}{\es}$ or
        $\e = \mcall{\v_c}{\vs,\ \plugx{\Xx_i}{\mcall{\m_i}{\vs_i}},\ \es}$
        give case (b) analogously to the case of
        $\e = \primcalld{\vs,\ \plugx{\Xx_i}{\mcall{\m_i}{\vs_i}},\ \es}$.
      \item $\e = \mcall{(\plugCx{\Cx_c}{\rdx_c})}{\es}$ or
        $\e = \mcall{\v_c}{\vs,\ \plugCx{\Cx_i}{\rdx_i},\ \es}$
        give case (c) analogously to the case of
        $\e = \primcalld{\vs,\ \plugCx{\Cx_i}{\rdx_i},\ \es}$.
    \end{itemize}
  \item[$\md$] This is (c): $\e = \plugCx{\hole}{\md}$.
  \item[$\evalg{\e'}$] by the induction hypothesis, $\e'$ is one of the following:
  \begin{description}
    \item[$\v'$] Then $\e = \evalg{\v'} = \plugCx{\hole}{\evalg{\v'}}$,
      which gives us (c).
    \item[$\plugx{\Xx'}{\mcall{\m'}{\vs'}}$] Then
      $\e = \evalg{\plugx{\Xx'}{\mcall{\m'}{\vs'}}}
          = \plugCx{\hole}{\evalg{\plugx{\Xx'}{\mcall{\m'}{\vs'}}}}$,
      which gives us (c).
    \item[$\plugCx{\Cx'}{\rdx'}$] Then
      $\e = \evalg{\plugCx{\Cx'}{\rdx'}} = \plugCx{\Cx}{\rdx'}$
          for $\Cx = \plugx{\hole}{\evalg{\Cx'}}$,
      which gives us (c).
  \end{description}
  \item[$\evalt{\MT}{\e'}$] Reasoning similarly to $\evalg{\e'}$,
    we have case (c).
\end{description}
\end{proof}

Together, world evaluation context \Cx and redex base \rdx make up a redex,
i.e. expression $\plugCx{\Cx}{\rdx}$ that can make a step. 

\begin{lemma}[Redex Steps]\label{lem:redex-steps}
  For any redex base \rdx, world context \Cx, and global table \MTg,
  \begin{enumerate}[label=(\alph*)]
    \item either\ $\evalwa{\MTg}{\plugCx{\Cx}{\rdx}}{\MTg'}{\plugCx{\Cx}{\e'}}$,
    \item or\qquad $\evalerrwa{\MTg}{\plugCx{\Cx}{\rdx}}{\err}$.
  \end{enumerate}
\end{lemma}
\begin{proof} By case analysis on \rdx, using rules from \figref{semantics}.
  \begin{description}
    \item[\x] By \WAE{VarErr}, $\evalerrwa{\MTg}{\plugCx{\Cx}{\x}}{\err}$,
      which is case (b).
    \item[\seq{\v}{\e}] By \WAE{Seq},
      $\evalwa{\MTg}{\plugCx{\Cx}{\seq{\v}{\e}}}{\MTg}{\plugCx{\Cx}{\e}}$,
      which is case (a).
    \item[\primcalld{\vs}] Depending on $\Primop(l, \vs)$, 
      \begin{itemize}
        \item either by \WAE{Primop},
          $\evalwa{\MTg}{\plugCx{\Cx}{\primcalld{\vs}}}{\MTg}{\plugCx{\Cx}{\v'}}$,
          i.e. case (a),
        \item or by \WAE{PrimopErr}
          $\evalerrwa{\MTg}{\plugCx{\Cx}{\primcalld{\vs}}}{\err}$, i.e. case (b).
      \end{itemize}
    \item[\mcall{\vnm}{\vs}] By \WAE{CalleeErr},
      $\evalerrwa{\MTg}{\plugCx{\Cx}{\mcall{\vnm}{\vs}}}{\err}$, i.e. case (b).
    \item[\md] By \WAE{MD},
      $\evalwa{\MTg}{\plugCx{\Cx}{\md}}{\MText{\MTg}{\md}}{\plugCx{\Cx}{\m}}$,
      which is case (a).
    \item[\evalg{\v}] By \WAE{ValGlobal},
      $\evalwa{\MTg}{\plugCx{\Cx}{\evalg{\v}}}{\MTg}{\plugCx{\Cx}{\v}}$,
      which is case (a).
    \item[\evalt{\MT}{\v}] By \WAE{ValLocal},
      $\evalwa{\MTg}{\plugCx{\Cx}{\evalt{\MT}{\v}}}{\MTg}{\plugCx{\Cx}{\v}}$,
      which is case (a).
    \item[\evalg{\plugx{\Xx}{\mcall{\m}{\vs}}}] By \WAE{CallGlobal},
      $\evalwa
         {\MTg}{\plugCx\Cx{\evalg{\plugx\Xx{\mcall\m\vs}}}}
         {\MTg}{\plugCx\Cx{\evalg{\plugx\Xx{\evalt{\MTg}{\mcall\m\vs}}}}}$,
      i.e. case (a). 
    \item[\evalt{\MT}{\plugx{\Xx}{\mcall{\m}{\vs}}}] Depending on the result
      of method call resolution \mcall{\m}{\vs} for \MT,
      \begin{itemize}
        \item either \WAE{CallLocal} gives us case (a),
        \item or \WAE{CallErr} gives us case (b).
      \end{itemize}
  \end{description}
\end{proof}

Finally, we can prove the \textbf{Progress \thmref{thm:paper:progress}}.
\begin{theorem}[Progress]\label{thm:progress}
  For any program \p and method table \MTg, 
  the program either reduces to a value,
  or it makes a step to another program,
  or it errs.
  That is, one of the following holds:
  \begin{enumerate}[label=(\alph*)]
    \item $\evalwa{\MTg}{\p}{\MTg'}{\v}$; or
    \item $\evalwa{\MTg}{\p}{\MTg'}{\p'}$; or
    \item $\evalerrwa{\MTg}{\p}{\err}$.
  \end{enumerate}
\end{theorem}
\begin{proof}\label{proof:progress}
  From $\p = \evalg{\e}$ and \lemref{lem:expr-form}, we know that
  one of the following holds:
  \begin{enumerate}
    \item $\p = \evalg{\v}$. This gives us case (a): %by \WAE{ValGlobal}:
      \begin{mathpar}
        \inferrule*[right=\WAE{ValGlobal}]
          { }
          { \evalwa{\MTg}{\evalg{\v}}{\MTg}{\v} }.
        % \inferrule*[right=\WAE{Normal}]
        %   { \evalwa{\MTg}{\evalg{\v}}{\MTg}{\v} }
        %   { \evalfullwa{\MTg}{\evalg{\v}}{\MTg}{\v}. }
      \end{mathpar}
    \item $\p = \evalg{\plugx{\Xx}{\mcall{\m}{\vs}}}$. This gives us case (b):
      %by \WAE{CallGlobal}:
      \begin{mathpar}
        \inferrule*[right=\WAE{CallGlobal}]
          { }
          { \evalwa
              {\MTg}{\evalg{\plugx{\Xx}{\mcall{\m}{\vs}}}}
              {\MTg}{\evalg{\plugx{\Xx}{\evalt{\MTg}{\mcall\m\vs}}}} }.
        % \inferrule*[right=\WAE{Normal}]
        %   { \evalwa
        %       {\MTg}{\evalg{\plugx{\Xx}{\mcall{\m}{\vs}}}}
        %       {\MTg}{\evalg{\plugx{\Xx}{\evalt{\MTg}{\mcall\m\vs}}}} }
        %   { \evalfullwa
        %       {\MTg}{\evalg{\plugx{\Xx}{\mcall{\m}{\vs}}}}
        %       {\MTg}{\evalg{\plugx{\Xx}{\evalt{\MTg}{\mcall\m\vs}}}}. }
      \end{mathpar}
    \item $\p = \evalg{\plugCx{\Cx}{\rdx}} = \plugCx{\Cx'}{\rdx}$
      for $\Cx' = \evalg{\Cx}$. By \lemref{lem:redex-steps},
      one of the following holds:
      \begin{itemize}
        \item $\evalwa{\MTg}{\plugCx{\Cx'}{\rdx}}{\MTg'}{\plugCx{\Cx'}{\e'}}$.
          This gives us (b) because $\evalg{\plugCx{\Cx}{\e'}} = \p'$:
          \[
            \evalwa
                {\MTg}{\evalg{\plugCx{\Cx}{\rdx}}}
                {\MTg'}{\evalg{\plugCx{\Cx}{\e'}}}.
          \]
          % \begin{mathpar}
          %   \inferrule*[right=\WAE{Normal}]
          %   { \evalwa
          %       {\MTg}{\evalg{\plugCx{\Cx}{\rdx}}}
          %       {\MTg'}{\evalg{\plugCx{\Cx}{\e'}}} }
          %   { \evalfullwa
          %       {\MTg}{\evalg{\plugCx{\Cx}{\rdx}}}
          %       {\MTg'}{\evalg{\plugCx{\Cx}{\e'}}}. }
          % \end{mathpar}
        \item $\evalerrwa{\MTg}{\plugCx{\Cx'}{\rdx}}{\err}$.
          This gives us (a):
          \[ \evalerrwa{\MTg}{\evalg{\plugCx{\Cx}{\rdx}}}{\err}. \]
          % \begin{mathpar}
          %   \inferrule*[right=\WAE{Error}]
          %   { \evalerrwa{\MTg}{\evalg{\plugCx{\Cx}{\rdx}}}{\err} }
          %   { \evalfullwa
          %       {\MTg}{\evalg{\plugCx{\Cx}{\rdx}}}
          %       {\MTg}{\err}. }
          % \end{mathpar}
      \end{itemize}
  \end{enumerate}
\end{proof}

\subsection{Canonical Representation of Expressions}\label{app:proofs:canonical}
%% --------------------------------------------------

In this section, we show that:
\begin{enumerate}
  \item The form of expression \e described in~\lemref{lem:expr-form}
    (\v, \plugx{\Xx}{\mcall{\m}{\vs}}, or \plugx{\Cx}{\rdx}) is in fact unique.
  \item Any normal- or error-evaluation step of \e is determined by
    its \plugx{\Cx}{\rdx} representation.
\end{enumerate}

To simplify proofs, we will use an auxiliary syntax \rep to distinguish
between expression forms, and an auxiliary judgment \Can{\e}{\rep}
to remember how those forms are built.
These new definitions are provided in~\figref{fig:wa-canoncial-forms}.

\begin{figure}
  \[
  \begin{array}{ccl@{\qquad}l}
      \\ \rep & ::= & & \text{\emph{Form of expression}}
      \\ &\Alt& \rVd &
      \\ &\Alt& \rXd &
      \\ &\Alt& \rCd &
      \\
      \\ & & \Can{\e}{\rep} & \text{\emph{Canonical form}}
  \end{array}
  \]
  \begin{mathpar}
    \inferrule[\WACAN{Val}]
    { }
    { \Can{\v}{\rVd} }
    
    \inferrule[\WACAN{Call}]
    { }
    { \Can{\mcalld}{\rXX{\hole}} }

    \inferrule[\WACAN{Redex}]
    { }
    { \Can{\rdx}{\rCC{\hole}} }
    \\

    \inferrule[\WACAN{SeqX}]
    { \Can{\e_1}{\rXX{\Xx_1}} }
    { \Can{\seq{\e_1}{\e_2}}{\rXX{\seq{\Xx_1}{\e_2}}} }

    \inferrule[\WACAN{SeqC}]
    { \Can{\e_1}{\rCC{\Cx_1}} }
    { \Can{\seq{\e_1}{\e_2}}{\rCC{\seq{\Cx_1}{\e_2}}} }
    \\

    \inferrule[\WACAN{PrimopX}]
    { \Can{\e_i}{\rXX{\Xx_i}} }
    { \Can{\primcall{l}{\vs',\,\e_i,\,\es}}{\rXX{\primcall{l}{\vs',\,\Xx_i,\,\es}}} }

    \inferrule[\WACAN{PrimopC}]
    { \Can{\e_i}{\rCC{\Cx_i}} }
    { \Can{\primcall{l}{\vs',\,\e_i,\,\es}}{\rCC{\primcall{l}{\vs',\,\Cx_i,\,\es}}} }
    \\

    \inferrule[\WACAN{CalleeX}]
    { \Can{\e_c}{\rXX{\Xx_c}} }
    { \Can{\mcall{\e_c}{\es}}{\rXX{\mcall{\Xx_c}{\es}}} }

    \inferrule[\WACAN{CalleeC}]
    { \Can{\e_c}{\rCC{\Cx_c}} }
    { \Can{\mcall{\e_c}{\es}}{\rCC{\mcall{\e_c}{\es}}} }
    \\

    \inferrule[\WACAN{CallX}]
    { \Can{\e_i}{\rXX{\Xx_i}} }
    { \Can{\mcall{\m'}{\vs',\,\e_i,\,\es}}{\rXX{\mcall{\m'}{\vs',\,\Xx_i,\,\es}}} }

    \inferrule[\WACAN{CallC}]
    { \Can{\e_i}{\rCC{\Cx_i}} }
    { \Can{\mcall{\m'}{\vs',\,\e_i,\,\es}}{\rCC{\mcall{\m'}{\vs',\,\Cx_i,\,\es}}} }
    \\

    \inferrule[\WACAN{EvalGlobalC}]
    { \Can{\e}{\rCd} }
    { \Can{\evalg{\e}}{\rCC{\evalg{\Cx}}} }

    \inferrule[\WACAN{EvalLocalC}]
    { \Can{\e}{\rCd} }
    { \Can{\evalt{\MT}{\e}}{\rCC{\evalt{\MT}{\Cx}}} }
  \end{mathpar}
\caption{Canonical forms of internal \juliette expressions}\label{fig:wa-canoncial-forms}
\end{figure}

\begin{lemma}[Reconstruction from Canonical Forms]\label{lem:can-rep-build}
  For all expressions \e, the following holds:
  \begin{enumerate}[label=(\alph*)]
    \item $\Can{\e}{\rVd} \quad\quad\ \implies \e = \v.$
    \item $\Can{\e}{\rXd} \implies \e = \plugx{\Xx}{\mcalld}.$
    \item $\Can{\e}{\rCd} \ \,\implies \e = \plugx{\Cx}{\rdx}.$
  \end{enumerate}
\end{lemma}
\begin{proof}
  By induction on the derivation of \Can{\e}{\rep}.
  \begin{enumerate}[label=(\alph*)]
  \item This case is trivial: there is only one constructor of \CanSym
    with $\rV$ expression form, \WACAN{Value}, and thus $\e = \v$ by inversion.
  \item There are multiple \CanSym constructors with $\rX$ forms.
    \begin{description}
      \item[\WACAN{Call}] Another trivial case. By inversion, $\e = \mcalld$,
        and also $\plugx{\hole}{\mcalld} = \mcalld.$
      \item[\WACAN{SeqX}] By inversion, $\e = (\seq{\e_1}{\e_2})$.
        By the induction hypothesis, $\e_1 = \plugx{\Xx_1}{\mcalld}$. Therefore,
        \[\plugx{(\seq{\Xx_1}{\e_2})}{\mcalld} = \seq{\plugx{\Xx_1}{\mcalld}}{\e_2}
            = \seq{\e_1}{\e_2}.\]
      \item[$\ldots$] Other cases are similar to \WACAN{SeqX}.
    \end{description}
  \item There are multiple \CanSym constructors with $\rC$ forms.
    \begin{description}
      \item[\WACAN{Redex}] This is a trivial case. By inversion, $\e = \rdx$,
        and also $\plugx{\hole}{\rdx} = \rdx.$
      \item[\WACAN{SeqC}] By inversion, $\e = (\seq{\e_1}{\e_2})$.
        by the induction hypothesis, $\e_1 = \plugx{\Cx_1}{\rdx}$. Therefore,
        \[\plugx{(\seq{\Cx_1}{\e_2})}{\rdx} = \seq{\plugx{\Cx_1}{\rdx}}{\e_2}
            = \seq{\e_1}{\e_2}.\]
      \item[$\ldots$] Other cases are similar to \WACAN{SeqC}.
    \end{description}
  \end{enumerate}
\end{proof}

\begin{lemma}[\Xx preserves Canonical $\rC$]\label{lem:X-preserves-canonical-C}
  For all contexts \Xx and expressions \e,
  \[\Can{\e}{\rCd} \quad \implies \quad \Can{\plugx{\Xx}{\e}}{\rCC{\plugx{\Xx}{\Cx}}}.\]
\end{lemma}
\begin{proof}
  By induction on \Xx.
  \begin{description}
    \item[\hole] Since $\plugx{\hole}{\e} = \e$ and $\plugx{\hole}{\Cx} = \Cx$,
      $\Can{\plugx{\hole}{\e}}{\rCC{\plugx{\hole}{\Cx}}}$ by assumption.
    \item[\seq{\Xx_1}{\e_2}] By definition,
      $\plugx{(\seq{\Xx_1}{\e_2})}{\e} = \seq{\plugx{\Xx_1}{\e}}{\e_2}$
      and $\plugx{(\seq{\Xx_1}{\e_2})}{\Cx} = \seq{\plugx{\Xx_1}{\Cx}}{\e_2}$.
      By the induction hypothesis for $\Xx_1$,
      \Can{\plugx{\Xx_1}{\e}}{\rCC{\plugx{\Xx_1}{\Cx}}}. Then:
      \begin{mathpar}
        \inferrule*[right=\WACAN{SeqC}]
        { \Can{\plugx{\Xx_1}{\e}}{\rCC{\plugx{\Xx_1}{\Cx}}} }
        { \Can{\seq{\plugx{\Xx_1}{\e}}{\e_2}}{\rCC{\seq{\plugx{\Xx_1}{\Cx}}{\e_2}}} }.
      \end{mathpar}
    \item[\ldots] Other cases are similar.
  \end{description}
\end{proof}

\begin{lemma}[\Cx preserves Canonical $\rC$]\label{lem:C-preserves-canonical-C}
  For all contexts $\Cx'$ and expressions \e,
  \[\Can{\e}{\rCd} \quad \implies \quad \Can{\plugx{\Cx'}{\e}}{\rCC{\plugx{\Cx'}{\Cx}}}.\]
\end{lemma}
\begin{proof}
  By induction on $\Cx'$.
  \begin{description}
    \item[\Xx] This case is covered by \lemref{lem:X-preserves-canonical-C}.
    \item[\plugx{\Xx}{\evalg{\Cx''}}] By definition, 
      $\plugx{(\plugx{\Xx}{\evalg{\Cx''}})}{\e} = \plugx{\Xx}{\evalg{\plugx{\Cx''}{\e}}}$
      and $\plugx{(\plugx{\Xx}{\evalg{\Cx''}})}{\Cx} = \plugx{\Xx}{\evalg{\plugx{\Cx''}{\Cx}}}.$
      By the induction hypothesis for $\Cx''$,
      \Can{\plugx{\Cx''}{\e}}{\rCC{\plugx{\Cx''}{\Cx}}}. Then:
      \begin{mathpar}
        \inferrule*[right=\WACAN{EvalGlobalC}]
        { \Can{\plugx{\Cx''}{\e}}{\rCC{\plugx{\Cx''}{\Cx}}} }
        { \Can{\evalg{\plugx{\Cx''}{\e}}}{\rCC{\evalg{\plugx{\Cx''}{\Cx}}}} },
      \end{mathpar}
      and by \lemref{lem:X-preserves-canonical-C},
      \[ \Can{\plugx{\Xx}{\evalg{\plugx{\Cx''}{\e}}}}{\rCC{\plugx{\Xx}{\evalg{\plugx{\Cx''}{\Cx}}}}}. \]
    \item[\plugx{\Xx}{\evalt{\MT}{\Cx''}}] Similarly to the previous case.
  \end{description}
\end{proof}

\begin{lemma}[\plugx{\Xx}{\mcalld} is Canonical]\label{lem:Xcall-canonical}
  For all \Xx and \mcalld,
  \[ \Can{\plugx{\Xx}{\mcalld}}{\rXd}. \]
\end{lemma}
\begin{proof}
  By induction on \Xx.
  \begin{description}
    \item[\hole] In this case, $\plugx{\hole}{\mcalld} = \mcalld$,
      and by \WACAN{Call}, \Can{\mcalld}{\rXX{\hole}}.
    \item[\seq{\Xx_1}{\e_2}] By definition, 
      $\plugx{(\seq{\Xx_1}{\e_2})}{\mcalld} = (\seq{\plugx{\Xx_1}{\mcalld}}{\e_2}).$
      By the induction hypothesis for $\Xx_1$ and \CanSym constructor: 
      \begin{mathpar}
        \inferrule*[right=\WACAN{SeqX}]
        { \Can{\plugx{\Xx_1}{\mcalld}}{\rXX{\Xx_1}} }
        { \Can{\seq{\plugx{\Xx_1}{\mcalld}}{\e_2}}{\rXX{\seq{\Xx_1}{\e_2}}} }.
      \end{mathpar}
    \item[\ldots] Other cases are similar.
  \end{description}
\end{proof}

\begin{lemma}[\plugx{\Xx}{\rdx} is Canonical]\label{lem:Xrdx-canonical}
  For all \Xx and \rdx,
  \[ \Can{\plugx{\Xx}{\rdx}}{\rCC{\Xx}}. \]
\end{lemma}
\begin{proof}
  By induction on \Xx.
  \begin{description}
    \item[\hole] In this case, $\plugx{\hole}{\rdx} = \rdx$,
      and by \WACAN{Redex}, \Can{\rdx}{\rCC{\hole}}.
    \item[\seq{\Xx_1}{\e_2}] By definition, 
      $\plugx{(\seq{\Xx_1}{\e_2})}{\rdx} = (\seq{\plugx{\Xx_1}{\rdx}}{\e_2}).$
      by the induction hypothesis for $\Xx_1$ and \CanSym constructor: 
      \begin{mathpar}
        \inferrule*[right=\WACAN{SeqC}]
        { \Can{\plugx{\Xx_1}{\rdx}}{\rCC{\Xx_1}} }
        { \Can{\seq{\plugx{\Xx_1}{\rdx}}{\e_2}}{\rCC{\seq{\Xx_1}{\e_2}}} }.
      \end{mathpar}
    \item[\ldots] Other cases are similar.
  \end{description}
\end{proof}

\begin{lemma}[\plugx{\Cx}{\rdx} is Canonical]\label{lem:Crdx-canonical}
  For all \Cx and \rdx,
  \[ \Can{\plugx{\Cx}{\rdx}}{\rCC{\Cx}}. \]
\end{lemma}
\begin{proof}
  By induction on \Cx.
  \begin{description}
    \item[\Xx] This case is covered by \lemref{lem:Xrdx-canonical}.
    \item[\plugx{\Xx}{\evalg{\Cx'}}] By definition, 
      $\plugx{(\plugx{\Xx}{\evalg{\Cx'}})}{\rdx} = \plugx{\Xx}{\evalg{\plugx{\Cx'}{\rdx}}}.$
      By the induction hypothesis for $\Cx'$ and \CanSym constructor: 
      \begin{mathpar}
        \inferrule*[right=\WACAN{EvalGlobalC}]
        { \Can{\plugx{\Cx'}{\rdx}}{\rCC{\Cx'}} }
        { \Can{\evalg{\plugx{\Cx'}{\rdx}}}{\rCC{\evalg{\Cx'}}} }.
      \end{mathpar}
      Then, by \lemref{lem:X-preserves-canonical-C},
      \[ \Can{\plugx{\Xx}{\evalg{\plugx{\Cx'}{\rdx}}}}{\rCC{\plugx{\Xx}{\evalg{\Cx'}}}}. \]
      \item[\plugx{\Xx}{\evalt{\MT}{\Cx'}}] Similarly to the previous case.
  \end{description}
\end{proof}

\begin{lemma}[Canonical Form is Unique]\label{lem:canonical-unique}
  For all expressions \e and representations $\rep_1, \rep_2$,
  \[ \Can{\e}{\rep_1} \land \Can{\e}{\rep_2} \quad \implies \quad \rep_1 = \rep_2. \]
\end{lemma}
\begin{proof}
  By induction on \e.
  \begin{description}
    \item[\v] There is only one \CanSym constructor for \v, \WACAN{Value}.
      Therefore, $\rep_1 = \rVd = \rep_2$.
    \item[\x] There is only one \CanSym constructor for \x, \WACAN{Redex}.
      Therefore, $\rep_1 = \rC(\hole, \x) = \rep_2$.
    \item[\seq{\e_1}{\e_2}] By inversion of \Can{\seq{\e_1}{\e_2}}{\rep_1},
      three cases are possible.
      \begin{enumerate}
        \item \e is a redex with $\e_1 = \v_{11}$ and
          \Can{\seq{\v_{11}}{\e_2}}{\rC(\hole, \seq{\v_{11}}{\e_2})},
          and thus \Can{\e_1}{\rV(\v_{11})}.
        \item \Can{\seq{\e_1}{\e_2}}{\rXX{\seq{\Xx_1}{\e_2}}} with
          \Can{\e_1}{\rXX{\Xx_1}}.
        \item \Can{\seq{\e_1}{\e_2}}{\rCC{\seq{\Cx_1}{\e_2}}} with
        \Can{\e_1}{\rCC{\Cx_1}}.
      \end{enumerate}
      By inversion of \Can{\seq{\e_1}{\e_2}}{\rep_2},
      similar three cases are possible, with different representations
      of $\e_1$. However, by the induction hypothesis, we know that the canonical
      representation of $\e_1$ is unique. And thus the representations
      of \seq{\e_1}{\e_2} will coincide too.
    \item[\ldots] Other cases are similar. 
  \end{description}
\end{proof}

Finally, we can prove the \textbf{Unique Form \lemref{lem:paper:expr-form-unique}},
reformulated here using \CanSym.
\begin{lemma}[Unique Canonical Representation of Expression]\label{lem:expr-form-canonical}
  Any expression \e can be uniquely represented in one of the following ways:
  \begin{enumerate}[label=(\alph*)]
    \item $\e = \v$ with \Can{\e}{\rVd}; or
    \item $\e = \plugx{\Xx}{\mcalld}$ with \Can{\e}{\rXd}; or
    \item $\e = \plugx{\Cx}{\rdx}$ with \Can{\e}{\rCd}.
  \end{enumerate}
\end{lemma}
\begin{proof}
  By \lemref{lem:expr-form}, we know that at least one of the
  conditions above holds. For every condition, \e has the corresponding
  canonical representation: by \WACAN{Value} for (a),
  by \lemref{lem:Xcall-canonical} for~(b),
  and by \lemref{lem:Crdx-canonical} for~(c).
  But by \lemref{lem:canonical-unique}, the canonical representation is unique.
  Thus, for example, if $\e = \plugx{\Cx_1}{\rdx_1}$ and
  $\e = \plugx{\Cx_2}{\rdx_2}$, then $\Cx_1 = \Cx_2$ and $\rdx_1 = \rdx_2$.
\end{proof}

This fact will be used implicitly in various proofs that follow.

\subsection{Determinism of \juliette Semantics}\label{app:proofs:determinism}
%% --------------------------------------------------

% \begin{figure}
%   \[
%   \begin{array}{ccl@{\qquad}l}
%      \iexp & := & \e \Alt \err & \text{\emph{Expression or Error}}
%   \end{array}
%   \]
%   \caption{Updated internal \juliette syntax}\label{fig:wa-calculus-int-syntax-upd}
% \end{figure}

\begin{lemma}[Only Redex Steps]\label{lem:only-redex-steps}
  For all expressions \e and method tables \MT, the following hold:
  \begin{enumerate}
    \item $\evalwa{\MT}{\e}{\MT'}{\e'} \quad \implies \quad \e = \plugCx{\Cx}{\rdx},$
    \item $\evalerrwa{\MT}{\e}{\err} \quad \implies \quad \e = \plugCx{\Cx}{\rdx},$
    \item $\e = \plugCx{\Cx}{\rdx} \quad\implies\quad
           \evalwa{\MT}{\e}{\MT'}{\e'} \text{ or }\ \evalerrwa{\MT}{\e}{\err}.$
  \end{enumerate}
\end{lemma}
\begin{proof}
  The first two are established by analyzing
  normal- and error-evaluation rules.
  The last statement can be established by case analysis on \rdx.
\end{proof}

Finally, we can prove the \textbf{Determinism \thmref{thm:paper:eval-deterministic}}.
\begin{theorem}[Determinism]\label{thm:eval-deterministic}
\juliette semantics is deterministic. That is,
for all expressions \e and method tables \MT,
if \e can make a normal- or error-evaluation step, the step is unique.
\end{theorem}
\begin{proof}\label{proof:eval-deterministic}
  By \lemref{lem:only-redex-steps}, we know that if \e can make a step,
  then exists some \Cx and \rdx such that $\e = \plugCx{\Cx}{\rdx}$.
  By \lemref{lem:expr-form-canonical}, we know that such a representation
  is unique.
  Finally, by case analysis on \rdx, we can see that for all redex bases
  except for \primcalld{\vs} and \evalt{\MT}{\plugx{\Xx}{\mcall{\m}{\vs}}},
  there is exactly one (normal- or error-evaluation) rule applicable.
  For \primcalld{\vs} and \evalt{\MT}{\plugx{\Xx}{\mcall{\m}{\vs}}},
  there are two rules for each, but their premises are incompatible.
  Thus, for any \plugCx{\Cx}{\rdx}, exactly one rule is applicable.
\end{proof}

\subsection{Preservation of concrete typing}
%% --------------------------------------------------

The following lemmas are needed for the correctness of optimization.

\begin{lemma}[Function Call is Irrelevant for Concrete Typing]\label{lem:fun-call}
  For all \Xx, \Gm, \mcall{\m}{\es}:
  \[ \typedwad{\plugx{\Xx}{\mcall{\m}{\es}}}{\g}
     \quad \implies \quad 
     \forall \e'.\ \typedwad{\plugx{\Xx}{\e'}}{\g}. \]
\end{lemma}
\begin{proof}
  By induction on \Xx.
  \begin{description}
    \item[\hole] This case is impossible case because there is no
      concrete typing for \mcall{\m}{\es}.
    \item[\seq{\Xx_1}{\e_2}] Here,
      $\plugx{\Xx}{\mcall{\m}{\es}}
       = \seq{(\plugx{\Xx_1}{\mcall{\m}{\es}})}{\e_2}.$ By assumption,
      \typedwad{\seq{(\plugx{\Xx_1}{\mcall{\m}{\es}})}{\e_2}}{\g},
      and thus by inversion of \WAT{Seq}, \typedwad{\e_2}{\g}. Therefore:
      \begin{mathpar}
        \inferrule*[right=\WAT{Seq}]
        { \typedwad{\e_2}{\g} }
        { \typedwad{\seq{(\plugx{\Xx_1}{\e'})}{\e_2}}{\g} }.
      \end{mathpar}
    \item[$\primopd(\vs\ \Xx_i\ \es)$] Here,
     $\plugx{\Xx}{\mcall{\m}{\es}}
      = \primopd(\vs\ \plugx{\Xx_i}{\mcall{\m}{\es}}\ \es)$. By assumption,
      \begin{mathpar}
        \inferrule*[right=\WAT{Primop}]
        { \typedwad{\v_j}{\g'_j} \\
          \typedwad{\plugx{\Xx_i}{\mcall{\m}{\es}}}{\g'_i} \\ 
          \typedwad{\e_k}{\g'_k} \\ \PrimopRT(l, \gs')=\g }
        { \typedwad{\primopd(\vs\ \plugx{\Xx_i}{\mcall{\m}{\es}}\ \es)}{\g} }.
      \end{mathpar}
      By induction hypothesis, \typedwad{\plugx{\Xx_i}{\e'}}{\g'_i}.
      Therefore, \typedwad{\primopd(\vs\ \plugx{\Xx_i}{\e'}\ \es)}{\g}
      by \WAT{Primop}.
    \item[$\Xx'(\es')$] This case is impossible case because there is no
      concrete typing for (\plugx{\Xx'}{\mcall{\m}{\es}})(\es').
    \item[$\v(\vs'\ \Xx'\ \es')$] This case is impossible case because there is no
      concrete typing for $\v(\vs'\ \plugx{\Xx'}{\mcall{\m}{\es}}\ \es')$.
  \end{description}
\end{proof}

\begin{theorem}[Normal Evaluation Preseves Concrete Typing]
  For all \e, \g, \MT, $\MT'$, $\e'$:
  \[ 
    \typedwa{}{}{\e}{\g} \quad \land \quad \evalwa{\MT}{\e}{\MT'}{\e'}
    \qquad \implies \qquad
    \typedwa{}{}{\e'}{\g}.
  \]
\end{theorem}
\begin{proof}
  By induction on the derivation of $\typedwa{}{}{\e}{\g}$.
  \begin{description}
    \item[\WAT{Var}] This case is impossible.
    \item[\WAT{Val}] By inversion, $\e = \v$ and \typedwa{}{}{\v}{\sigma},
      but \v cannot make a step, so the case is vacuously true.
    \item[\WAT{MD}] By inversion, \e is a method definition \mdefd with
      \typedwa{}{}{\mdefd}{\mty}.
      By \WAE{MD}, \evalwa{\MT}{\mdefd}{\MText{\MT}{\mdefd}}{\m},
      and \typedwa{}{}{\m}{\mty} because $\typeof(\m) = \mty$.
    \item[\WAT{Seq}] By inversion, $\e = \seq{\e_1}{\e_2}$,
      \typedwa{}{}{\seq{\e_1}{\e_2}}{\g}, and \typedwa{}{}{\e_2}{\g}.
      There are two possibilities for \seq{\e_1}{\e_2} to make a step.
      \begin{enumerate}
        \item $\e_1 = \v_{11}$ and $\evalwa{\MT}{\seq{\v_{11}}{\e_2}}{\MT}{\e_2}$.
          By assumption, \typedwa{}{}{\e_2}{\g}.
        \item $\e_1 = \plugCx{\Cx_1}{\rdx_1}$ and
          $\evalwa{\MT}{\seq{\plugCx{\Cx_1}{\rdx_1}}{\e_2}}
                  {\MT'}{\seq{\plugCx{\Cx_1}{\e_1'}}{\e_2}}$
          with $\e' = \seq{\plugCx{\Cx_1}{\e_1'}}{\e_2}$. Then:
          \begin{mathpar}
            \inferrule*[right=\WAT{Seq}]
            { \typedwa{}{}{\e_2}{\g} }
            { \typedwa{}{}{\seq{\plugCx{\Cx_1}{\e_1'}}{\e_2}}{\g} }.
          \end{mathpar}
      \end{enumerate}
    \item[\WAT{Primop}] By inversion, $\e = \primopd(\es)$,
      \typedwa{}{}{\e_i}{\g'_i}, and $\PrimopRT(l, \gs')=\g$.
      There are two ways how \primopd(\es) can make a step.
      \begin{enumerate}
        \item $\primopd(\es) = \primopd(\vs)$ and
          $\evalwa{\MT}{\primopd(\vs)}{\MT}{\v'}$, where $\Primop(l, \vs)=\v'$.
          By the properties of \Primop and \PrimopRT, we know that
          $\typeof(\v') = \g$, and thus by \WAT{Val}, \typedwa{}{}{\v'}{\g}.
        \item $\primopd(\es) = \primopd(\vs\ \plugCx{\Cx_i}{\rdx_i}\ \es)$ and
          \[
            \evalwa{\MT}{\primopd(\vs\ \plugCx{\Cx_i}{\rdx_i}\ \es)}
                   {\MT'}{\primopd(\vs\ \plugCx{\Cx_i}{\e'_i}\ \es)}
            \quad \iff \quad
            \evalwa{\MT}{\plugCx{\Cx_i}{\rdx_i}}{\MT'}{\plugCx{\Cx_i}{\e'_i}}.
          \]
          By assumption, \typedwa{}{}{\plugCx{\Cx_i}{\rdx_i}}{\g'_i}
          and \plugCx{\Cx_i}{\rdx_i} makes a step.
          Therefore, by the induction hypothesis, we know
          \typedwa{}{}{\plugCx{\Cx_i}{\e'_i}}{\g'_i}. Because other arguments
          of \primopd are the same, by \WAT{Primop},
          \typedwa{}{}{\primopd(\vs\ \plugCx{\Cx_i}{\rdx_i}\ \es)}{\g}.
      \end{enumerate}
    \item[\WAT{EvalGlobal}] By inversion, $\e = \evalg{\e_e}$ and
      \typedwa{}{}{\e_e}{\g}. There are three possibilities for \evalg{\e_e}
      to make a step.
      \begin{enumerate}
        \item $\evalg{\e_e} = \evalg{\v}$ and \evalwa{\MT}{\evalg{\v}}{\MT}{\v}.
          But we already know that \typedwa{}{}{\v}{\g}.
        \item $\evalg{\e_e} = \evalg{\plugx{\Xx}{\mcalld}}$ and
          \evalwa{\MT}{\evalg{\plugx{\Xx}{\mcalld}}}
                 {\MT}{\evalg{\plugx{\Xx}{\evalt{\MT}{\mcalld}}}}.
          By \lemref{lem:fun-call}, we know that if
          \typedwa{}{}{\plugx{\Xx}{\mcalld}}{\g}, then
          \typedwa{}{}{\plugx{\Xx}{\evalt{\MT}{\mcalld}}}{\g}.
          Thus, by \WAT{EvalGlobal}:
          \[ \typedwa{}{}{\evalg{\plugx{\Xx}{\evalt{\MT}{\mcalld}}}}{\g}. \]
        \item $\evalg{\e_e} = \evalg{\plugCx{\Cx}{\rdx_e}}$ and
          \[
            \evalwa{\MT}{\evalg{\plugCx{\Cx}{\rdx_e}}}
                   {\MT'}{\evalg{\plugCx{\Cx}{\e'_e}}}
            \quad \iff \quad
            \evalwa{\MT}{\plugCx{\Cx}{\rdx_e}}{\MT'}{\plugCx{\Cx}{\e'_e}}.
          \]
          Because \typedwa{}{}{\plugCx{\Cx}{\rdx_e}}{\g},
          by the induction hypothesis, we have
          \typedwa{}{}{\plugCx{\Cx}{\e'_e}}{\g}. Therefore:
          \begin{mathpar}
            \inferrule*[right=\WAT{EvalGlobal}]
            { \typedwa{}{}{\plugCx{\Cx}{\e'_e}}{\g} }
            { \typedwa{}{}{\evalg{\plugCx{\Cx}{\e'_e}}}{\g} }.
          \end{mathpar}
      \end{enumerate}
    \item[\WAT{EvalLocal}] Similarly to previous case. The only difference is
      in evaluation step (2), which is
      \evalwa{\MT}{\evalt{\MT_l}{\plugx{\Xx}{\mcalld}}}
             {\MT}{\evalt{\MT_l}{\plugx{\Xx}{\e_b\subst{\xs}{\vs}}}},
      but \lemref{lem:fun-call} applies in the same way.
  \end{description}
\end{proof}

\subsection{Value Substitution and Optimization}\label{app:proofs:value-subst}
%% --------------------------------------------------

Value substitution is defined in \figref{fig:app:wa-value-subst},
which simply reproduces \figref{fig:wa-value-subst} for convenience.

\begin{lemma}[Value Substitution Preserves Concrete Typing]\label{lem:subst-preserves-typing}
  For all $\e, \Gm, \g, \gm$:
  \[ \typedwad{\e}{\g}\ \ \land\ \ \substokd \quad\implies\quad
     \typedwa{}{}{\gm(\e)}{\g}. \]
\end{lemma}
\begin{proof}
  By induction on the derivation of \typedwad{\e}{\g}.
  \begin{description}
    \item[\WAT{Var}] By inversion of \typedwad{\x}{\g},
      we know that $\Gm(\x) = \g$. Therefore, from \substokd, we have
      \begin{mathpar}
        \inferrule*[right=\WAT{Val}]
        { \typeof(\gm(\x)) = \g }
        { \typedwa{}{}{\gm(\x)}{\g} }.
      \end{mathpar}
    \item[\WAT{Val}] Since $\gm(\v) = \v$, \typedwa{}{}{\gm(\v)}{\g}
      by \WAT{Val}.
    \item[\WAT{MD}] Since 
      $\gm(\mdefd) = \mdef{\m}{\obar{\x::\t}}{\gm'(\e)}$ where
      $\gm' = \gm \setminus \xs$, \WAT{MD} still holds.
    \item[\WAT{Seq}] By inversion of \typedwad{(\seq{\e_1}{\e_2})}{\g},
      we know that \typedwad{\e_2}{\g}. By the induction hypothesis on typing of $\e_2$,
      \typedwa{}{}{\gm(\e_2)}{\g}. By definition,
      $\gm(\seq{\e_1}{\e_2}) = \seq{\gm(\e_1)}{\gm(\e_2)}$. Therefore:
      \begin{mathpar}
        \inferrule*[right=\WAT{Seq}]
        { \typedwa{}{}{\gm(\e_2)}{\g} }
        { \typedwa{}{}{\seq{\gm(\e_1)}{\gm(\e_2)}}{\g} }.
      \end{mathpar}
    \item[\WAT{Primop}] Similarly to the previous case.
    \item[\WAT{EvalGlobal}] By inversion of \typedwad{\evalg{\e'}}{\g},
      we know \typedwad{\e'}{\g}. Then, by the induction hypothesis,
      \typedwa{}{}{\gm(\e')}{\g}. Since $\gm(\evalg{\e'}) = \evalg{\gm(\e')},$
      we have
      \begin{mathpar}
        \inferrule*[right=\WAT{EvalGlobal}]
        { \typedwa{}{}{\gm(\e')}{\g} }
        { \typedwa{}{}{\evalg{\gm(\e')}}{\g} }.
      \end{mathpar}
    \item[\WAT{EvalLocal}] Similarly to the previous case.
  \end{description}
\end{proof}

\begin{figure}
  \begin{mathpar}
  \begin{array}{rcll}
    \gm & ::= & \obar{\x \mapsto \v} & \text{\emph{Value substitution}}\\
      & & \text{where } \forall i,j. \x_i \neq \x_j
  \end{array}
  \\
  \inferrule*[right=\gm-Ok]
    { \dom(\Gm) = \dom(\gm) \\
      \forall \x \in \dom(\gm).\ 
        [\Gm(\x) = \g \iff \typeof(\gm(\x)) = \g] }
    { \substokd }
  \end{mathpar}
  \caption{Value substitution}\label{fig:app:wa-value-subst}
\end{figure}

Now we can prove the \textbf{Value Substitution Preserves
Optimization~\lemref{lem:paper:subst-preserves-optim}}.
\begin{lemma}[Value Substitution Preserves Optimization]
  \label{lem:subst-preserves-optim}
  For all $\SpecEnv, \Gm, \e, \MT, \e', \MT', \gm$, such that
  $\forall \v \in \gm.\ \tableoptexprd{\v}$,
  the following holds:
  \[
    \left(\exproptd{\e}{\e'}\ \land\ \substokd \right)
    \quad\implies\quad
    \expropt{\SpecEnv}{}{\MT}{\gm(\e)}{\MT'}{\gm(\e')}.
  \]
\end{lemma}
\begin{proof}\label{proof:subst-preserves-optim}
  By induction on the derivation of \exproptd{\e}{\e'}.
  \begin{description}
    \item[\WAOE{Val}] By assumption, \exproptd{\v}{\v} with $\v \neq \m$.
      By definition, $\gm(\v) = \v$, and by constructor, the optimization holds
      for any typing environment, including the empty one:
      \begin{mathpar}
        \inferrule*[right={\WAOE{Val}}]{ \v \neq \m }
        { \expropt{\SpecEnv}{}{\MT}{\v}{\MT'}{\v} }.
      \end{mathpar}
    \item[\WAOE{Seq}] By assumption,
      \exproptd{\seq{\e_1}{\e_2}}{\seq{\e_1'}{\e_2'}}. Then, by inversion,
      we also have judgments \exproptd{\e_1}{\e_1'} and \exproptd{\e_2}{\e_2'}.
      By definition, 
      \[ \gm(\seq{\e_1}{\e_2}) = \seq{\gm(\e_1)}{\gm(\e_2)}\
         \text{ and }\ 
         \gm(\seq{\e_1'}{\e_2'}) = \seq{\gm(\e_1')}{\gm(\e_2')}. \]
      By the induction hypothesis on the optimization of $\e_1$, we have
      \[ \expropt{\SpecEnv}{}{\MT}{\gm(\e_1)}{\MT'}{\gm(\e_1')}, \]
      and similarly for $\e_2$. Therefore:
      \begin{mathpar}
        \inferrule*[right={\WAOE{Seq}}]
        { \expropt{\SpecEnv}{}{\MT}{\gm(\e_1)}{\MT'}{\gm(\e_1')} \\
          \expropt{\SpecEnv}{}{\MT}{\gm(\e_2)}{\MT'}{\gm(\e_2')}}
        { \expropt{\SpecEnv}{}{\MT}{\seq{\gm(\e_1)}{\gm(\e_2)}}
                              {\MT'}{\seq{\gm(\e_1')}{\gm(\e_2')}} }.
      \end{mathpar}
    \item[\WAOE{Inline}] In this case, $\e = \mcall{\m}{\nvs}$,
      $\e' = (\seq{\skp}{\e''})$, and the assumption is:
      \begin{mathpar}
        \inferrule*[right=\WAOE{Inline}]
        { \typedwad{\nv_i}{\g_i} \\
          \getmd(\MT,\m,\gs) = \mdef{\m}{\obar{\jty{\x}{\t}}}{\e_b} \\
          \expropt{\SpecEnv}{\Gm}{\MT}{\e_b\subst{\xs}{\nvs}}{\MT'}{\e''} }
        { \exproptd{\mcall{\m}{\nvs}}{\seq{\skp}{\e''}} }.
      \end{mathpar}
      By the induction hypothesis on the optimization of $\e_b\subst{\xs}{\nvs}$, we have
      \[ \expropt{\SpecEnv}{}{\MT}{\gm(\e_b\subst{\xs}{\nvs})}
                             {\MT'}{\gm(\e'')}. \]
      Because $\e_b$ should not have free variables bound in \Gm and
      $\dom(\Gm) = \dom(\gm)$ (thanks to \substokd),
      $\gm(\e_b\subst{\xs}{\nvs}) = \e_b\subst{\xs}{\gm(\nvs)}$.
      Next, by definition, $\gm(\mcall{\m}{\nvs}) = \mcall{\m}{\gm(\nvs)}$
      and $\gm(\seq{\skp}{\e''}) = \seq{\skp}{\gm(\e'')}$.
      Because \substokd, by \lemref{lem:subst-preserves-typing},
      we have \typedwa{}{}{\gm(\nv_i)}{\g_i}. Thus:
      \begin{mathpar}
        \inferrule*[right=\WAOE{Inline}]
        { \typedwa{}{}{\gm(\nv_i)}{\g_i} \\
          \getmd(\MT,\m,\gs) = \mdef{\m}{\obar{\jty{\x}{\t}}}{\e_b} \\
          \expropt{\SpecEnv}{}{\MT}{\e_b\subst{\xs}{\gm(\nvs)}}{\MT'}{\gm(\e'')} }
        { \expropt{\SpecEnv}{}{\MT}{\mcall{\m}{\gm(\nvs)}}{\MT'}{\seq{\skp}{\gm(\e'')}} }.
      \end{mathpar}
    \item[\WAOE{Direct}] In this case, $\e = \mcall{\m}{\es}$,
      $\e' = \mcall{\m'}{\es'}$, and the assumption is:
      \begin{mathpar}
        \inferrule*[right=\WAOE{Direct}]
        { \exproptd{\e_i}{\e'_i} \\ \typedwad{\e'_i}{\g_i} \\
          (\mdeqd) \in \SpecEnv
        }
        { \exproptd{\mcall{\m}{\es}}{\mcall{\m'}{\es'}} }.
      \end{mathpar}
      By the induction hypothesis on optimization of $\e_i$, 
      \[ \expropt{\SpecEnv}{}{\MT}{\gm(\e_i)}{\MT'}{\gm(\e'_i)}. \]
      By \lemref{lem:subst-preserves-typing},
      \[ \typedwa{}{}{\gm(\e_i')}{\g_i}. \]
      By definition, $\gm(\mcall{\m}{\es}) = \mcall{\m}{\gm(\es)}$
      and $\gm(\mcall{\m'}{\es'}) = \mcall{\m'}{\gm(\es')}$.
      Therefore, the desired optimization holds:
      \begin{mathpar}
        \inferrule*[right=\WAOE{Direct}]
        { \expropt{\SpecEnv}{}{\MT}{\gm(\e_i)}{\MT'}{\gm(\e'_i)} \\
          \typedwa{}{}{\gm(\e_i')}{\g_i} \\
          (\mdeqd) \in \SpecEnv
        }
        { \expropt{\SpecEnv}{}{\MT}{\mcall{\m}{\gm(\es)}}
                              {\MT'}{\mcall{\m'}{\gm(\es')}} }.
      \end{mathpar}
    \item[\ldots] Other cases go similarly to the considered ones.
  \end{description}
\end{proof}

\subsection{Correctness of Optimization}\label{app:proofs:opt-correctness}
%% --------------------------------------------------

Recall the following lemmas from \secref{subsec:correctness}.

\begin{lemma}[Context Irrelevance]\label{lem:context-irrelevant}
  For all \rdx, \Cx, $\Cx'$, \MT, the following holds:
  \begin{enumerate}
    \item For all $\MT'$, $\e'$,
      \[ \evalwa{\MT}{\plugCx{\Cx}{\rdx}}{\MT'}{\plugCx{\Cx}{\e'}}
         \quad \iff \quad 
         \evalwa{\MT}{\plugCx{\Cx'}{\rdx}}{\MT'}{\plugCx{\Cx'}{\e'}}. \]
    \item \[ \evalerrwa{\MT}{\plugCx{\Cx}{\rdx}}{\err}
      \quad \iff \quad \evalerrwa{\MT}{\plugCx{\Cx'}{\rdx}}{\err}. \]
  \end{enumerate}
\end{lemma}
\begin{proof}
  See \lemref{lem:paper:context-irrelevant}
  on page~\pageref{lem:paper:context-irrelevant}.
\end{proof}

\begin{lemma}[Simple-Context Irrelevance]\label{lem:simple-context-irrelevant}
  For all \MT, \Cx, \e, \MTg, $\e'$, $\MTg'$, \Xx, the following holds:
  \[ 
     \evalwa{\MTg}{\plugCx{\Cx}{\evalt{\MT}{\e}}}
            {\MTg'}{\plugCx{\Cx}{\evalt{\MT}{\e'}}}
     \quad\implies\quad
     \evalwa{\MTg}{\plugCx{\Cx}{\evalt{\MT}{\plugx{\Xx}{\e}}}}
            {\MTg'}{\plugCx{\Cx}{\evalt{\MT}{\plugx{\Xx}{\e'}}}}.
  \]
\end{lemma}
\begin{proof}
  See \lemref{lem:paper:simple-context-irrelevant}
  on page~\pageref{lem:paper:simple-context-irrelevant}.
\end{proof}

\begin{lemma}[Optimization Preserves Values]\label{lem:opt-preserves-values}
For all $\SpecEnv, \MT, \MT', \Gm, \e, \v$, the following hold:
\[ \exproptd{\v}{\e} \quad\implies\quad \e = \v
    \qquad\text{ and }\qquad
    \exproptd{\e}{\v} \quad\implies\quad \e = \v.\]
\end{lemma}
\begin{proof}
See \lemref{lem:paper:optimization-preserves-values}
on page~\pageref{lem:paper:optimization-preserves-values}.
\end{proof}

Finally, we can prove
the \textbf{Bisimulation \lemref{lem:paper:expr-optimization-bisimulation}}.
\begin{lemma}[Bisimulation]\label{lem:expr-optimization-bisimulation}
  For all method tables $\MT, \MT'$, method-optimization environment $\SpecEnv$,\\
  and expressions $\e_1$, ${\e_1}\!'$, such that
  \[ \tableopt{\SpecEnv}{\e_1}{\MT}{\MT'} \quad \text{and} \quad
     \expropt{\SpecEnv}{}{\MT}{\e_1}{\MT'}{\e_1'}, \]
  for all global tables $\MT_g, {\MT_g}\!'$ and world context \Cx,
  the following hold:
  \begin{enumerate}
    \item Forward direction:
      \[
      \begin{array}{rl}
        \forall \e_2. &
          \evalwa{\MT_g}{\plugCx\Cx{\evalt{\MT}{\e_1}}}
                 {\MT_g'}{\plugCx\Cx{\evalt{\MT}{\e_2}}} \\
        & \implies \\
        & \exists \e_2'.\
          \evalwa{\MT_g}{\plugCx\Cx{\evalt{\MT'}{\e_1'}}}
                 {\MT_g'}{\plugCx\Cx{\evalt{\MT'}{\e_2'}}}
          \ \ \ \land\ \
          \expropt{\SpecEnv}{}{\MT}{\e_2}{\MT'}{\e_2'}.
      \end{array}
      \]
    \item Backward direction:
      \[
      \begin{array}{rl}
        \forall \e_2'. &
          \evalwa{\MT_g}{\plugCx\Cx{\evalt{\MT'}{\e_1'}}}
                 {\MT_g'}{\plugCx\Cx{\evalt{\MT'}{\e_2'}}} \\
        & \implies \\
        & \exists \e_2.\
          \evalwa{\MT_g}{\plugCx\Cx{\evalt{\MT}{\e_1}}}
                 {\MT_g'}{\plugCx\Cx{\evalt{\MT}{\e_2}}}
          \ \ \ \land\ \
          \expropt{\SpecEnv}{}{\MT}{\e_2}{\MT'}{\e_2'}.
      \end{array}
      \]
  \end{enumerate}
\end{lemma}
\begin{proof}\label{proof:optimization-correct}
  The proof goes by induction on the derivation of
  optimization \expropt{\SpecEnv}{}{\MT}{\e_1}{\MT'}{\e_1'}.
  \begin{itemize}
    \item Case \WAOE{Val}, where $\e_1 = \e'_1 = \v$
      and \expropte{\v}{\v}, %and $\v \neq \m$,
      is trivial: there are no steps that
      \plugCx\Cx{\evalt{\MT}{\v}} can make to arrive into
      \plugCx\Cx{\evalt{\MT}{\_}} (\v is not \rdx).
      The only possible step is
      \[\evalwa{\MT_g}{\plugCx\Cx{\evalt{\MT}{\v}}}{\MT_g}{\plugCx\Cx{\v}},\]
      and thus direction (1) is vacuously true.
      Same applies to \plugCx\Cx{\evalt{\MT'}{\v}} and direction (2).
    
    \item Case \WAOE{ValFun} is similar to \WAOE{Val}.
    
    \item Case \WAOE{Var}, where $\e_1 = \e'_1 = \x$
      and \expropte{\x}{\x}, is trivial:
      \plugCx\Cx{\evalt{\MT}{\x}} cannot make a normal-evaluation step,
      only error step, so direction (1) is vacuously true.
      Same applies to \plugCx\Cx{\evalt{\MT'}{\x}} and direction (2).

    \item Case \WAOE{Global}, where $\e_1 = \e'_1 = \evalg{\e}$
      and \expropte{\evalg{\e}}{\evalg{\e}}.
      \begin{itemize}
        \item If \e is a value, \plugCx{\Cx}{\evalt{*}{\evalg{\v}}} reduces
          to \plugCx{\Cx}{\evalt{*}{\v}},
          and optimization holds by \WAOE{Val/ValFun}.
        \item If \e is \plugx{\Xx}{\mcalld},
          \[
            \evalwa
            {\MTg}{\plugCx\Cx{\evalt{*}{\evalg{\plugx{\Xx}{\mcalld}}}}}
            {\MTg}{\plugCx\Cx{\evalt{*}{\evalt{\MTg}{\plugx{\Xx}{\mcalld}}}}},
          \]
          and the desired optimization relation holds by \WAOE{Local}.
        \item If \e is \plugCx{\Cx'}{\rdx}
          and reduces to some \plugCx{\Cx'}{\e'},
          by~\lemref{lem:context-irrelevant} (context irrelevance):
          \[\begin{array}{c}
            \evalwa{\MTg}{\plugCx{\Cx'}{\rdx}}
                   {\MTg'}{\plugCx{\Cx'}{\e'}}
            \\ \iff \\
            \evalwa{\MTg}{\plugCx\Cx{\evalt{*}{\evalg{\plugCx{\Cx'}{\rdx}}}}}
                   {\MTg'}{\plugCx\Cx{\evalt{*}{\evalg{\plugCx{\Cx'}{\e'}}}}}. \\
            \end{array}\]
          Then, the desired optimization relation holds by \WAOE{Global}
          (we rely on the fact that evaluation does not generate new method
           names in the syntax).
      \end{itemize}

    \item Case \WAOE{Local} is similar to \WAOE{Global}.

    \item Case \WAOE{MD}, where $\e_1 = \e'_1 = \md$
      and \expropte{\md}{\md}. By \WAE{MD},
      \[
        \evalwa
          {\MTg}{\plugCx\Cx{\evalt{\MT}{\md}}}
          {\MText{\MTg}{\md}}{\plugCx\Cx{\evalt{\MT}{\m}}}
      \]
      and
      \[
        \evalwa
          {\MTg}{\plugCx\Cx{\evalt{\MT'}{\md}}}
          {\MText{\MTg}{\md}}{\plugCx\Cx{\evalt{\MT'}{\m}}},
      \]
      where $\mathop{name}(\md) = \m$.
      Because by inversion of \WAOE{MD},
      \tableoptexprd{\mathop{name}(\md)}, the assumption
      of \WAOE{ValFun} holds, and thus, \expropte{\m}{\m}.

    \item Case \WAOE{Seq}, where $\e_1 = (\seq{\e_{11}}{\e_{12}})$
      and $\e_1' = (\seq{\e'_{11}}{\e'_{12}})$:
      \begin{mathpar}
        \inferrule*[right=\WAOE{Seq}]
        { \expropte{\e_{11}}{\e'_{11}} \\ \expropte{\e_{12}}{\e'_{12}} }
        { \expropte{\seq{\e_{11}}{\e_{12}}}{\seq{\e'_{11}}{\e'_{12}}}. }
      \end{mathpar}
      Let us consider the \emph{forward} direction (1) first.
      To make a step
      \[ \evalwa{\MT_g}{\plugCx\Cx{\evalt{\MT}{\seq{\e_{11}}{\e_{12}}}}}
                {\MT_g'}{\plugCx\Cx{\evalt{\MT}{\_}}}, \]
      the source expression \seq{\e_{11}}{\e_{12}} has to be either some
      \plugx{\Xx}{\mcalld} or \plugx{\Cx'}{\rdx}. By analyzing \CanSym forms,
      we know all such cases:
      \begin{enumerate}
        \item \WACAN{SeqX} with $\Xx = \seq{\Xx_1}{\e_{12}}$
          and \Can{\seq{\e_{11}}{\e_{12}}}{\rXd}, which by inversion
          gives \Can{\e_{11}}{\rXX{\Xx_1}}.
          Therefore, \plugCx\Cx{\evalt{\MT}{\seq{\e_{11}}{\e_{12}}}} steps
          by \WAE{CallLocal}:
          \[
            \evalwa
            {\MT_g}{\plugCx\Cx{\evalt{\MT}{\seq{\plugx{\Xx_1}{\mcalld}}{\e_{12}}}}}
            {\MT_g}{\plugCx\Cx{\evalt{\MT}{\seq{\plugx{\Xx_1}{\e_b\subst{\xs}{\vs}}}{\e_{12}}}}}.
          \]
          Since \plugCx\Cx{\evalt{\MT}{\plugx{\Xx_1}{\mcalld}}} steps similarly,
          \[
            \evalwa
            {\MT_g}{\plugCx\Cx{\evalt{\MT}{\plugx{\Xx_1}{\mcalld}}}}
            {\MT_g}{\plugCx\Cx{\evalt{\MT}{\plugx{\Xx_1}{\e_b\subst{\xs}{\vs}}}}},
          \]
          we can apply induction to the optimization of
          $\e_{11} = \plugx{\Xx_1}{\mcalld}$.
          The induction hypothesis says that $\exists \e'_{21}$, such that:
          \[
            \evalwa
              {\MT_g}{\plugCx\Cx{\evalt{\MT'}{\e'_{11}}}}
              {\MT_g}{\plugCx\Cx{\evalt{\MT'}{\e'_{21}}}}
            \quad\land\quad
            \expropte{\plugx{\Xx_1}{\e_b\subst{\xs}{\vs}}}{\e'_{21}}.
          \]
          Then, by~\lemref{lem:paper:simple-context-irrelevant}, we have:
          \[
            \evalwa
            {\MT_g}{\plugCx\Cx{\evalt{\MT'}{\seq{\e'_{11}}{\e'_{12}}}}}
            {\MT_g}{\plugCx\Cx{\evalt{\MT'}{\seq{\e'_{21}}{\e'_{12}}}}},
          \]
          and by \WAOE{Seq}, the desired optimization relation holds:
          \begin{mathpar}
            \inferrule*[right=\WAOE{Seq}]
            { \expropte{\plugx{\Xx_1}{\e_b\subst{\xs}{\vs}}}{\e'_{21}} \\
              \expropte{\e_{12}}{\e'_{12}} }
            { \expropte{\seq{\plugx{\Xx_1}{\e_b\subst{\xs}{\vs}}}{\e_{12}}}
                       {\seq{\e'_{21}}{\e'_{12}}} }.
          \end{mathpar}
        \item \WACAN{Redex} with $\e_{11} = \v_{11}$
          and \Can{\seq{\v_{11}}{\e_{12}}}{\rCC{\hole}}.
          Therefore, \plugCx\Cx{\evalt{\MT}{\seq{\e_{11}}{\e_{12}}}} steps
          by \WAE{Seq}:
          \[
            \evalwa
            {\MT_g}{\plugCx\Cx{\evalt{\MT}{\seq{\v_{11}}{\e_{12}}}}}
            {\MT_g}{\plugCx\Cx{\evalt{\MT}{\e_{12}}}}.
          \]
          Because \expropte{\v_{11}}{\e'_{11}} by assumption,
          we know that $\e'_{11} = \v_{11}$ by \lemref{lem:opt-preserves-values}.
          Therefore, the optimized expression steps by \WAE{Seq}:
          \[
            \evalwa
            {\MT_g}{\plugCx\Cx{\evalt{\MT'}{\seq{\v_{11}}{\e_{12}'}}}}
            {\MT_g}{\plugCx\Cx{\evalt{\MT'}{\e_{12}'}}}.
          \]
          And we already know that \expropte{\e_{12}}{\e'_{12}}.
        \item \WACAN{SeqC} with $\Cx' = \seq{\Cx_1}{\e_{12}}$ and
          \begin{mathpar}
            \inferrule*[right=\WACAN{SeqC}]
            { \Can{\e_{11}}{\rCC{\Cx_1}} }
            { \Can{\seq{\e_{11}}{\e_{12}}}{\rCC{\seq{\Cx_1}{\e_{12}}}} }.
          \end{mathpar}
          Thus, we know that
          $ \plugCx\Cx{\evalt{\MT}{\seq{\e_{11}}{\e_{12}}}} 
             = \plugCx\Cx{\evalt{\MT}{\seq{\plugCx{\Cx_1}{\rdx}}{\e_{12}}}}, $
          which means that the step
          \[\begin{array}{c}
            \evalwa{\MTg}{\plugCx\Cx{\evalt{\MT}{\seq{\plugCx{\Cx_1}{\rdx}}{\e_{12}}}}}
                    {\MTg'}{\plugCx\Cx{\evalt{\MT}{\seq{\plugCx{\Cx_1}{\e'}}{\e_{12}}}}}
            \\ \iff
            \\ \evalwa{\MTg}{\plugCx\Cx{\evalt{\MT}{\plugCx{\Cx_1}{\rdx}}}}
                      {\MTg'}{\plugCx\Cx{\evalt{\MT}{\plugCx{\Cx_1}{\e'}}}}.
          \end{array}\]
          The latter fact, which can also be written as
          \[
            \evalwa
            {\MTg}{\plugCx\Cx{\evalt{\MT}{\e_{11}}}}
            {\MTg'}{\plugCx\Cx{\evalt{\MT}{\plugCx{\Cx_1}{\e'}}}},
          \]
          allows us to apply induction to the optimization of
          $\e_{11} = \plugCx{\Cx_1}{\rdx}$
          (note that $\plugCx{\Cx_1}{\e'}$ plays the role of $\e_2$
          from the theorem statement).
          By the induction hypothesis, $\exists \e'_{21}$, such that:
          \[
            \evalwa
              {\MTg}{\plugCx\Cx{\evalt{\MT'}{\e'_{11}}}}
              {\MTg'}{\plugCx\Cx{\evalt{\MT'}{\e'_{21}}}}
            \quad\land\quad 
            \expropte{\plugCx{\Cx_1}{\e'}}{\e'_{21}}. \]
          But then, similarly to case (a), we get:
          \[
            \evalwa
            {\MTg}{\plugCx\Cx{\evalt{\MT'}{\seq{\e'_{11}}{\e'_{12}}}}}
            {\MTg'}{\plugCx\Cx{\evalt{\MT'}{\seq{\e'_{21}}{\e'_{12}}}}}
          \]
          and
          \begin{mathpar}
            \inferrule*[right=\WAOE{Seq}]
            { \expropte{\plugCx{\Cx_1}{\e'}}{\e'_{21}} \\
              \expropte{\e_{12}}{\e'_{12}} }
            { \expropte{\seq{\plugCx{\Cx_1}{\e'}}{\e_{12}}}
                       {\seq{\e'_{21}}{\e'_{12}}} }.
          \end{mathpar}
      \end{enumerate}
      The \emph{backward} direction (2) proceeds analogously.

    \item Case \WAOE{Primop} is similar to \WAOE{Seq}.
    
    \item Case \WAOE{Call} where $\e_1 = \mcall{\e_c}{\es}$
      and $\e'_1 = \mcall{\e'_c}{\es'}$:
      \begin{mathpar}
        \inferrule*[right=\WAOE{Call}]
        { \expropte{\e_c}{\e'_c} \\ \expropte{\e_i}{\e'_i} }
        { \expropte{\mcall{\e_c}{\es}}{\mcall{\e'_c}{\es'}} }.
      \end{mathpar}
      Let us consider the \emph{forward} direction (1) first.
      To make a step
      \[ \evalwa{\MT_g}{\plugCx\Cx{\evalt{\MT}{\mcall{\e_c}{\es}}}}
                {\MT_g'}{\plugCx\Cx{\evalt{\MT}{\_}}}, \]
      the source expression \mcall{\e_c}{\es} has to be either some
      \plugx{\Xx}{\mcalld} or \plugx{\Cx'}{\rdx}. By analyzing \CanSym forms,
      we know all such cases. Most of them are similar to the
      using-induction-hypothesis cases of \WAOE{Seq}:
      \WACAN{CalleeX} and \WACAN{CallX}---to \WACAN{SeqX},
      \WACAN{CalleeC} and \WACAN{CallC}---to \WACAN{SeqC}.
      The only interesting case is \WACAN{Call}, when
      $\e_1 = \mcall{\e_c}{\es} = \mcall{\m'}{\vs'}$.
      But then $\e'_1 = \mcall{\e'_c}{\es'} = \mcall{\m'}{\vs'}$.
      We know that \plugCx\Cx{\evalt{\MT}{\mcall{\m'}{\vs'}}} can step
      only by \WAE{CallLocal},
      \[
        \evalwa
          {\MTg}{\plugCx\Cx{\evalt{\MT}{\mcall{\m'}{\vs'}}}}
          {\MTg}{\plugCx\Cx{\evalt{\MT}{\e_b\subst{\xs}{\vs'}}}},
      \]
      and that \expropte{\m'}{\m'} by \WAOE{ValFun}. Thus, since
      $\m' \in \dom(\MT)$, we also know that $\m' \in \dom(\MT')$.
      Furthermore, by assumption, \tableopt{\SpecEnv}{\e_1}{\MT}{\MT'},
      so we know that the corresponding method definition of $\m'$
      exists in $\MT'$ (with the body $\e_b'$), and:
      \[ \expropt{\SpecEnv}{\obar{\x:\t}}
         {\MT}{\e_b}
         {\MT'}{\e_b'\subst{\xs'}{\xs}}. \]
      Since $\mcall{\m'}{\vs'}$ can be resolved in $\MT'$,
      we have a step similar to the above by rule \WAOE{CallLocal}:
      \[
        \evalwa
          {\MTg}{\plugCx\Cx{\evalt{\MT'}{\mcall{\m'}{\vs'}}}}
          {\MTg}{\plugCx\Cx{\evalt{\MT'}{\e'_b\subst{\xs'}{\vs'}}}}.
      \]
      By~\lemref{lem:subst-preserves-optim} (value substitution
      preserves optimization), we get the desired relation:
      \[ \expropte{\e_b\subst{\xs}{\vs'}}{\e'_b\subst{\xs'}{\vs'}}. \]
      The \emph{backward} direction (2) proceeds analogously.

    \item Case \WAOE{Inline} where $\e_1 = \mcalld$,
      $\e'_1 = (\seq{\skp}{\e'})$, and:
      \begin{mathpar}
        \inferrule
        { \typedwa{}{}{\v_i}{\g_i} \\
          \getmd(\MT,\m,\gs) = \mdef{\m}{\obar{\jty{\x}{\t}}}{\e_b} \\
          \expropte{\e_b\subst{\xs}{\vs}}{\e'} }
        { \expropte{\mcall{\m}{\vs}}{\seq{\skp}{\e'}} }.
      \end{mathpar}
      (because of the empty typing context, all near-values are just values).\\
      Let us start with the \emph{forward} direction (1) again.
      The only way for \plugCx\Cx{\evalt{\MT}{\mcalld}} to make a step is
      by using \WAE{CallLocal} rule (note that \typedwa{}{}{\v_i}{\g_i}
      means $\typeof(\v_i) = \g_i$):
      \[
        \evalwa
          {\MTg}{\plugCx\Cx{\evalt{\MT}{\mcalld}}}
          {\MTg}{\plugCx\Cx{\evalt{\MT}{\e_b\subst{\xs}{\vs}}}}.
      \]
      But \plugCx\Cx{\evalt{\MT'}{\seq{\skp}{\e'}}} can step by \WAE{Seq}:
      \[
        \evalwa
          {\MTg}{\plugCx\Cx{\evalt{\MT'}{\seq{\skp}{\e'}}}}
          {\MTg}{\plugCx\Cx{\evalt{\MT'}{\e'}}},
      \]
      and the desired optimization relation holds by assumption:
      \[ \expropte{\e_b\subst{\xs}{\vs}}{\e'}. \]
      The \emph{backward} direction (2) applies same the reasoning:
      \plugCx\Cx{\evalt{\MT'}{\seq{\skp}{\e'}}} steps only by \WAE{Seq},
      and \plugCx\Cx{\evalt{\MT}{\mcalld}} can step by \WAE{CallLocal}
      accordingly, so that \expropte{\e_b\subst{\xs}{\vs}}{\e'}.

    \item Case \WAOE{Direct} where $\e_1 = \mcall{\m}{\es}$,
      $\e'_1 = \mcall{\m'}{\es'}$, and:
      \begin{mathpar}
        \inferrule
        { \expropte{\e_i}{\e'_i} \\ \typedwa{}{}{\e'_i}{\g_i} \\
          (\mdeqd) \in \SpecEnv }
        { \expropte{\mcall{\m}{\es}}{\mcall{\m'}{\es'}} }.
      \end{mathpar}
      Let us consider the \emph{forward} direction (1) first.
      To make a step
      \[ \evalwa{\MT_g}{\plugCx\Cx{\evalt{\MT}{\mcall{\m}{\es}}}}
                {\MT_g'}{\plugCx\Cx{\evalt{\MT}{\_}}}, \]
      the source expression \mcall{\m}{\es} has to be either \mcalld, or some
      \plugx{\Xx}{\mcall{\m''}{\vs''}} or \plugx{\Cx'}{\rdx}.
      By analyzing \CanSym forms, we know all such cases:
      \begin{enumerate}
        \item \WACAN{Call} with $\mcall{\m}{\es} = \mcalld$ and
          \Can{\mcalld}{\rXX{\hole}}. Because
          \expropte{\v_i}{\e'_i} and \typedwa{}{}{\e'_i}{\g_i}, we know
          $\e'_i = \v_i$ and thus $\typeof(\v_i) = \g_i$.
          Because $(\mdeqd) \in \SpecEnv$, we know that call \mcalld can be
          resolved in \MT, and thus
          \plugCx\Cx{\evalt{\MT}{\mcalld}} steps by \WAE{CallLocal}:
          \[
            \evalwa
              {\MTg}{\plugCx\Cx{\evalt{\MT}{\mcalld}}}
              {\MTg}{\plugCx\Cx{\evalt{\MT}{\e_b\subst{\xs}{\vs}}}}.
          \]
          By other properties of well-defined $(\mdeqd) \in \SpecEnv$,
          we also know that call \mcall{\m'}{\vs} can be resolved in $\MT'$,
          which means
          \[
            \evalwa
              {\MTg}{\plugCx\Cx{\evalt{\MT'}{\mcall{\m'}{\vs}}}}
              {\MTg}{\plugCx\Cx{\evalt{\MT'}{\e'_b\subst{\xs'}{\vs}}}}.
          \]
          and that
          \[
            \expropt{\SpecEnv}{\obar{\x:\g}}
              {\MT}{\e_b}
              {\MT'}{\e'_b\subst{\xs'}{\xs}}.
          \]
          By~\lemref{lem:subst-preserves-optim} (value substitution
          preserves optimization), we get the desired relation:
          \[ \expropte{\e_b\subst{\xs}{\vs}}{\e'_b\subst{\xs'}{\vs}}. \]
        \item \WACAN{CallX} with
          $\Xx = \mcall{\m}{\vs_{..{i-1}},\,\Xx_i,\,\es_{{i+1}..}}$ and:
          \begin{mathpar}
            \inferrule*[right=\WACAN{CallX}]
            { \Can{\e_i}{\rX(\Xx_i, \mcall{\m''}{\vs''})} }
            { \Can{\mcall{\m}{\vs_{..{i-1}},\,\e_i,\,\es_{{i+1}..}}}
                  {\rX(\mcall{\m}{\vs_{..{i-1}},\,\Xx_i,\,\es_{{i+1}..}}, \mcall{\m''}{\vs''})} }.
          \end{mathpar}
          This case proceeds similarly to \WACAN{SeqX}, with reduction
          being driven by $\e_i$ instead of $\e_{11}$.
        \item \WACAN{CallC} with
        $\Cx' = \mcall{\m}{\vs_{..{i-1}},\,\Cx_i,\,\es_{{i+1}..}}$ and:
          \begin{mathpar}
            \inferrule*[right=\WACAN{CallC}]
            { \Can{\e_i}{\rCC{\Cx_i}} }
            { \Can{\mcall{\m}{\vs_{..{i-1}},\,\e_i,\,\es_{{i+1}..}}}
                  {\rCC{\mcall{\m}{\vs_{..{i-1}},\,\Cx_i,\,\es_{{i+1}..}}}}. }
          \end{mathpar}
          This case proceeds similarly to \WACAN{SeqC}, with reduction
          being driven by $\e_i$ instead of $\e_{11}$.
      \end{enumerate}
      The \emph{backward} direction (2) proceeds analogously.
  \end{itemize}
\end{proof}

The main \textbf{Optimization Correctness \thmref{thm:paper:table-opt}}
can be found on page~\pageref{thm:paper:table-opt}.
  %\clearpage
  \section{Optimization Algorithm}\label{app:opt-algo}
%%-----------------------------------------------------------------------------

This sections defines the expression optimization algorithm
that is implemented in the Redex model.
The algorithm is defined by the reflexive-transitive closure of the $\leadsto$
relation given in \figref{fig:wa-expr-opt-algo-context}.

The algorithm is based on the optimization relation \exproptd{\e}{\e'}:
we conjecture that the result of running the algorithm is in the optimization
relation with the original expression.

\newcommand{\maxinlinecount}{\ensuremath{\texttt{\emph{I}}}\xspace}
\newcommand{\maxspecializationcount}{\ensuremath{\texttt{\emph{S}}}\xspace}
\newcommand{\sigg}{\ensuremath{\texttt{\mcall{\m}{\obar{\g}}}}\xspace}
\newcommand{\sigt}{\ensuremath{\texttt{\mcall{\m}{\obar{\t}}}}\xspace}
\newcommand{\optrr}[7]{\Gm \ \vdash \ \langle \Omega, \ \Phi^\t, \ \SpecEnv, \ \evalt{#2}{#1}
    \rangle\ \leadsto\ \langle #3, \ #4, \ #5, \ \evalt{#7}{#6}\rangle}
\newcommand{\getInlineCount}{\ensuremath{\mathop{\mathbf{iCount}}}\xspace}
\newcommand{\incrementInlineCount}{\ensuremath{\mathop{\mathbf{iCount\nolinebreak\hspace{-.05em}\raisebox{.4ex}{\small\bf +}\nolinebreak\hspace{-.10em}\raisebox{.4ex}{\small\bf +}}}}\xspace}
\newcommand{\getSpecializationCount}{\ensuremath{\mathop{\mathbf{sCount}}}\xspace}

\begin{figure}
    \[
    \begin{array}{ccl@{\qquad}l}
    \Omega & ::= & & \text{\emph{Inline environment}} \\
             & \Alt & \varnothing & \text{empty environment}  \\
             & \Alt & \Omega, \mdeq{\obar{\t}}{\m}{\mathbb{N}} & \text{inline mapping extension} \\
    \\
    \maxinlinecount & & & \text{max inline count} \\
    \maxspecializationcount & & & \text{max specialization count} \\\\
    \Phi^\t & ::= & & \text{\emph{Direct-call environment}} \\
             & \Alt & \varnothing & \text{empty environment} \\
             & \Alt & \Phi^\t, \mdeq{\obar{\t}}{\m}{\m'} & \text{specialization mapping extension} \\
    \\ \OEx & ::= & & \text{\emph{Optimization context}}
    \\ &\Alt& \hole & \text{hole}
    \\ &\Alt& \seq{\OEx}{\e} \Alt \seq{\e}{\OEx} & \text{sequence}
    \\ &\Alt& \primcall{l}{\es\ \OEx\ \es'} & \text{primop call (argument)}
    \\ &\Alt& \mcall{\OEx}{\es} & \text{function call (callee)}
    \\ &\Alt& \mcall{\e}{\es\ \OEx\ \es'} & \text{function call (argument)}
    \end{array}
    \]
    \[
    \begin{array}{rcl}
      \getInlineCount(\Omega, \mcall{\m}{\obar{\t}}) & = &
        n \\
        & & \text{if } \mdeq{\obar{\t}}{\m}{n} \in \Omega \\
      \getInlineCount(\Omega, \mcall{\m}{\obar{\t}}) & = &
        0 \\
        & & \text{if } \nexists n : \mdeq{\obar{\t}}{\m}{n} \in \Omega \\
      \\
      \incrementInlineCount(\Omega, \mcall{\m}{\obar{\t}}) & = &
        (\varnothing, \mdeq{\obar{\t}}{\m}{n+1}) \cup \Omega_{rest} \\
        & & \text{if } (\varnothing, \mdeq{\obar{\t}}{\m}{n}) \cup \Omega_{rest} == \Omega\ \&\& \\
        & & (\varnothing, \mdeq{\obar{\t}}{\m}{n}) \cap \Omega_{rest} == \varnothing \\
      \incrementInlineCount(\Omega, \mcall{\m}{\obar{\t}}) & = &
        \Omega, \mdeq{\obar{\t}}{\m}{1} \\
        & & \text{if } \nexists n : \mdeq{\obar{\t}}{\m}{n} \in \Omega \\
      \\
      \getSpecializationCount(\varnothing,\ \m) & = &
        0 \\
      \getSpecializationCount((\SpecEnv_{rest}, \mdeqd),\ \m_{cmp}) & = &
        \getSpecializationCount(\SpecEnv_{rest},\ \m_{cmp}) \\
        & & \text{if } \m_{cmp} \neq \m \\
      \getSpecializationCount((\SpecEnv_{rest}, \mdeqd),\ \m_{cmp}) & = &
        1 + \getSpecializationCount(\SpecEnv_{rest},\ \m_{cmp}) \\
        & & \text{if } \m_{cmp} == \m \\
    \end{array}
    \]
  \caption{Expression optimization algorithm: syntax \& auxiliary functions}
\end{figure}

\begin{figure}
  \begin{mathpar}
        \inferrule[\WAOE{Inline}]
            { \typedwad{\nv_i}{\g_i} \\
              \getmd(\MT,\m,\gs) = \mdef{\m}{\obar{\jty{\x}{\t}}}{\e_b} \\
              \getInlineCount(\Omega,\ \mcall{\m}{\obar{\t}})\ < \ \maxinlinecount \\
              \incrementInlineCount(\Omega,\ \mcall{\m}{\obar{\t}}) = \Omega' }
            { \optrr{\plugx{\OEx}{\mcall{\m}{\nvs}}}{\MT}
            {\Omega'}{\Phi^\t}{\SpecEnv}
            {\plugx{\OEx}{\seq{\skp}{\e_b\subst{\xs}{\nvs}}}}{\MT} }

        \inferrule[\WAOE{Specialize-Existing}]
            { \exists\ \e_i \in \obar{\e} : \typeof(\e_i) \notin \g \\
              \nexists\ \obar{\g''},\m'' : \mdeq{\obar{\g''}}{\m''}{\m} \in \SpecEnv \\
              \typedwad{\e_i}{\g_i} \\
              \mdeq{\obar{\g}}{\m}{\m'} \in \SpecEnv }
            { \optrr{\plugx{\OEx}{\mcall{\m}{\obar{\e}}}}{\MT}
            {\Omega}{\Phi^\t}{\SpecEnv}
            {\plugx{\OEx}{\mcall{\m'}{\obar{\e}}}}{\MT} }

        \inferrule[\WAOE{Specialize-New}]
            { \exists\ \e_i \in \obar{\e} : \typeof(\e_i) \notin \g \\
              \nexists\ \obar{\g''},\m'' : \mdeq{\obar{\g''}}{\m''}{\m} \in \SpecEnv \\
              \typedwad{\e_i}{\g_i} \\
              \getmd(\MT,\m,\gs) = \mdef{\m}{\obar{\jty{\x}{\t}}}{\e_b} \\
              \nexists\ \m_{opt} : \mdeq{\obar{\g}}{\m}{\m_{opt}} \in \SpecEnv \\
              \getSpecializationCount(\SpecEnv,\ \m)\ < \ \maxspecializationcount \\
              \m' \ \text{ does not occur in } \MT \\
              \MT' = \MText{\mdef{\m'}{\obar{\jty{\x}{\g}}}{\e_b}}{\MT} \\
              \SpecEnv' = \SpecEnv, \mdeq{\obar{\g}}{\m}{\m'} }
            { \optrr{\plugx{\OEx}{\mcall{\m}{\obar{\e}}}}{\MT}
            {\Omega}{\Phi^\t}{\SpecEnv'}
            {\plugx{\OEx}{\mcall{\m'}{\obar{\e}}}}{\MT'} }

        \inferrule[\WAOE{Direct-Call-Existing}]
            { \exists\ \e_i \in \obar{\e},\ \typeof(\e_i) \notin \g \\
              \nexists\ \obar{\g''},\m'' : \mdeq{\obar{\g''}}{\m''}{\m} \in \SpecEnv \\
              \typedwad{\e_i}{\g_i} \\
              \getmd(\MT,\m,\gs) = \mdef{\m}{\obar{\jty{\x}{\t}}}{\e_b} \\
              \nexists\ \m_{opt},\ \mdeq{\obar{\g}}{\m}{\m_{opt}} \in \SpecEnv \\
              \getSpecializationCount(\SpecEnv,\ \m)\ \geq \ \maxspecializationcount \\
              \mdeq{\obar{\t}}{\m}{\m'} \in \Phi^\t \\
              \SpecEnv' = \SpecEnv, \mdeq{\obar{\g}}{\m}{\m'} }
            { \optrr{\plugx{\OEx}{\mcall{\m}{\obar{\e}}}}{\MT}
            {\Omega}{\Phi^\t}{\SpecEnv'}
            {\plugx{\OEx}{\mcall{\m'}{\obar{\e}}}}{\MT} }

        \inferrule[\WAOE{Direct-Call-New}]
            { \exists\ \e_i \in \obar{\e} : \typeof(\e_i) \notin \g \\
              \nexists\ \obar{\g''},\m'' : \mdeq{\obar{\g''}}{\m''}{\m} \in \SpecEnv \\
              \typedwad{\e_i}{\g_i} \\
              \getmd(\MT,\m,\gs) = \mdef{\m}{\obar{\jty{\x}{\t}}}{\e_b} \\
              \nexists\ \m_{opt} : \mdeq{\obar{\g}}{\m}{\m_{opt}} \in \SpecEnv \\
              \getSpecializationCount(\SpecEnv,\ \m)\ \geq \ \maxspecializationcount \\
              \m' \ \text{ does not occur in } \MT \\
              \MT' = \MText{\mdef{\m'}{\obar{\jty{\x}{\t}}}{\e_b}}{\MT} \\
              \SpecEnv' = \SpecEnv, \mdeq{\obar{\g}}{\m}{\m'} \\
              \Phi^{\t'} = \Phi^\t, \mdeq{\obar{\t}}{\m}{\m'} }
            { \optrr{\plugx{\OEx}{\mcall{\m}{\obar{\e}}}}{\MT}
            {\Omega}{\Phi^{\t'}}{\SpecEnv'}
            {\plugx{\OEx}{\mcall{\m'}{\obar{\e}}}}{\MT'} }

  \end{mathpar}
  \caption{Expression optimization algorithm: step}\label{fig:wa-expr-opt-algo-context}
\end{figure}

}

%% =============================================================================
\APPENDICITIS
{\bibliography{bib/jv,bib/all}}
{\clearpage\bibliography{main}}

\end{document}